\documentclass[sigconf]{acmart}

\usepackage[linesnumbered, ruled]{algorithm2e}
\usepackage{graphicx}
\usepackage{amsmath}
\usepackage{float}
\usepackage{xspace}
\usepackage{hyperref}
\usepackage[noabbrev, capitalize, nameinlink]{cleveref}
\usepackage{listings}
\usepackage{fancyhdr}
\usepackage{multirow}
\usepackage{subcaption}
\usepackage[many]{tcolorbox}
\usepackage{appendix}

\newcommand{\Semijoin}[0]{\ltimes}
\newcommand{\TreeStruct}{\texttt{LargestRoot}\xspace}
\newcommand{\SafeSubJoin}{\texttt{SafeSubjoin}\xspace}
\newcommand{\StoL}{\texttt{Small2Large}\xspace}

\newcommand{\UseBF}{\texttt{ProbeBF}\xspace}
\newcommand{\LogicalUseBF}{\texttt{LogicalProbeBF}\xspace}

\newcommand{\CreateBF}{\texttt{CreateBF}\xspace}
\newcommand{\LogicalCreateBF}{\texttt{LogicalCreateBF}\xspace}

\newcommand{\robustmetric}{RF\xspace}

\newcommand{\NoPT}{w/o RPT \xspace}
\newcommand{\Yann}{Yannakakis\xspace}
\newcommand{\YannAlg}{Yannakakis algorithm\xspace}
\newcommand{\PT}{Predicate Transfer\xspace}
\newcommand{\pt}{PT\xspace}

\newcommand{\BF}{Bloom filter\xspace}
\newcommand{\BFs}{Bloom filters\xspace}
\newcommand{\RPT}{Robust Predicate Transfer\xspace}

\newcommand{\rpt}{RPT\xspace}

\newcommand{\duckdb}{DuckDB\xspace}
\newcommand{\tpch}{TPC-H\xspace}
\newcommand{\job}{JOB\xspace}
\newcommand{\tpcds}{TPC-DS\xspace}
\newcommand{\dsb}{DSB\xspace}

\newcommand{\rleftarrows}{\mathrel{\raise.75ex\hbox{\oalign{%
  $\hfil\scriptstyle\relbar$\cr
  \vrule width0pt height.5ex$\scriptstyle\smash\leftarrow$\cr}}}}
\newcommand{\rightlarrows}{\mathrel{\raise.75ex\hbox{\oalign{%
  $\scriptstyle\rightarrow$\hfil\cr
  $\scriptstyle\vrule width0pt height.5ex\relbar$\cr}}}}
\newcommand{\leftrarrows}{\mathrel{\raise.75ex\hbox{\oalign{%
  $\scriptstyle\leftarrow$\cr
  \vrule width0pt height.5ex$\hfil\scriptstyle\relbar$\cr}}}}
\newcommand{\lrightarrows}{\mathrel{\raise.75ex\hbox{\oalign{%
  $\scriptstyle\relbar$\hfil\cr
  $\scriptstyle\vrule width0pt height.5ex\smash\rightarrow$\cr}}}}
\newcommand{\Rrelbar}{\mathrel{\raise.75ex\hbox{\oalign{%
  $\scriptstyle\relbar$\cr
  \vrule width0pt height.5ex$\scriptstyle\relbar$}}}}

\makeatletter
\def\leftrightarrowsfill@{\arrowfill@\leftrarrows\Rrelbar\lrightarrows}
\def\rightleftarrowsfill@{\arrowfill@\rleftarrows\Rrelbar\rightlarrows}
\newcommand{\xleftrightarrows}[2][]{\ext@arrow 3399\leftrightarrowsfill@{#1}{#2}}
\newcommand{\xrightleftarrows}[2][]{\ext@arrow 3399\rightleftarrowsfill@{#1}{#2}}
\makeatother

\begin{document}

\title{Debunking the Myth of Join Ordering: Toward Robust SQL Analytics}

\author{Junyi Zhao}
\affiliation{
  \institution{Tsinghua University}
  \city{Beijing}
  \country{China}}
\email{zhaojy20@mails.tsiinghua.edu.cn}

\author{Kai Su}
\affiliation{
  \institution{Tsinghua University}
  \city{Beijing}
  \country{China}}
\email{suk23@mails.tsinghua.edu.cn}

\author{Yifei Yang}
\affiliation{
  \institution{University of Wisconsin-Madison}
  \city{Madison}
  \country{USA}}
\email{yyang673@wisc.edu}

\author{Xiangyao Yu}
\affiliation{
  \institution{University of Wisconsin-Madison}
  \city{Madison}
  \country{USA}}
\email{yxy@cs.wisc.edu}

\author{Paraschos Koutris}
\affiliation{
  \institution{University of Wisconsin-Madison}
  \city{Madison}
  \country{USA}}
\email{paris@cs.wisc.edu}

\author{Huanchen Zhang}
\authornote{Huanchen Zhang is also affiliated with the Shanghai Qi Zhi Institute. Corresponding author}
\affiliation{
  \institution{Tsinghua University}
  \city{Beijing}
  \country{China}}
\email{huanchen@tsinghua.edu.cn}

\acmConference[SIGMOD 25']{International Conference on Management of Data}{June 22-27}{Berlin, Germany}

\begin{abstract}
    Join order optimization is critical in achieving good query performance. Despite decades of research and practice, modern query optimizers could still generate inferior join plans that are orders of magnitude slower than optimal. Existing research on robust query processing often lacks theoretical guarantees on join-order robustness while sacrificing query performance. In this paper, we rediscover the recent \PT technique from a robustness point of view. We introduce two new algorithms, \TreeStruct and \SafeSubJoin, and then propose \RPT (\rpt) that is provably robust against arbitrary join orders of an acyclic query. We integrated \RPT with \duckdb, a state-of-the-art analytical database, and evaluated against all the queries in \tpch, \job, and \tpcds benchmarks. Our experimental results show that \rpt improves join-order robustness by orders of magnitude compared to the baseline. With \rpt, the largest ratio between the maximum and minimum execution time out of random join orders for a single acyclic query is only $1.6\times$ (the ratio is close to 1 for most evaluated queries). Meanwhile, applying \rpt also improves the end-to-end query performance by $\approx$$1.5\times$ (per-query geometric mean). We hope that this work sheds light on solving the practical join ordering problem.

\end{abstract}

\maketitle

\section{Introduction}
\label{sec:intro}

A query optimizer is a critical and perhaps most difficult component to develop in a relational database management system (RDBMS). Despite decades of research and practice, modern query optimizers are still far from reliable~\cite{leis2015HowGood}. Among the many challenges, join ordering is the crown jewel of query optimization. Determining an optimal join order requires not only an efficient algorithm to search the enormous plan space but also an accurate cardinality estimation of the intermediate results. The latter is extremely difficult despite recent efforts to bring machine learning techniques to the problem~\cite{2020are_we_ready, lehmann2023IsLearned}. The reality is that the optimizers today constantly generate plans that are orders of magnitude slower than optimal~\cite{2014optimizationsolved?, 2019tutorial_robust, 2021survey_optimizer}.

Prior research on robust query processing typically approaches the problem in two ways. The first is to prefer plans with more stable performance against cardinality estimation uncertainties during query optimization~\cite{2002LEC, 2005RCE, 2007plan_diagram, 2008strict_plan_diagram}. Such a ``conservative'' plan, however, often sacrifices query performance, and there is no theoretical guarantee of the plan's robustness. Another approach (i.e., re-optimization) is to collect the true cardinalities of intermediate results and reinvoke the optimizer at query execution time to generate better (remaining) plans~\cite{1998reopt, 2000eddies, 2004pop, Perron19, 2023reopt_zhao, justen2024polar}. Nonetheless, the overhead of materializing the intermediate results at pre-defined re-optimization points often offsets the benefit of switching to a more efficient plan.

Fortunately, the seminal \YannAlg offers encouraging theoretical results~\cite{yannakakis1981YA}. The algorithm guarantees a complexity linear to the input + output size for any acyclic query regardless of its join order. The key idea is to perform a full semi-join reduction on the input relations (i.e., the semi-join phase) before joining them (i.e., the join phase) so that the remaining tuples must appear in the query's final output. Despite the appealing theoretical guarantee, \YannAlg received few adoptions because of the costly semi-join operation.

The recent \PT (PT) algorithm proposes to speed up the above semi-joins by building \BFs instead of full hash tables~\cite{yang2023PT}. The original paper focused on the impressive performance advantages of the technique with an order-of-magnitude improvement over the default query plans on a prototype system. Although \PT was inspired by the \YannAlg, it fails to inherit the strong theoretical guarantee for acyclic queries because the algorithm does not ensure a full reduction of the input relations.

In this paper, we rediscover \PT from a robustness point of view. We propose \textbf{\RPT} (\rpt) with two new algorithms on top of the original PT to guarantee join-order robustness. We first introduce the \TreeStruct algorithm, which finds a \emph{join tree} of an acyclic query by constructing a maximum spanning tree on its weighted join graph, to warrant a full reduction in the semi-join phase (aka transfer phase in \rpt). To guarantee the robustness in the join phase of \rpt, we propose the \SafeSubJoin algorithm to verify the ``safety'' of a join order (i.e., its runtime cost is at most a constant factor away from the optimal) if the query is not $\gamma$-acyclic.

We implemented the \RPT algorithm in \duckdb, a state-of-the-art in-process analytical database management system. The modifications to \duckdb were non-invasive: we introduced two new operators for building and probing \BFs and inserted an \rpt optimization step/submodule into the optimizer's workflow. Our evaluation includes the three most widely used benchmarks for analytical workloads: \tpch~\cite{TPCH}, \job~\cite{JOB}, and \tpcds~\cite{TPCDS}. We measure the \emph{join-order robustness} of a query using the performance gap between executing different random join orders. The smaller the gap, the more robust the query.

The experimental results are promising. Compared to the baseline (i.e., \duckdb without \rpt integrated), \rpt improves the robustness factor (i.e., ratio between the maximum and minimum execution time out of 200 random join orders) by orders of magnitude for acyclic queries (which accounts for $94\%$ of the queries in the benchmarks). \rpt allows most queries to have a robustness factor close to 1, and the largest performance gap between the best and worst join orders is only $2.8\times$ among all the acyclic queries in \tpch, \job, and \tpcds. We then zoomed in and verified the robustness of the \TreeStruct algorithm. Furthermore, applying \rpt improves the end-to-end execution time per query by $\approx$$1.5\times$ (geometric mean) over the baseline. We also concluded that it is not worthwhile to consider bushy plans for \rpt because they brought little performance gain compared to left-deep in our evaluation.

The implications of our results could impact the design of future query engines and optimizers. With \RPT, join order optimization is no longer a critical challenge for acyclic queries because of \rpt's strong theoretical guarantee and practical efficiency. Future optimizers could limit their search space to left-deep plans (or simply pick a random join order) and become much more tolerant against cardinality estimation errors. Despite our promising results in achieving practical join order robustness, whether an instance-optimal join algorithm exists for cyclic queries remains an open problem.

We make three primary contributions in this paper. First, we propose two new algorithms (with rigorous proofs) to make \PT robust against arbitrary join orders. Second, we show that our \RPT algorithm is easy to integrate by implementing it in \duckdb, a state-of-the-art analytical system. Finally, we discover through experiments that \rpt exhibits outstanding robustness while improving the overall query performance at the same time, a big step toward solving the practical join ordering problem.

\section{Preliminaries}
\label{sec:prelim}

In this section, we first discuss the challenges and prior efforts in solving the join ordering problem. We then describe the \YannAlg and \PT in detail.


\begin{figure*}[t!]
  \begin{subfigure}[t]{.25\linewidth}%
    \center
    \includegraphics[width=\linewidth]{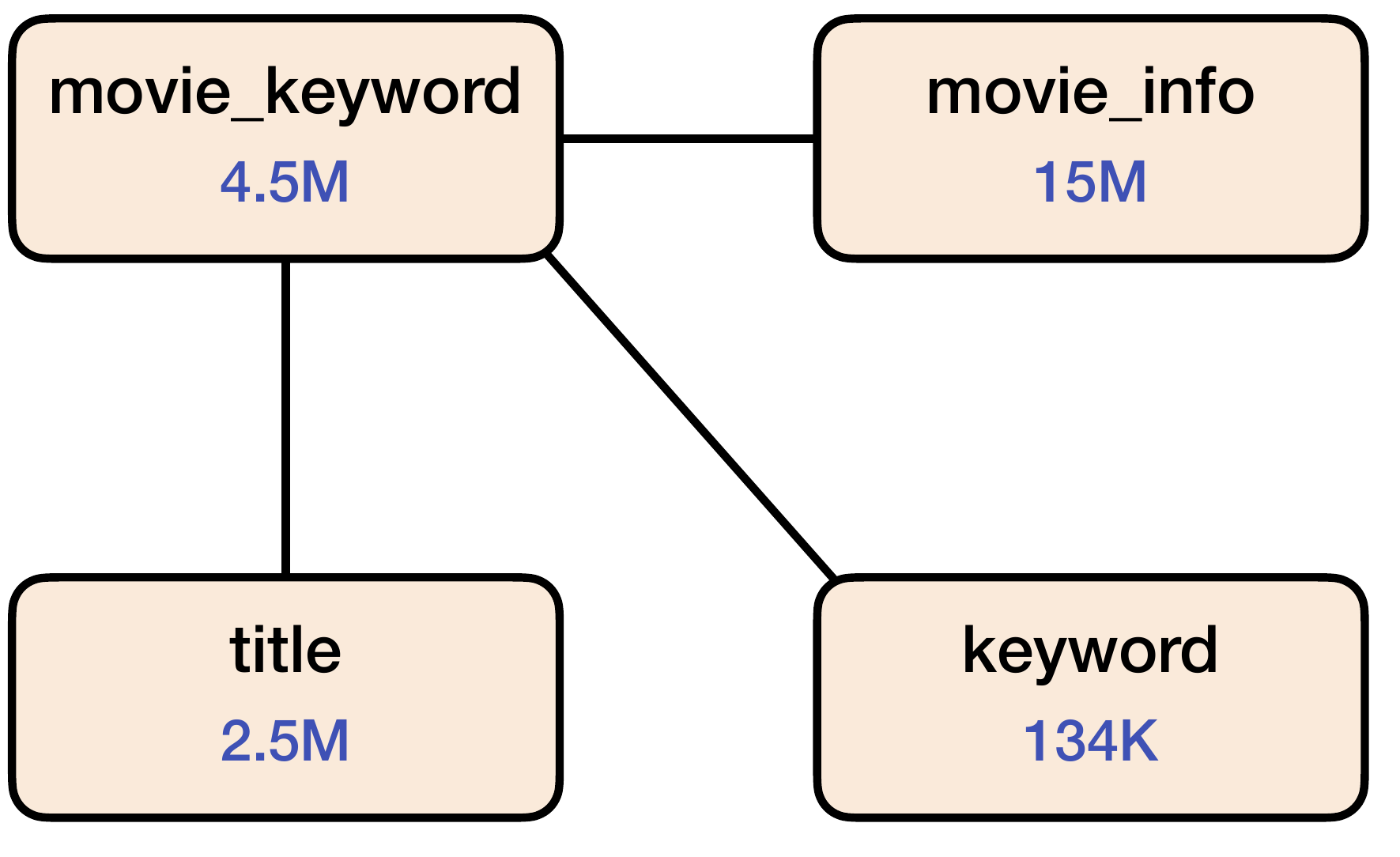}
    \caption{Join Graph}
    \label{fig:yannakakis-join-graph}
  \end{subfigure}
  \hfill
  \begin{subfigure}[t]{.42\linewidth}%
    \center
    \includegraphics[width=\linewidth]{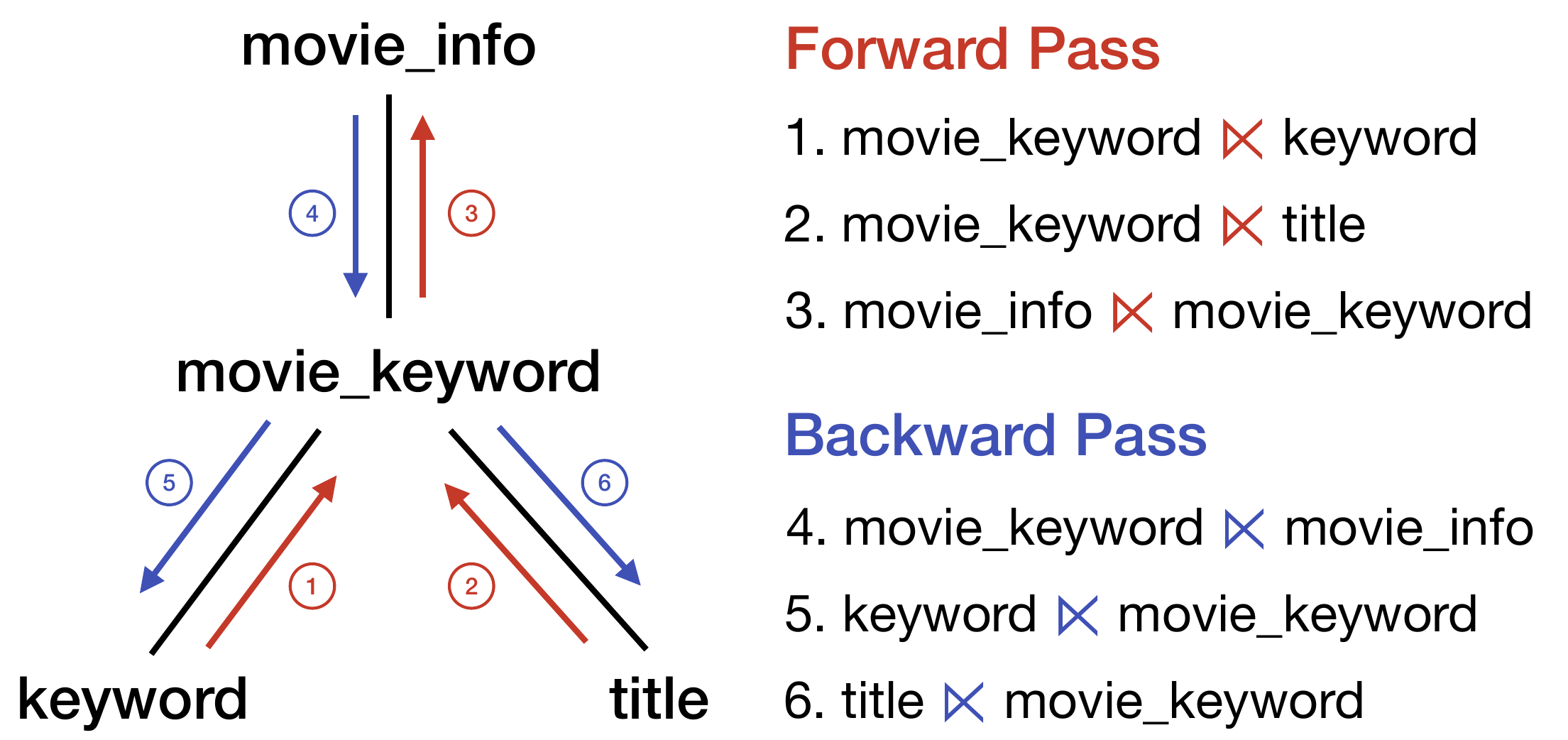}
    \caption{Semi-Join Phase}
    \label{fig:yannakakis-semi-join}
  \end{subfigure}
  \hfill
  \begin{subfigure}[t]{.28\linewidth}%
    \center
    \includegraphics[width=\linewidth]{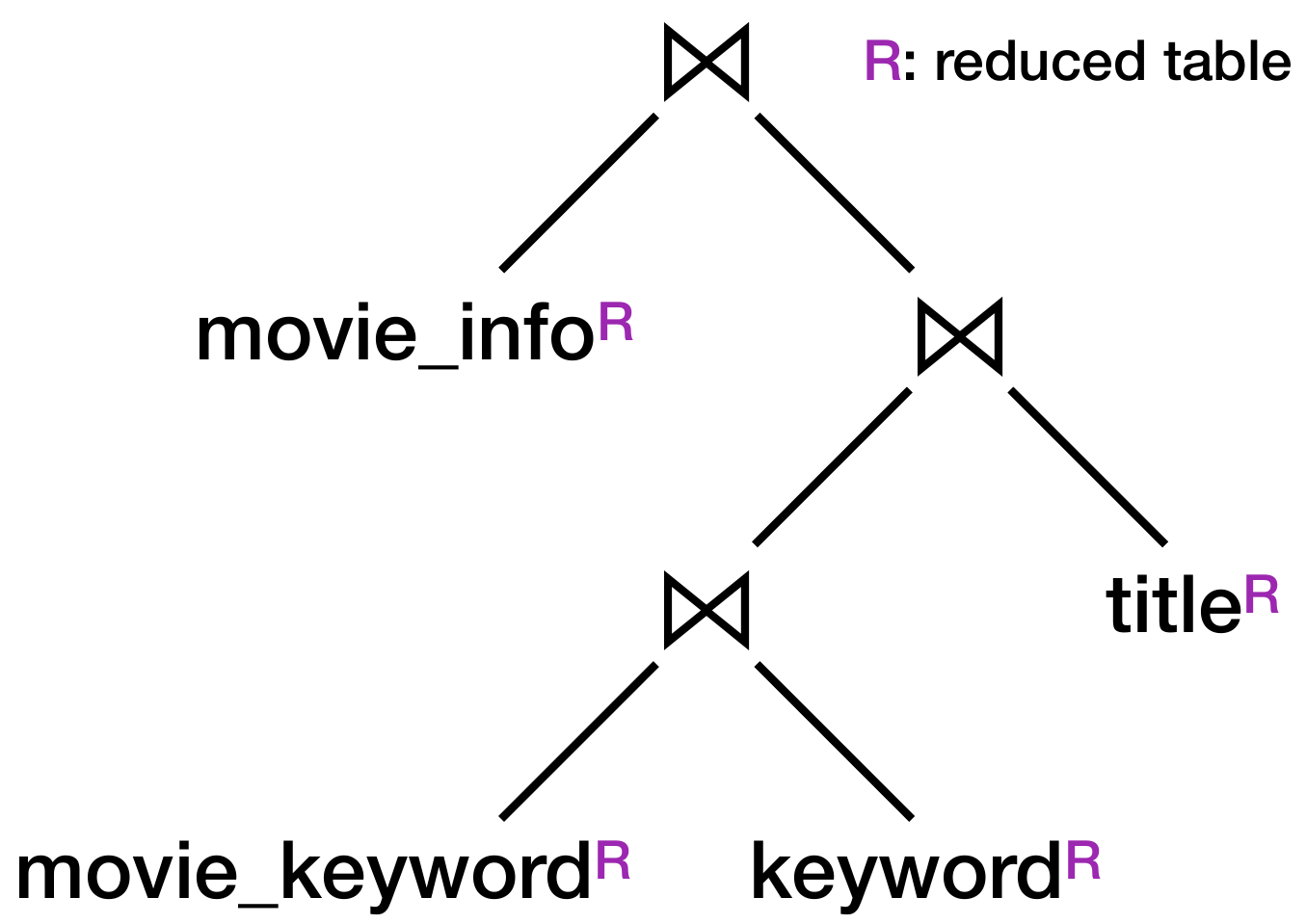}
    \caption{Join Phase}
    \label{fig:yannakakis-join}
  \end{subfigure}
  \caption{
    \textbf{\YannAlg on JOB 3a.}
  }
  \label{fig:yannakakis}
\end{figure*}

\subsection{Join Order Optimization}
\label{sec:prelim:jo}

Optimizing join orders is one of the most important tasks in query optimization. A bad join order leads to large intermediate results and can be orders of magnitude slower than the optimal plan~\cite{selinger1979,2019tutorial_robust,leis2015HowGood}. Obtaining an optimal join order in a modern query optimizer requires accurate cardinality estimation and efficient plan enumeration. Both remain difficult after over 40 years of research.

Cardinality estimation (CE) predicts the number of tuples produced by each operator in a query plan. Obtaining accurate estimations of the join cardinalities is extremely difficult. Without detailed statistics, the query optimizer typically makes the following assumptions: (1) Uniformity: values in a column are uniformly distributed within the global min/max; (2) Independence: values in different columns are uncorrelated; (3) Inclusion: every value from the probe side of the join must appear in the build side. These assumptions are rarely valid in real-world applications. Although having further statistics (e.g., histograms on joint distributions) can improve the accuracy of join cardinality estimation, it is prohibitively expensive to maintain comprehensive cross-column statistics. Even worse, studies show that a small estimation error will propagate exponentially with respect to the number of joins~\cite{1991errorpropagation,robustoptimization}. Leis et al.~\cite{leis2015HowGood} reported that none of the optimizers in real-world DBMSs (including the commercial ones) can estimate join cardinalities accurately: most of them under-estimate by 2-4 orders of magnitude when the number of joins $\ge 5$. Recent proposals tackle the CE problem using machine learning and deep learning techniques~\cite{2019MSCN, 2019Naru, 2019light_ML, halford2019bayesian, 2020deepdb, zhu2020flat, 2020deep, park2020quicksel, shetiya2020astrid, yang2020neurocard, liu2021fauce, wu2021unified}, but none of them so far has shown evidence of robust estimation.

Plan enumeration refers to the process of searching equivalent join orders and finding the query plan that has the minimal cost. Plan enumeration has been proven to be NP-hard~\cite{1984NP}. Prior work developed efficient algorithms based on dynamic programming~\cite{moerkotte2006DPccp, 2008dphyper}, and they are sufficient when the number of joins is small (e.g., < 10). However, because the search space grows exponentially with respect to the number of joins, it becomes impractical to perform an exhaustive search for a query with a large number of joins. Optimizers must fall back to heuristics-based approaches (e.g., the genetic algorithm in PostgreSQL~\cite{Postgres} and the greedy algorithm in DuckDB~\cite{duckdb}), sacrificing plan optimality for a reasonable optimization complexity.

A query execution is \textbf{robust} if its performance is never too far from the optimal even if the cardinality estimations are way off (which is inevitable)~\cite{2019tutorial_robust, robustoptimization}. Under the context of join ordering, it means that the risk of choosing a catastrophic join order due to wild CE errors is low. There are typically two ways to improve the robustness of SQL execution. The first is to favor a ``\textbf{robust plan}'' rather than the cost-optimal one to take into account the uncertainty during query optimization~\cite{2002LEC, 2005RCE, 2007plan_diagram, 2008strict_plan_diagram}. For example, the optimizer would estimate cardinalities using intervals (rather than single values)~\cite{proactive} or probability distributions~\cite{2005RCE} and choose plans that have stable costs within certain confidence intervals. However, such a robust plan may not exist, and the plan chosen often exhibits a noticeable performance hit compared to the optimal~\cite{robustoptimization}.

The second approach to improving plan robustness is through \textbf{re-optimization}~\cite{1998reopt, 1999reopt_shared_nothing, 2000eddies, 2004pop, 2007pop_parallel, 2016planbouquets, Perron19, 2023reopt_zhao, justen2024polar}.
The main idea is to correct CE mistakes while executing the query. Re-optimization must define specific materialization points in the query plan (usually at pipeline breakers) and collect statistics to obtain the true cardinalities at those points. If there is a large gap between the true cardinality and the estimated one, the system will re-invoke the optimizer, hoping to generate a better plan for the remaining operations. Although re-optimization enables self-correction at run time, materializing intermediate results is often costly and might compromise the end-to-end query performance.

\subsection{\YannAlg \& \PT}
\label{sec:prelim:pt}

An alternative thought of approaching join-order robustness is to design a join algorithm with bounded intermediate result sizes. Given a join query $Q$, let $N$ be the total number of tuples in all the input relations, and let $OUT$ be the number of tuples in the query output. The classic \textbf{\YannAlg}~\cite{yannakakis1981YA} guarantees a query complexity of $O(N + OUT)$\footnote{Considering the query size to be constant.}, which is the same as simply scanning the input and writing the output. Therefore, \YannAlg is instance optimal. The key idea is to pre-filter the tuples in the input relations that will not appear in the final output. The pre-filtering is realized via a series of semi-join reductions. A semi-join  $R \ensuremath{\ltimes} S$ outputs tuples from the left relation that have a match in the right relation. More formally, $R \ensuremath{\ltimes} S = \pi_{\text{attr(}R\text{)}}(R \bowtie S)$. In other words, a semi-join uses the right table as a filter to eliminate unmatched tuples in the left table.

Given a \emph{join graph} of a query where each vertex is a table scan, and each edge represents an equi-join (e.g., \cref{fig:yannakakis-join-graph}), the \YannAlg first picks an arbitrary vertex as root and obtains a \emph{join tree} (e.g., \cref{fig:yannakakis-semi-join}) via the GYO ear removal algorithm \cite{gyo}. The algorithm requires that the join graph is \emph{acyclic}  ($\alpha$-acyclic to be precise~\cite{yannakakis1981YA}) so that a join tree always exists. \YannAlg then proceeds to the \textbf{semi-join phase} \cite{usingsemi}, consisting of a \emph{forward pass} and a \emph{backward pass}. In the forward pass, the algorithm traverses the join tree from leaf to root (e.g., post-order traversal). For each node $R$, suppose its children are $S_1, S_2, \cdots, S_n$. The algorithm performs semi-join reduction on $R$ using all of its children (i.e., \texttt{for} $i = 1, 2, \cdots, n: R \ensuremath{\ltimes} S_i$). An example is shown in \cref{fig:yannakakis-semi-join}. Once the forward pass reaches the root, the algorithm starts the backward pass from root to leaf (e.g., level-order traversal). For each node $R$ with its parent $P$, $R \ensuremath{\ltimes} P$ is performed. The backward pass ends when all the leaf nodes are visited.

After this, the \YannAlg enters the \textbf{join phase}, where normal binary joins (e.g., hash joins) are carried out on the reduced tables, as shown in \cref{fig:yannakakis-join}. Each binary join must map to an edge in the join tree from the semi-join phase to guarantee a non-decreasing intermediate result. Because the semi-join phase ensures that \emph{all} tuples that will not contribute to the query output are removed (i.e., a full reduction), the join phase is proven to complete in $O(OUT)$ time.

Although \YannAlg exhibits appealing theoretical guarantees, few modern database management systems adopt it because the traditional hash-table-based implementation of the semi-joins makes the algorithm slow. The recent \textbf{\PT} (PT) technique proposed by Yang et al.~\cite{yang2023PT} solves this performance problem by using \BFs to conduct approximate semi-joins in the \YannAlg. Specifically, for each $R \ensuremath{\ltimes} S$ in the semi-join phase (it is called the \PT phase in PT), PT builds a \BF $\mathcal{B}_S$ with the join keys in $S$ and then uses the tuples in $R$ to probe $\mathcal{B}_S$. If the probe returns false for a tuple $t$ in $R$, $t$ is eliminated. Otherwise, $t$ is inserted into a different \BF $\mathcal{B}_R$ (could use a different join key) to prepare for the next semi-join in either the forward or backward pass.

Compared to the original \YannAlg, \PT trades a \emph{small} accuracy loss (caused by false positives of the \BFs) for a faster semi-join reduction. An inaccurate pre-filtering result does not affect the algorithm's correctness because the false positives will be removed during the subsequent join phase. Besides performance improvement, \PT generalizes \YannAlg to arbitrary join graphs, including cyclic ones. Instead of converting an acyclic join graph into a join tree, \PT transforms any join graph into a DAG (i.e., a \emph{transfer graph}) using a simple heuristic that assigns the direction of each edge from the smaller table to the larger one. Unfortunately, \PT does not inherit the strong theoretical guarantee from \YannAlg for acyclic queries because it could generate transfer schedules that lead to incomplete semi-join reductions. We will propose a new algorithm to fix this in the next section. For cyclic joins, although \PT improved the query performance empirically in many cases, there is no theoretical guarantee on the intermediate result sizes.

\section{Toward Join-Order Robustness}
\label{sec:modeling}

This section introduces new algorithms with analyses to make \PT robust for acyclic queries. In \cref{sec:modeling:transfer}, we propose the \TreeStruct algorithm in the transfer phase (i.e., the counterpart of the semi-join phase in \Yann) that not only guarantees a full reduction but also minimizes the \BF construction time. \cref{sec:modeling:join} discusses approaches to guarantee that the join order selected in the join phase is ``safe'' (i.e., there is no intermediate result blowup).

\subsection{Generating a Robust Transfer Schedule}
\label{sec:modeling:transfer}

The transfer phase in the original \PT algorithm~\cite{yang2023PT} adopts \StoL, a simple heuristic-based algorithm to build the transfer graph. As described in \cref{sec:prelim:pt}, \StoL assigns the direction for each edge in the (undirected) join graph from the smaller table to the larger table to form a DAG. \PT then generates a \emph{transfer schedule} (i.e., the forward and backward passes of \BFs) by following the edges in this DAG. The \StoL algorithm, however, does not guarantee a full reduction for acyclic queries. As shown in \cref{fig:small2large-example}, for example, consider the natural join $R(A,B) \Join S(A,C) \Join T(B, D)$ where $|R| < |S| < |T|$. In this case, \StoL will generate a transfer graph that leads to a forward pass of $S \Semijoin_b R$ and $T \Semijoin_b R$ followed by a backward pass of $R \Semijoin_b S$ and $R \Semijoin_b T$. This transfer schedule fails to ``connect'' $S$ and $T$: if $S$ has a predicate, this filter information can never reach $T$ via the transfer of \BFs (and vice versa), leading to an incomplete reduction.

\begin{figure}[t!]
    \center
    \includegraphics[width=\linewidth]{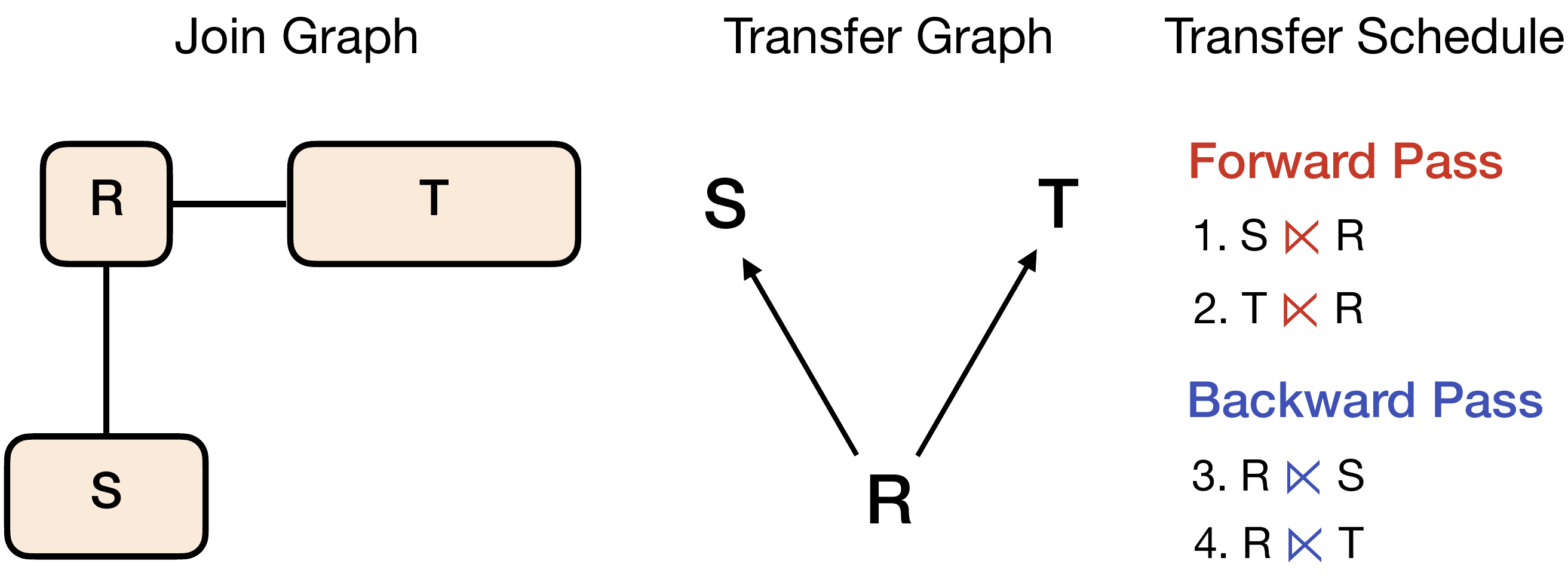}
    \caption{An example of the \StoL algorithm in the original \PT}
    \label{fig:small2large-example}
\end{figure}

Although \StoL cannot pre-filter all the non-result tuples, pushing larger tables toward the end of the transfer schedule is insightful because a smaller table is likely a more selective filter. The new \TreeStruct will preserve this strategy while guaranteeing a full reduction. Before diving into our new algorithm, let us define the concepts of a \emph{join tree} and \emph{acyclicity} precisely.

Without loss of generality, we only consider natural joins with a connected join graph in this section\footnote{For equality predicates such as $R.A = S.B$, we treat $A$ and $B$ as the same attribute in this context. If the join graph has multiple components, we can generalize the concept of join tree to join forest}. For a natural join query $q$, its \emph{join graph} $G_q$ is an undirected graph where the vertices are the relations in $q$. If two relations have attributes in common, they are connected by an edge in $G_q$. A \emph{join tree} $T_q$ is a spanning tree of $G_q$ such that for every attribute $A$, the relations containing $A$ induce a \emph{connected} subgraph $T^A_q$ of $T_q$. The join tree is then used to define the \emph{acyclicity} of a query:

\begin{definition}[$\alpha$-acyclicity \cite{yannakakis1981YA}]\label{def:acyclic}
A natural join query $q$ is {\em acyclic} if and only if there exists a \emph{join tree} of $q$.
\end{definition}

Acyclicity is crucial for \YannAlg to achieve the $O(N + OUT)$ complexity because it guarantees a non-decreasing intermediate result in the join phase. If the subgraph for an attribute $A$ is \emph{not} connected, a tuple may survive in the first join involving $A$ but later get eliminated by the second join using $A$. This breaks the above non-decreasing property. An acyclic natural join satisfies the following lemma:

\begin{lemma}[\cite{Maier1983}]\label{thm:mst}
Let $q$ be an acyclic natural join query. For each edge $(R,S)$ in the join graph $G_q$, where $R$ and $S$ are the vertices (i.e., relations), define the weight of the edge $w(R,S)$ as the number of shared attributes between $R$ and $S$: $w(R,S)=\left|\operatorname{attr}(R)\cap\operatorname{attr}(S)\right|$. 
Then, a subgraph of $G_q$ is a join tree of $q$ if and only if it is a maximum spanning tree for $G_q$.
\end{lemma}

The intuition behind the lemma is that for a spanning tree $T$ of $G_q$, $T$'s total weight equals the summation of the edge count of each attribute-induced subgraph. $T$ is a join tree means that every attribute-induced subgraph $T^A$ is connected. This is equivalent to saying that every $T^A$ is a subtree, and it is impossible for any $T^A$ to have more edges (otherwise $T$ will not be a tree). $T$ must be a maximum spanning tree (MST). Note that the weights defined on the edges are not considered heuristics for join costs. They are used to transform the problem of finding a join tree into the problem of finding an MST in the join graph.

\SetAlgoNlRelativeSize{0}
\begin{algorithm}[t!]
    \KwIn{join graph $G_q$}
    \KwOut{tree $T$}
    $T \gets \varnothing$; $\mathcal{R} \gets$ all relations; $\mathcal{R}' \gets \{R_{max}\}$\;

    \While{$\mathcal R' \neq \mathcal R$}{
    Find an edge $e = \{R,S\} \in E(G_q)$ with the largest weight
    such that $R\in\mathcal R\setminus \mathcal R',S\in\mathcal R'$. Choose the edge with the largest $R$ to break ties\;\label{step:choose}
    
    Add $e$ to $T$ with direction from $R$ to $S$\;
    
    $\mathcal R'\gets \mathcal R'\cup\{R\}$\;\label{step:expand_r}
    }

    return $T$\;
    \caption{\TreeStruct}
    \label{alg:LargestRoot}
\end{algorithm}

We now know that for an acyclic query, a join tree guarantees a full (semi-join) reduction of the query, and we can find a join tree by constructing a maximum spanning tree on its weighted join graph. We next introduce our \TreeStruct algorithm. As shown in \cref{alg:LargestRoot}, we use Prim's algorithm to construct a maximum spanning tree $T$ on the join graph $G_q$. The edges in $T$ point from leaf to root, indicating a forward pass schedule. Because the algorithm starts with the largest relation $R_{max}$ in $\mathcal{R}'$, $R_{max}$ is the root of $T$ (hence the name \TreeStruct). And because of \cref{thm:mst}, $T$ is a join tree if query $q$ is acyclic, guaranteeing a full reduction in the transfer phase. 

Placing the largest relation at the root of the join tree is important, especially for queries following a star schema. It is more efficient to filter the much larger fact table using the dimension tables first before building a \BF on the fact table. In addition, \TreeStruct pushes larger relations toward the root by including them early in $T$ in the tie-breaking strategy in \cref{step:choose}. This allows larger relations to get filtered first by probing other \BFs before building their own, thus minimizing the total \BF construction time in the transfer phase. Notice that \cref{step:choose} in \TreeStruct does not specify a tie-breaking policy for choosing $S \in \mathcal{R}'$. In reality, most edges have weight 1 because relations typically join on only one attribute. Although the choice of $S$ does not compromise the theoretical guarantee of the algorithm producing an MST, it could affect the shape of the join tree. In general, a flatter tree allows more parallelism in building the \BFs, while a deeper tree might allow filtering irrelevant tuples out earlier. Both the tie-breaking policies for $R$ and $S$ do not affect the strong theoretical guarantee (i.e., a full reduction) of \TreeStruct.

Unlike \YannAlg, \TreeStruct also applies to cyclic queries. The algorithm's output is still a spanning tree with the largest relation at the root, but it is not a join tree. In this case, the transfer schedule generated by \TreeStruct does not guarantee a fully reduced instance for the subsequent join phase. Still, it transfers any predicate to all relations at least once and is effective empirically, as we will show in the experiments in \cref{sec:eval}.


\subsection{Choosing a Safe Join Order}
\label{sec:modeling:join}

Once the transfer phase generates a fully reduced instance of the database, the algorithm enters the join phase to produce the final output. According to \YannAlg, the join order is derived from the join tree used in the semi-join phase by performing the joins bottom up. Although such a (almost fixed) join order guarantees the asymptotic complexity of \YannAlg (i.e.,$O(N + OUT)$), it prevents the optimizer from exploring more join orders that potentially have smaller costs. Ideally, we want to leverage the cost models in the optimizer to search for cheaper plans, but we want to constrain the optimizer to only consider join orders with intermediate results always upper bounded by the output size. Such a ``safe'' join order provides a (theoretical) robustness guarantee: its runtime cost is at most a constant factor away from the optimal. In other words, even ill-behaved data distributions will not cause the runtime to deviate more than some bounded quantity.

\begin{definition}[\cite{safesubjoin}]\label{def:safe}
Let $q$ be an acyclic natural join query. A subjoin $q'$ of $q$ is {\em safe} if for every fully reduced instance $I$, we have $q'(I) = \pi_{\operatorname{attr}(q')}(q(I))$.
\end{definition}

The above definition ensures that if a subjoin $q'$ is safe, then the output of $q'$ is a projection of the final output, and thus $|q'(I)| \leq |q(I)|$. If every subjoin of a join order is safe, then the cumulative intermediate result size is within a constant factor of $|q(I)|$ (i.e., the optimal). It is straightforward to see that subjoins that involve Cartesian products can be unsafe. But unsafe subjoins are not restricted to Cartesian products. Consider the natural join $q = R(A,B,C) \Join S(A,B) \Join T(B,C)$. Let $I$ be the fully reduced instance:
$R = \{(1,1,1), (2,1,2), \dots, (n,1,n)\}$, $S = \{(1,1),(2,1), \dots, (n,1)\}$, and $T = \{(1,1),(1,2), \dots, (1,n)\}$.
Then subjoin $q' = S(A,B) \Join T(B,C)$ is unsafe because $|q'(I)| = n^2$, while $|q(I)| = n$. Therefore, any query plan that joins $S$ with $T$ first -- even on a fully reduced instance -- will create a quadratic blowup on the intermediate result.

One approach to avoid unsafe join orders is to \emph{identify the class of acyclic queries} for which any join order that does not involve Cartesian products is safe.

\begin{definition}[$\gamma$-acyclicity~\cite{acyclic_degrees}]\label{def:gamma}
A natural join query $q$ is $\gamma$-acyclic if and only if there is no $\gamma$-cycle in $q$. This is equivalent to (1) $q$ is $\alpha$-acyclic, and (2) we cannot find three relations $R, S, T$ with attributes $x, y, z$ that form a $\gamma$-cycle of size 3: $R(x, y), S(y, z), T(x, y, z)$.
\end{definition}

\begin{lemma}[\cite{acyclic_degrees}]
Every connected join expression\footnote{No Cartesian products, binary joins only.} $\theta$ of $q$ is monotone (i.e., no tuple gets removed while executing any binary join in $\theta$) if and only if $q$ is $\gamma$-acyclic.
\end{lemma}

\begin{theorem}\label{thm:gamma_safe}
Every subjoin (without Cartesian products) for natural join query $q$ is safe if and only if $q$ is $\gamma$-acyclic.
\end{theorem}

\begin{proof}
    It is sufficient to show that every subjoin is safe if and only if every connected join expression is monotone.

    Consider any connected join expression $\theta'$ for subjoin $q'$, $\theta'_1$ for subjoin $q'_1$, and $\theta'_2$ for subjoin $q'_2$,
    where $\theta'=\theta'_1\Join \theta'_2$.
    Because every subjoin without Cartesian products is safe, we have
    $q'(I)=\pi_{\operatorname{attr}(q')}(q(I))$, 
    $q'_1(I)=\pi_{\operatorname{attr}(q'_1)}(q(I))$, and
    $q'_2(I)=\pi_{\operatorname{attr}(q'_2)}(q(I))$.
    Because $\operatorname{attr}(q'_i) \subseteq \operatorname{attr}(q')$ for $i = 1, 2$,
    we have $|\pi_{\operatorname{attr}(q')}(q(I))| \ge |\pi_{\operatorname{attr}(q_i')}(q(I))|$.
    Therefore, $\theta'$ is monotone.

    For the other direction,
    consider any connected join expression $\theta'$ of a subjoin $q'$.
    Extend $\theta'$ to a complete join expression $\theta$ of $q$.
    Because $\theta'$ is part of $\theta$ and every connected join expression of $q$ is monotone,
    we have $q'(I)=\pi_{\operatorname{attr}(q')}(q(I))$ for any fully reduced instance $I$.
    Therefore, $q'$ is safe.
\end{proof}

\cref{thm:gamma_safe} gives a strong \textbf{robustness guarantee}: if a query is $\gamma$-acyclic, we can fully trust the optimizer for join ordering on a fully-reduced instance (i.e., the join phase) because it can never pick an unsafe join order. $\gamma$-acyclic queries are a subset of $\alpha$-acyclic (i.e., acyclic) queries according to \cref{def:gamma}. To quickly check for $\gamma$-acyclicity in practice, it is sufficient (not necessary) to show that no two relations in the join graph are directly connected by more than one edge (i.e., no composite-key joins).

For queries that are acyclic but not $\gamma$-acyclic, we must \emph{supervise the optimizer} to check whether a given subjoin is safe. 
A safe subjoin can be characterized by the following lemma:

\begin{lemma}[\cite{safesubjoin}]\label{lemma:safe}
Let $q$ be an acyclic natural join query. A subjoin $q'$ of $q$ is safe if and only if there exists some join tree of $q$ such that the relations in $q'$ are connected.
\end{lemma}

For the example natural join $q = R(A,B,C) \Join S(A,B) \Join T(B,C)$, there is only one join tree for $q$: $S - R - T$. Hence, both $R \Join S$ and $R \Join T$ are safe subjoins, but $S \Join T$ is not. Using \cref{lemma:safe}, we developed the \SafeSubJoin algorithm to detect whether a subjoin $q'$ is safe. As shown in \cref{alg:SafeSubJoin}, \SafeSubJoin first computes a maximum spanning tree $T'$ for $q'$ using the \TreeStruct algorithm. It then continues to run another instance of \TreeStruct by modifying the initialization step as $T \gets T'$; $\mathcal{R} \gets$ all relations in $q$; $\mathcal{R}' \gets$ all relations in $q'$. \SafeSubJoin returns true if the resulting spanning tree $T$ is a maximum spanning tree of $q$ (i.e., a join tree of $q$).

\SetAlgoNlRelativeSize{0}
\begin{algorithm}[t!]
    \KwIn{natural join $q$, subjoin $q'$}
    \KwOut{True or False}
    $T' \gets$ \TreeStruct($G_{q'}$)\;
    $T \gets$ \TreeStruct($G_{q}$) with the initialization step as:
    $T \gets T'$; $\mathcal{R} \gets$ all relations in $q$; $\mathcal{R}' \gets$ all relations in $q'$\;
    \eIf{$T$ is a maximum spanning tree of $q$}{
        return True\;
    }{
        return False\;
    }    
    \caption{\SafeSubJoin}
    \label{alg:SafeSubJoin}
\end{algorithm}

\section{Integration with \duckdb}
\label{sec:impl}

We describe how to integrate the \emph{\RPT} (\rpt) algorithm into \duckdb (v0.9.2)~\cite{2019duckdb}, a fast data analytics system, in this section. We first describe \duckdb's execution model and its optimizer briefly in \cref{sec:impl:duckdb}. We next introduce the new \BF operators in \cref{sec:impl:bf}. Finally, we introduce the new \RPT module that inserts the \BF operators into the query plan based on the transfer schedule obtained by running the \TreeStruct algorithm in \cref{sec:modeling:transfer}.



\subsection{\duckdb Preliminaries}
\label{sec:impl:duckdb}

\duckdb is a state-of-the-art in-process analytical database management system. It adopts a push-based vectorized execution engine, where each pipeline (i.e., a sequence of physical operators) processes tuples in batches (i.e., a data chunk, default batch size = 2048) to amortize the interpretation overhead and improve CPU parallelism. As shown in \cref{fig:duckdb-pipeline}, each physical operator can be in one of the following three roles: \emph{source}, \emph{operator}, and \emph{sink}, depending on its position within the pipeline. The source implements the \texttt{GetData} function to retrieve a new data chunk at the beginning of a pipeline. Intermediate operators implement the \texttt{Execute} interface that computes on an input data chunk and then outputs the result chunk. The sink operator is located at the end of a pipeline and is typically a pipeline breaker. Its interface consists of three functions: \texttt{Sink}, \texttt{Combine}, and \texttt{Finalize}. \texttt{Sink} is called to receive and buffer the data chunks until incoming data is exhausted. Next, \texttt{Combine} and \texttt{Finalize} are called to perform some final computations to get ready to distribute data to the next pipeline (or final output). \texttt{Combine} is called once per thread, while \texttt{Finalize} is called when all threads are finished.

\begin{figure}[t!]
    \centering
    \includegraphics[width=\linewidth]{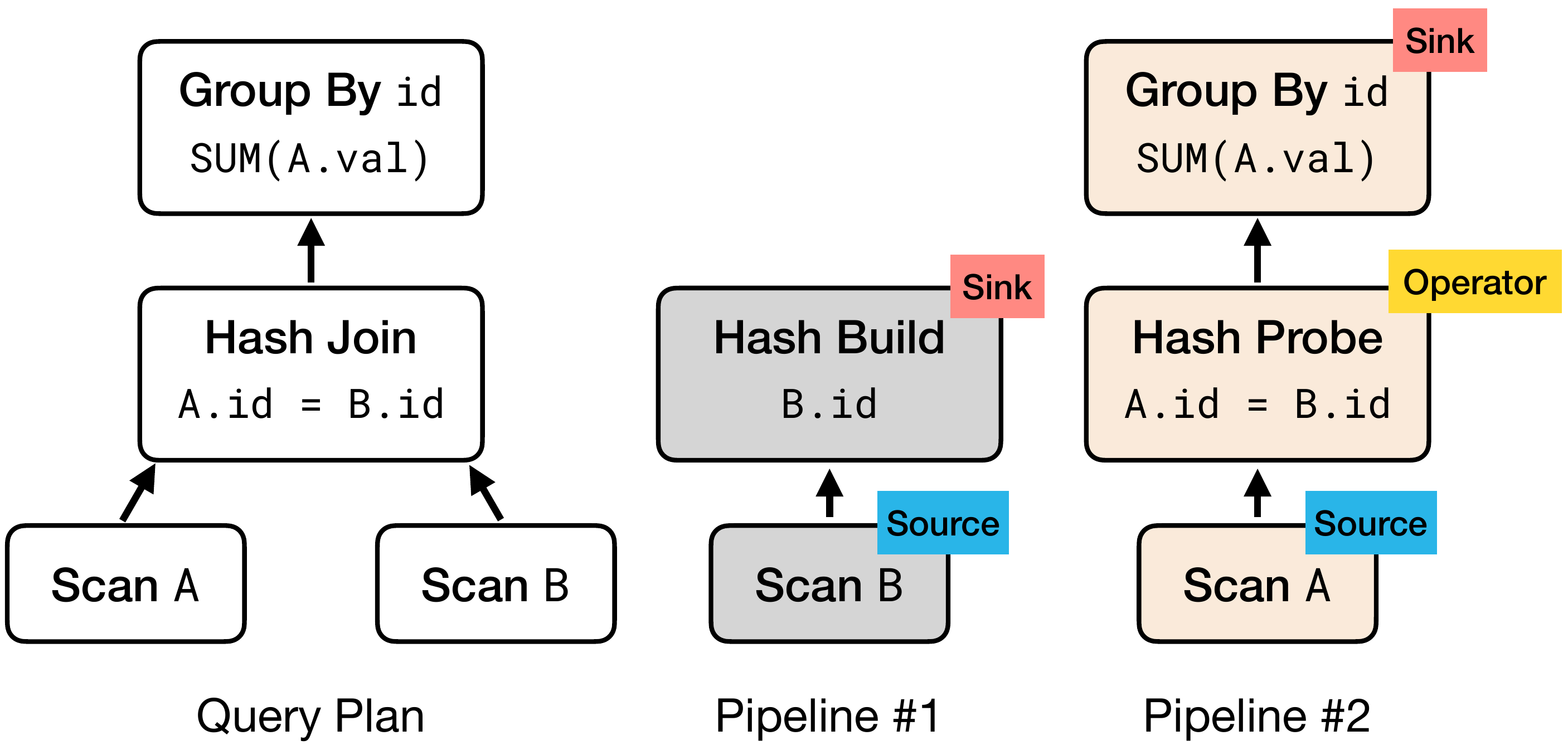}
    \caption{Example pipelines \& operator roles in \duckdb~\cite{2023DuckDBPPT}}
    \label{fig:duckdb-pipeline}
\end{figure}

\duckdb's optimizer includes separate logical and physical optimization phases, as shown in \cref{fig:duckdb-optimizer}. Logical optimization performs a sequence of steps such as expression rewrite and filter pushdown, each of which is a separate submodule in the logical optimizer. \duckdb's join order submodule uses dynamic programming for join order optimization~\cite{2008dphyper} and falls back to a greedy algorithm for large/complex join graphs.

\begin{figure}[t!]
    \centering
    \includegraphics[width=\linewidth]{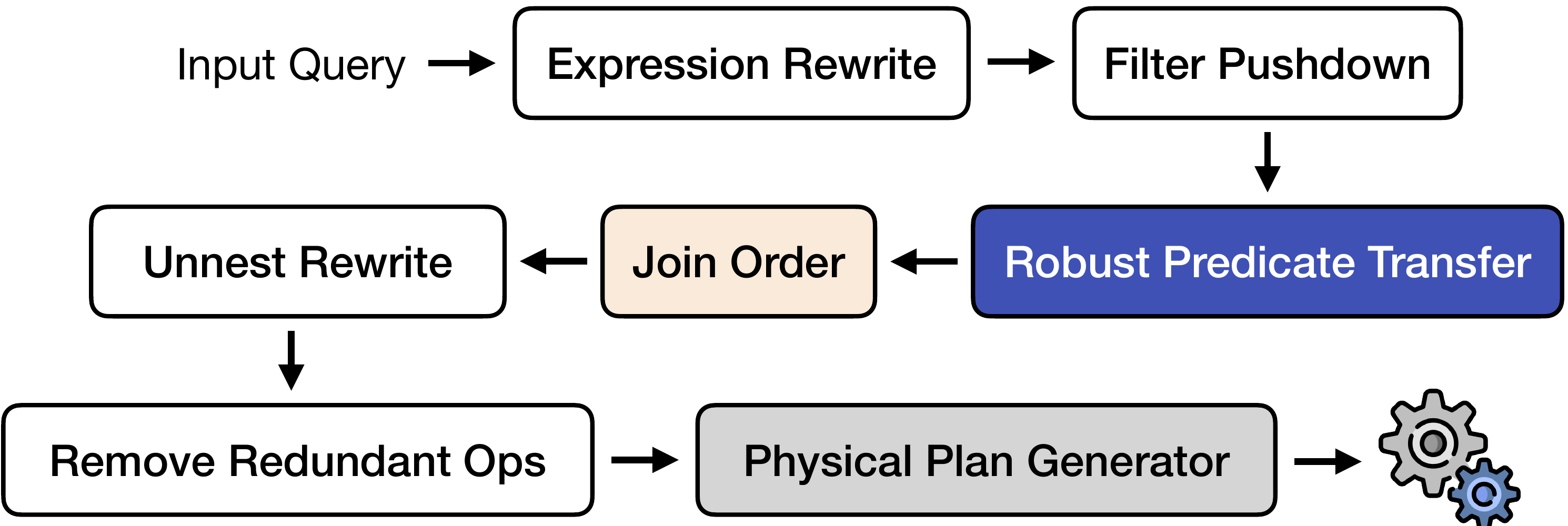}
    \caption{Workflow of \duckdb's optimizer}
    \label{fig:duckdb-optimizer}
\end{figure}

\subsection{Bloom Filter Operators}
\label{sec:impl:bf}

To implement \RPT in \duckdb, we introduce two new physical operators based on \BFs: \CreateBF and \UseBF. We use the \BF implementation from Apache Arrow 16.0~\cite{Arrow}. It is a blocked \BF~\cite{putze2007cache} with operations accelerated using AVX2 instructions. Because a vectorized probe to the \BF returns a bit vector while \duckdb uses a selection vector to mark valid entries in a data chunk, we implemented an efficient bit-to-selection vector conversion according to~\cite{bitvector2selvec}.

\CreateBF is a physical operator that gathers/buffers the input data chunks and creates one or more \BFs on given columns. Its logical counterpart \LogicalCreateBF will be used in the logical optimization. \CreateBF can act as both a \emph{sink} and a \emph{source} (more about this in \cref{sec:impl:rpt}). In the \texttt{Sink} function, we receive input data chunks and keep them in thread-local buffers. No computation is needed for \texttt{Combine}. At \texttt{Finalize}, we traverse each thread-local data buffer to create a \BF for each given column. The false positive rate (FPR) of the \BF is set to $2\%$ (Arrow's default). When using \CreateBF as a source, we implement \texttt{GetData} by assigning each thread a disjoint range of chunk IDs in the data buffers for parallel scanning.

\UseBF is another physical operator that outputs the \BF result for each tuple in the input data chunk. Similarly, it has a logical counterpart \LogicalUseBF. \UseBF is used as an intermediate operator. The \texttt{Execute} function takes in a data chunk, uses the tuples to probe the Bloom filter(s) in a vectorized fashion, and outputs the data chunk with an updated selection vector.


\subsection{\RPT Module}
\label{sec:impl:rpt}

\begin{figure}[t!]
    \centering
    \includegraphics[width=\linewidth]{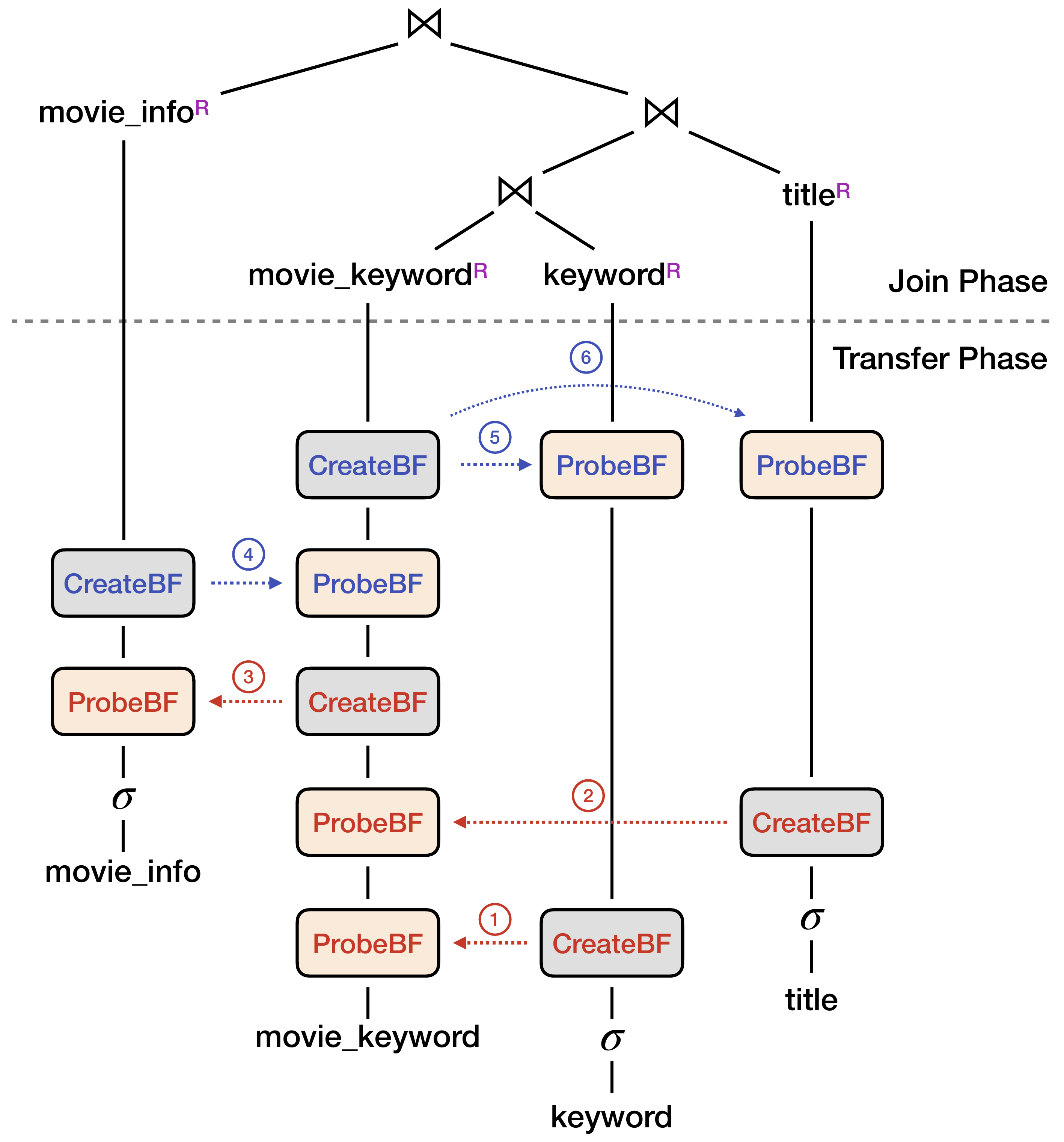}
    \caption{Query plan of JOB 3a with \RPT integrated \textnormal{-- red denotes forward pass while blue denotes backward pass.}}
    \label{fig:pt-plan}
\end{figure}

We introduce the \RPT module in \duckdb's logical optimizer to insert \LogicalCreateBF and \LogicalUseBF operators into the query plan, as shown in \cref{fig:duckdb-optimizer}. The RPT module constructs a join graph from the input plan and runs the \TreeStruct algorithm to obtain a transfer schedule (including forward and backward passes). For each semi-join $R \Semijoin S$ in the transfer schedule, we insert a \LogicalCreateBF for $S$ and a \LogicalUseBF for $R$ using $S$'s \BF. These logical operators are later replaced by \CreateBF and \UseBF in the physical plan generator.

Take the query plan for JOB 3a as an example. Suppose the transfer schedule generated by \TreeStruct is the same as that in \cref{fig:yannakakis-semi-join}. Then \cref{fig:pt-plan} shows the physical plan after inserting \CreateBF and \UseBF. The solid black lines represent the flow of data chunks (from the bottom up) while the dashed red/blue arrows indicate the transfer of \BFs (via shared memory). Each \CreateBF first acts as a \emph{sink} operator that buffers the data chunks at the end of the pipeline and creates a \BF. Then \CreateBF functions as the \emph{source} operator of the next pipeline, where it feeds the buffered data chunks to subsequent operators such as \UseBF and hash join.

We also implemented optimizations to prune unnecessary \BF operations. In particular, if the build-side relation in a primary-foreign-key join has not been filtered before, we can omit the pair of \CreateBF and \UseBF because the semi-join is trivial (i.e., it does not eliminate any tuple). We can also skip the entire backward pass if the transfer order aligns with the join order in the join phase.

\section{Evaluation}
\label{sec:eval}

We evaluate the robustness of \rpt-integrated \duckdb in this section\footnote{Our source code can be found in https://github.com/zzjjyyy/PredTransDuckDB}. We conduct the experiments on a physical machine with two $\text{Intel}^\text{\textregistered} \text{Xeon}^\text{\textregistered}$ Platinum 8474C @ 2.1GHz, 512GB DDR5 RAM, and 8TB Samsung 870 QVO SATA III 2.5" SSD. The operating system is Debian 12.5. We compare vanilla \duckdb (labeled as \duckdb) against \duckdb equipped with \rpt (labeled as \rpt) using four standard benchmarks: \tpch (SF = 100)~\cite{TPCH}, Join Order Benchmark (\job)~\cite{JOB}, \tpcds (SF = 100)~\cite{TPCDS}, and \dsb (SF = 100)~\cite{DSB}. We run the experiments under \duckdb's main-memory setting where the tables are pre-loaded and decompressed in the buffer pool. We examine the case where the base tables and intermediate results do not fit in memory in~\cref{sec:eval:on-disk}. We execute the queries using a single thread except for the multi-threaded experiments in \cref{sec:eval:multi-thread}.

\subsection{End-to-end Robustness}
\label{sec:eval:end-to-end}

In the following experiments, we modified \duckdb's optimizer to generate random join orders. For each evaluated query, we generate $N$ random \emph{left-deep} plans and $N$ random \emph{bushy} plans, where $N$ is proportional to the number of joins $m$ in that query. Specifically, we set $N = 20$ for the simplest 3-join queries and $N = 1000$ for the most complex query (i.e., Query 29 from JOB) with 17 joins, and therefore $N = 70m - 190$ for $3 \le m \le 17$. To produce a left-deep plan, we randomly pick a base table that is joinable\footnote{Has an edge in the join graph, i.e., no Cartesian product.} with the current (intermediate) table as the right-most leaf at each iteration. For bushy plans, we randomly remove two joinable tables from the candidate set (which initially contains all base tables) and insert their intermediate table back at each iteration until the set contains only one element (i.e., the final plan).

\begin{figure*}[t!]
    \centering
    \begin{subfigure}{0.34\linewidth}
        \includegraphics[width=\linewidth]{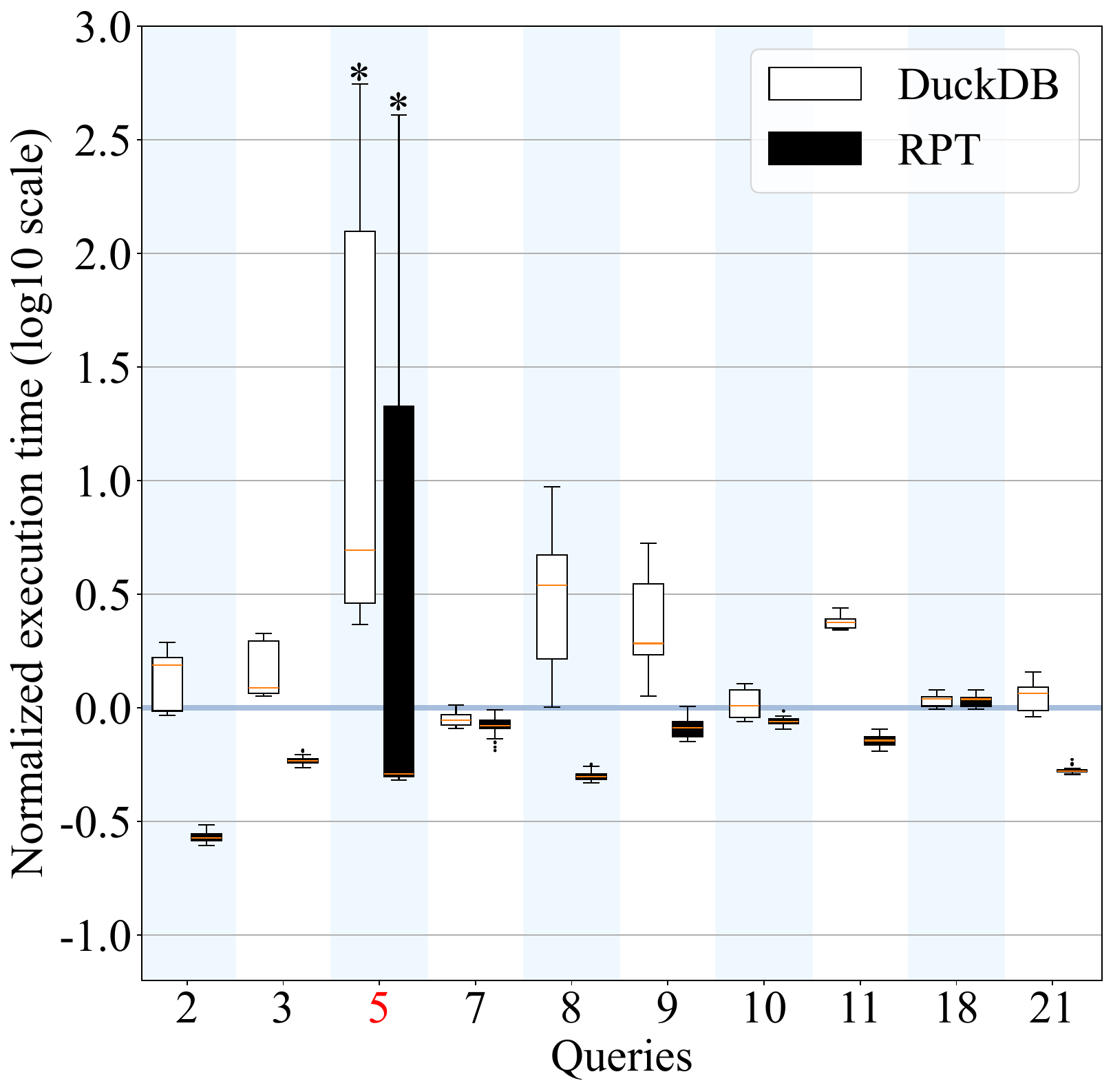}
        \caption{\tpch}
        \label{fig:eval-left-deep-tpch}
    \end{subfigure}
    \begin{subfigure}{0.64\linewidth}
        \includegraphics[width=\linewidth]{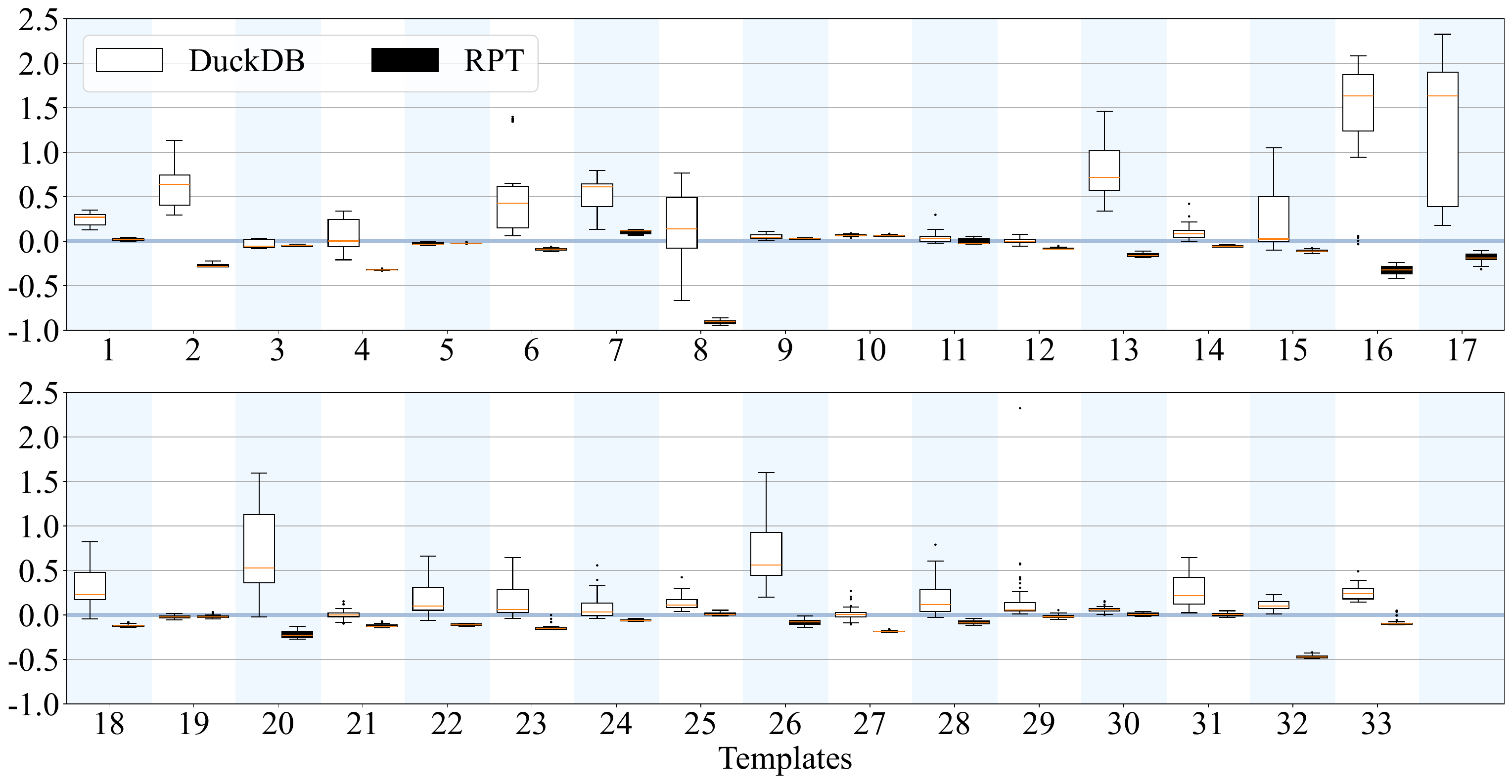}
        \caption{\job}
        \label{fig:eval-left-deep-job}
    \end{subfigure}
    \begin{subfigure}{0.96\linewidth}
        \includegraphics[width=\linewidth]{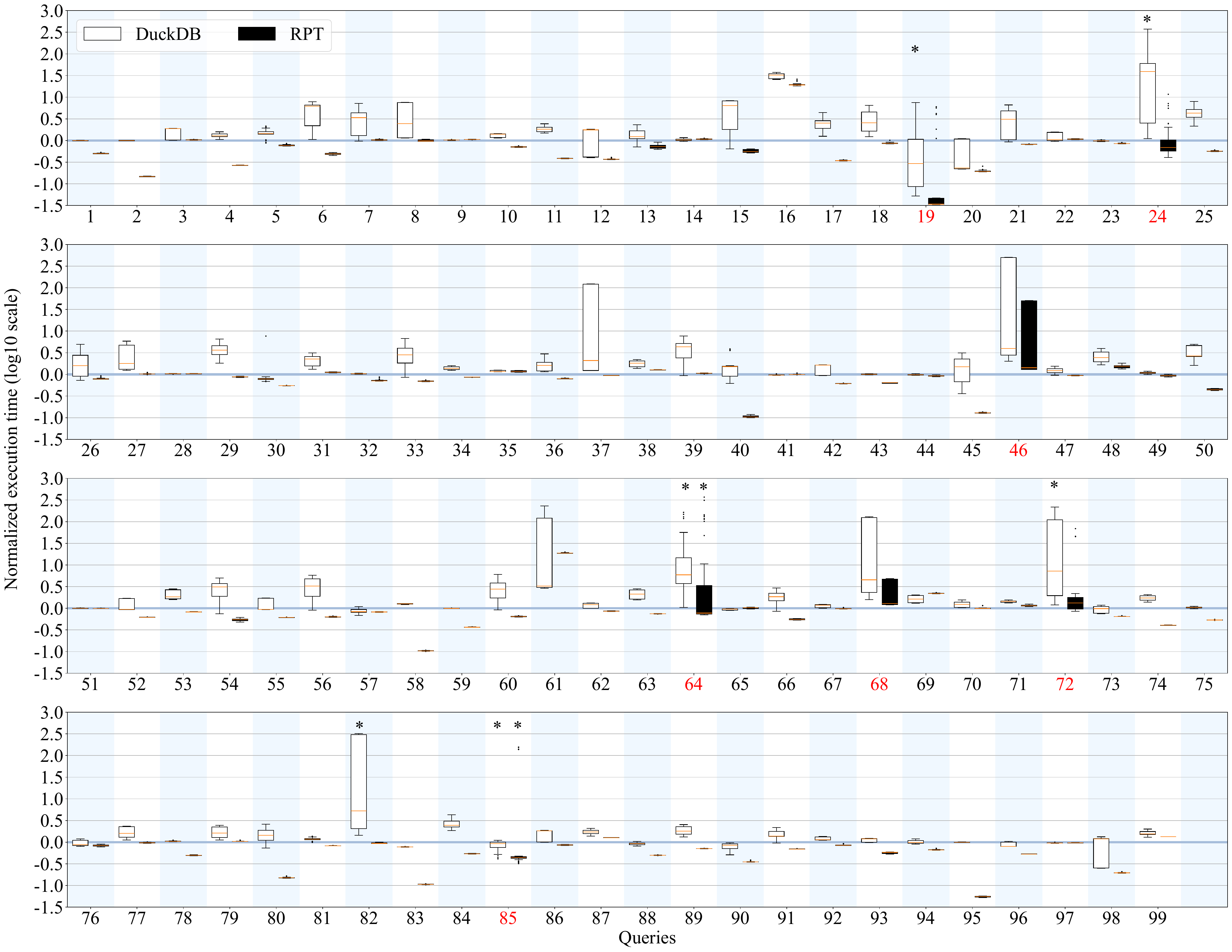}
        \caption{\tpcds}
        \label{fig:eval-left-deep-tpcds}
    \end{subfigure}
    \caption{Distribution of the execution time of random left-deep plans for each query in \tpch, \job, and \tpcds \textnormal{-- Normalized by the execution time of default \duckdb. The figure is log-scaled. The box denotes 25- to 75-percentile (with the orange line as the median), while the horizontal lines denote min and max (excluding outliers). `*' indicates timeouts. Cyclic queries are in red.}}
    \label{fig:eval-left-deep}
\end{figure*}

\begin{table}[t!]
\begin{center}
    \caption{Robustness Factors for left-deep joins.}
    \begin{tabular}{p{33pt}|p{12pt}p{12pt}p{12pt}|p{12pt}p{12pt}p{13pt}|p{12pt}p{12pt}p{12pt}}
    \toprule
    \centering \robustmetric & \multicolumn{3}{|c|}{\tpch} & \multicolumn{3}{|c|}{\job} & \multicolumn{3}{|c}{\tpcds} \\
                             & Avg & Min & Max             & Avg    & Min  & Max       & Avg  & Min & Max            \\
    \midrule
    \centering \duckdb       & 2.7 & 1.2 & 9.3             & 30.4   & 1.1  & 371       & 7.2  & 1.0 & 224            \\
    \centering \textbf{\rpt} & 1.3 & 1.2 & 1.5             & 1.2    & 1.0  & 1.6       & 1.1  & 1.0 & 1.5            \\
    \bottomrule
    \end{tabular}
    \label{tab:RF-left-deep}
\end{center}
\end{table}

\begin{table}[t!]
\begin{center}
    \caption{Robustness Factors for bushy joins.}
    \begin{tabular}{p{33pt}|p{12pt}p{12pt}p{12pt}|p{12pt}p{12pt}p{13pt}|p{12pt}p{12pt}p{12pt}}
    \toprule
    \centering \robustmetric & \multicolumn{3}{|c|}{\tpch} & \multicolumn{3}{|c|}{\job} & \multicolumn{3}{|c}{\tpcds} \\
                             & Avg & Min & Max             & Avg    & Min  & Max       & Avg & Min & Max              \\
    \midrule
    \centering \duckdb         & 5.1 & 1.2 & 13.7            & 120   & 1.1  & 1747       & 35.0 & 1.0 & 1226            \\
    \centering \textbf{\rpt}          & 1.8 & 1.2 & 3.0             & 1.6   & 1.1  & 7.7        & 1.8  & 1.0 & 4.2             \\
    \bottomrule
    \end{tabular}
    \label{tab:RF-bushy}
\end{center}
\end{table}

\begin{table}[t!]
\begin{center}
    \caption{Average speedups over \duckdb (optimizer's plan)}
    \begin{tabular}{c|c|c|c|c}
    \toprule
    \centering Speedup        & \tpch          & \job               & \tpcds        & \dsb          \\
    \midrule
    \centering Bloom Join     & $1.15\times$   & $1.13\times$       & $1.05\times$  & $1.06\times$  \\
    \centering \pt            & $1.45\times$   & $1.46\times$       & $1.27\times$  & $1.18\times$  \\
    \centering \textbf{\rpt}  & $1.44\times$   & $1.46\times$       & $1.56\times$  & $1.54\times$  \\
    \bottomrule
    \end{tabular}
    \label{tab:rpt-perf}
\end{center}
\end{table}

\begin{figure*}[t!]
    \centering
    \begin{subfigure}{0.34\linewidth}
        \includegraphics[width=\linewidth]{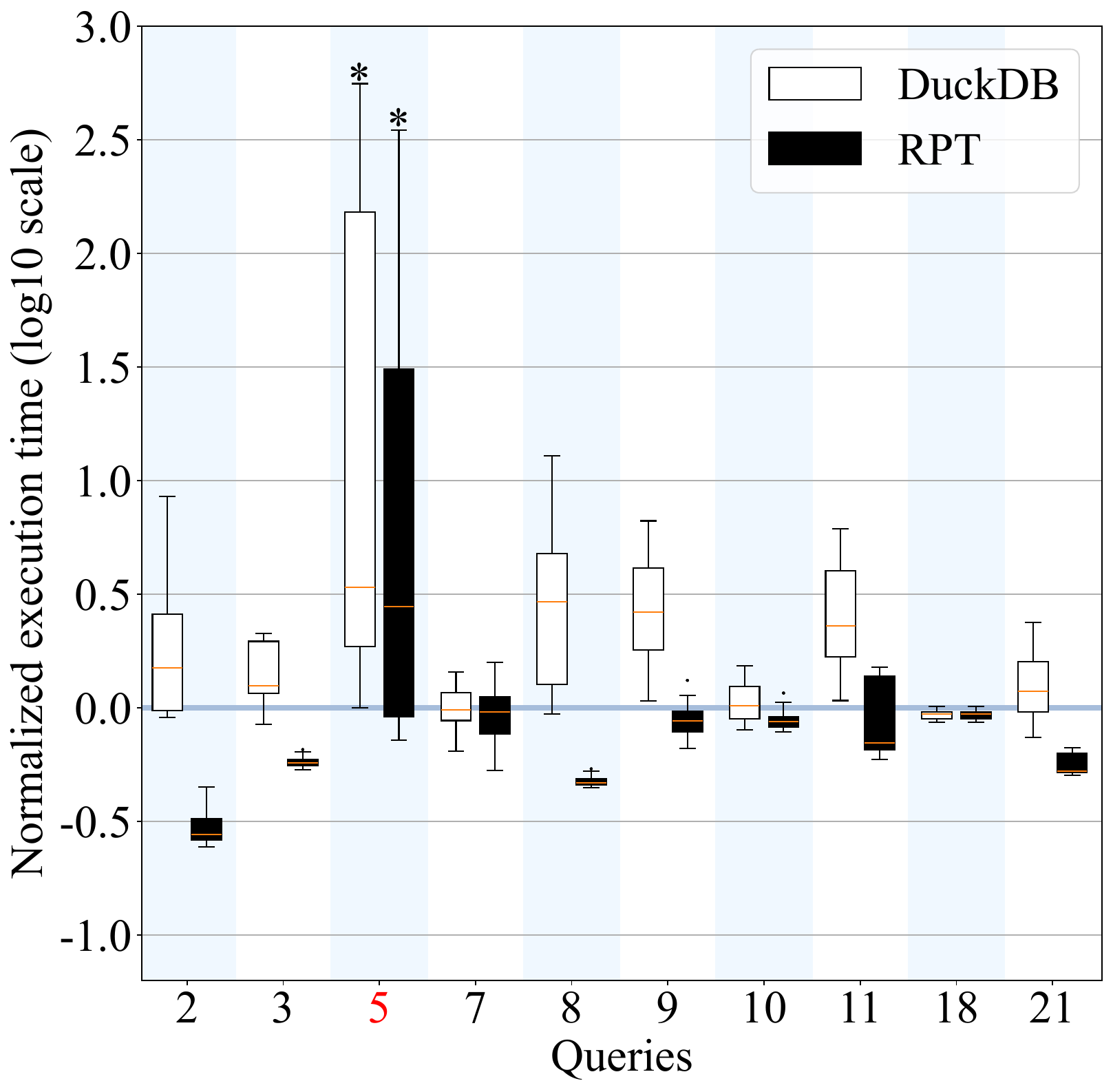}
        \caption{\tpch}
        \label{fig:eval-bushy-tpch}
    \end{subfigure}
    \begin{subfigure}{0.64\linewidth}
        \includegraphics[width=\linewidth]{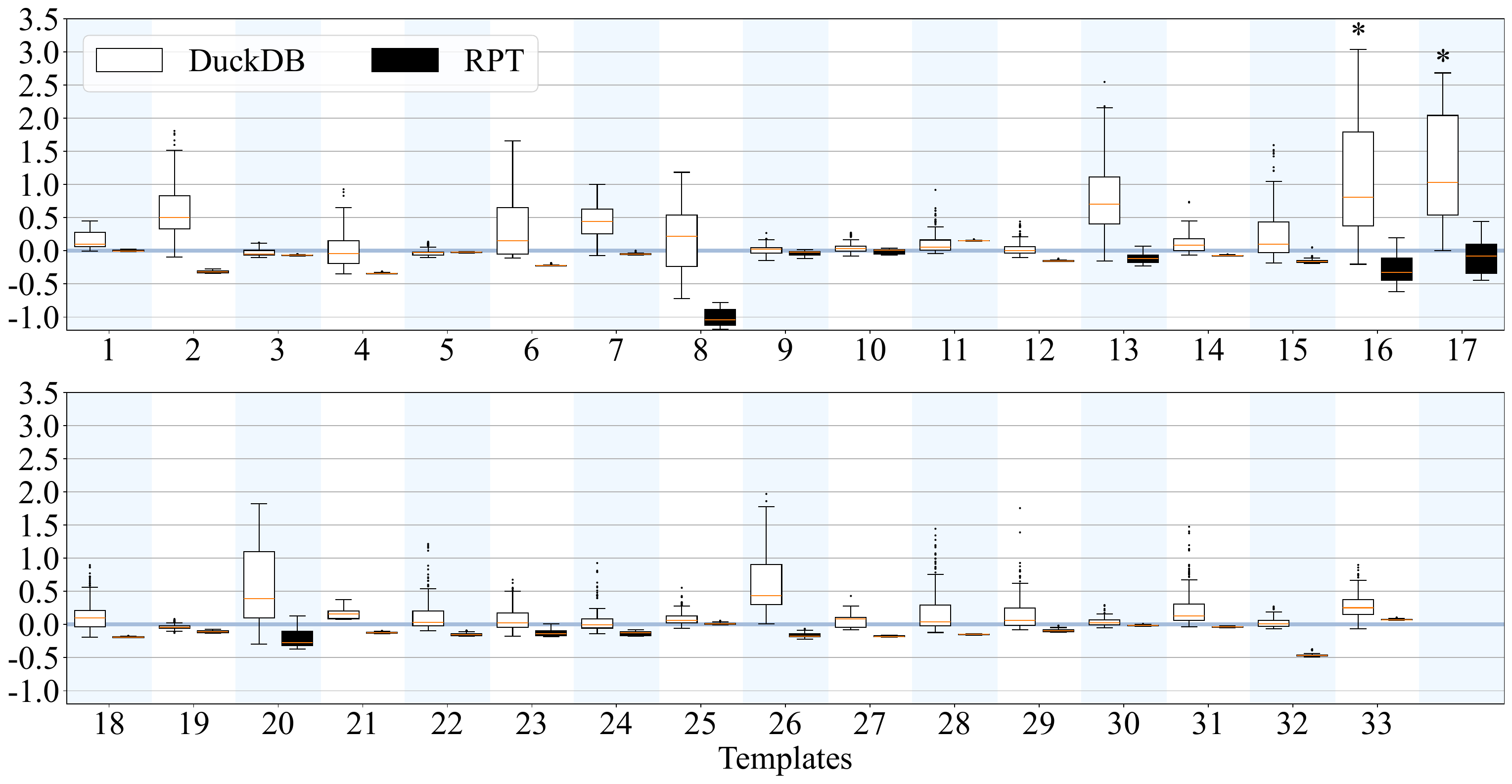}
        \caption{\job}
        \label{fig:eval-bushy-job}
    \end{subfigure}
    \caption{Distribution of the execution time of random bushy plans for each query in \tpch and \job \textnormal{-- Normalized by the execution time of default \duckdb. The figure is log-scaled. The box denotes 25- to 75-percentile (with the orange line as the median), while the horizontal lines denote min and max (excluding outliers). `*' indicates timeouts. Cyclic queries are in red.}}
    \label{fig:eval-bushy}
\end{figure*}

\subsubsection{Acyclic Queries (left-deep)}
\cref{fig:eval-left-deep} shows the distribution of the end-to-end execution time of the random left-deep plans for each query. We omit queries in \tpch with less than two joins because they are trivial in terms of join ordering. For \job queries, we present one result for each of the 33 query templates. The execution times (for the baseline and \rpt) for each query are normalized by the time $t_{opt}$ of \duckdb running its default optimizer's plan. The figure is in \emph{log scale} with the normalization line (i.e., horizontal zero) highlighted. We set the timeout to $1000 \times t_{opt}$. The `*' above a bar indicates that at least one of the random plans incurs a timeout for this query. Cyclic queries are marked by red query numbers.

We observe impressive join order robustness when using \rpt in \duckdb for \emph{all} acyclic queries. To quantify this, we define the \emph{Robustness Factor} (\robustmetric) as the ratio between the maximum and the minimum execution time. \cref{tab:RF-left-deep} presents the average, min, and max RFs for \duckdb and \rpt in each benchmark. The average RF for \duckdb with \rpt is consistently near 1 with the max (i.e., worst-case) RF = 1.473 for Query 13 in \tpcds. This is orders-of-magnitude more robust compared to the baseline. 

Additionally, applying \rpt improves the end-to-end query performance for most queries in the benchmarks (most \rpt boxes are below 0 in \cref{fig:eval-left-deep}). \cref{tab:rpt-perf} presents the average speedups of \rpt over the default \duckdb running its optimizer's plan (i.e., $t_{opt}$). We also included Bloom Join~\cite{bloomjoin} and the original \PT~\cite{yang2023PT} (\pt) as references\footnote{The execution time of each query can be found in Appendix A}. Besides robustness guarantees, applying \rpt reduces the execution time per query by $\approx1.5\times$ on average (geometric mean). Bloom join only achieves a marginal speedup over the baseline, and it does not improve join-order robustness\footnote{The full robustness results for Bloom join and \pt can be found in Appendix B}. \rpt outperforms the original \pt in \tpcds and \dsb thanks to the \texttt{\TreeStruct} algorithm. More importantly, \rpt guarantees query robustness. \cref{fig:RPTisBetter} shows selected queries from \job and \tpcds where the performance of the original \pt is sensitive to different join orders. The root cause is that the transfer schedules produced by \pt can lead to an incomplete reduction in the semi-join phase, as discussed in~\cref{sec:modeling:transfer}.

\begin{figure}[t!]
    \centering
    \includegraphics[width=\linewidth]{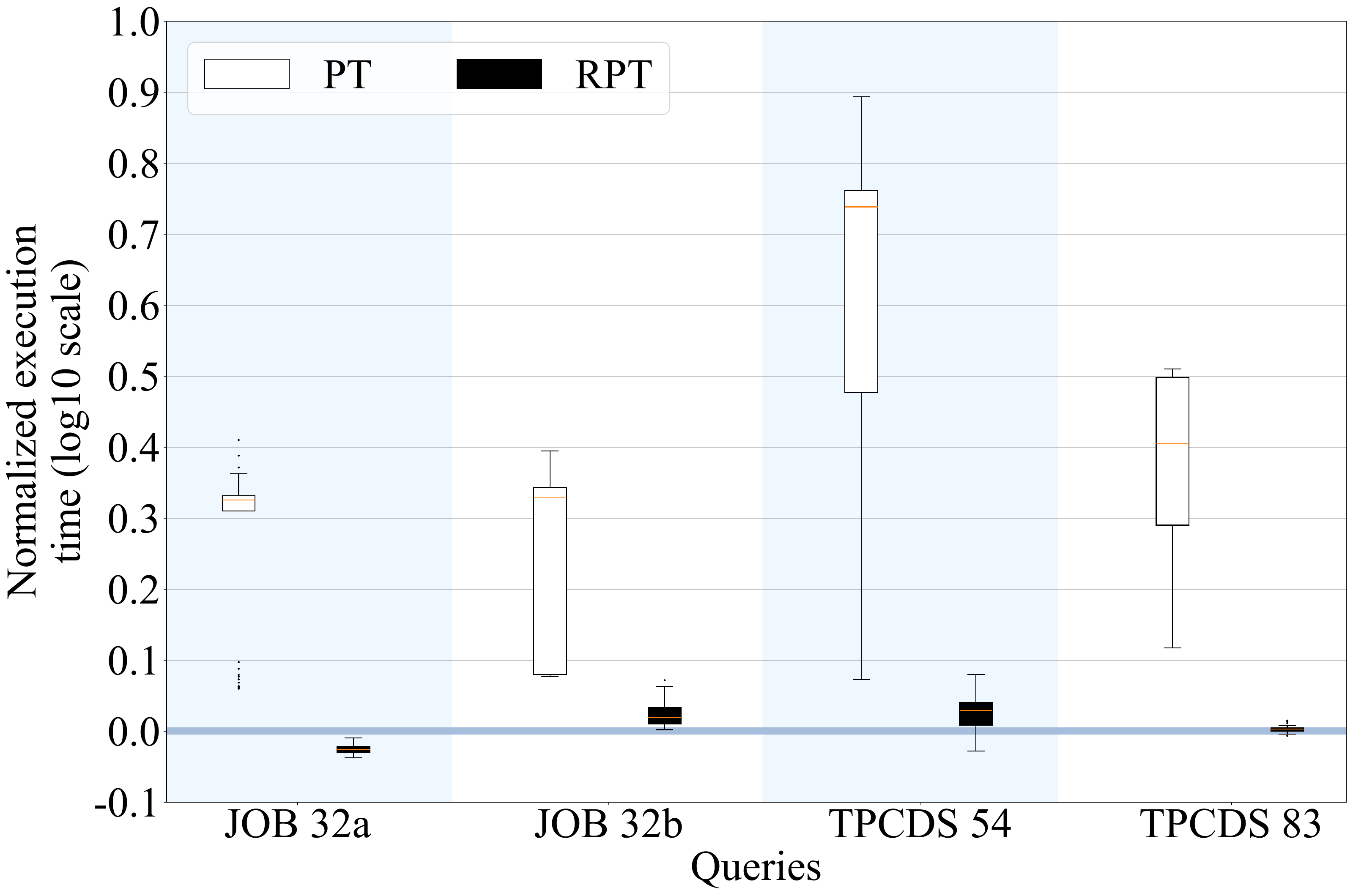}
    \caption{The distribution of the execution time of \pt and \rpt with random left deep plans for the selected query in \job and \tpcds \textnormal{-- Normalized by the execution time of \rpt with the optimizer's join order. The figure is log-scaled.}}
    \label{fig:RPTisBetter}
\end{figure}

The above results are encouraging. They show that join order optimization might no longer be a critical challenge at least for acyclic queries (which are the majority) if we implement joins using \rpt. In fact, the execution time of \rpt using the optimizer's join order is within the horizontal lines (i.e., min to max excluding outliers) for \emph{every} acyclic query in \cref{fig:eval-left-deep}. Future optimizers could, therefore, become much more efficient: they can better tolerate cardinality estimation errors, and they require simpler join enumeration algorithms because a left-deep plan is already good enough. 

A few acyclic queries (13, 29, and 48) in \tpcds have slightly larger variances than the others in \cref{fig:eval-left-deep-tpcds}. Query 13 and 48 include predicates that cannot be pushed down before the tables are joined in \duckdb, e.g., ($R.a$ < 100 AND $S.b$ < 200) OR ($R.a$ > 500 AND $S.b$ > 400). It is preferable to join $R$ and $S$ earlier if the predicate is selective. Query 29 is acyclic but not $\gamma$-acyclic. According to the analysis in \cref{sec:modeling:join}, certain join orders are unsafe. Although these special cases exhibit adequate robustness with random plans, they can still benefit from the optimizer.

\subsubsection{Acyclic Queries (bushy)}
We show the distribution of the end-to-end execution time of random bushy plans in \cref{fig:eval-bushy} with the robustness factors summarized in \cref{tab:RF-bushy}. We omit results for individual queries of \tpcds in \cref{fig:eval-bushy}\footnote{They can be found in Appendix C}. When including bushy plans, \rpt exhibits similar robustness measures against random join orders as in the left-deep case, with the average RF < 1.8 and the max RF = 7.7 for Query 17e in \job.

We notice a slight robustness degradation for a few queries (e.g., \tpch Q7 and \job 16b, 17e) when switching from left-deep to bushy plans. They share the common reason that the optimizer mistakenly placed the larger table on the build side of hash joins in the worst plans (out of random bushy plans). As shown in \cref{fig:wrong-build-side}, for example, picking the wrong build side for the top hash join alone in JOB 17e slows down the query by $37\%$. Such a mistake is unlikely in a left-deep plan because each base table (i.e., build side) is typically filtered heavily in the transfer phase of \rpt while the size of the intermediate result (i.e., probe side) increases monotonically in the join phase.

\begin{figure*}[t!]
    \centering
    \includegraphics[width=\linewidth]{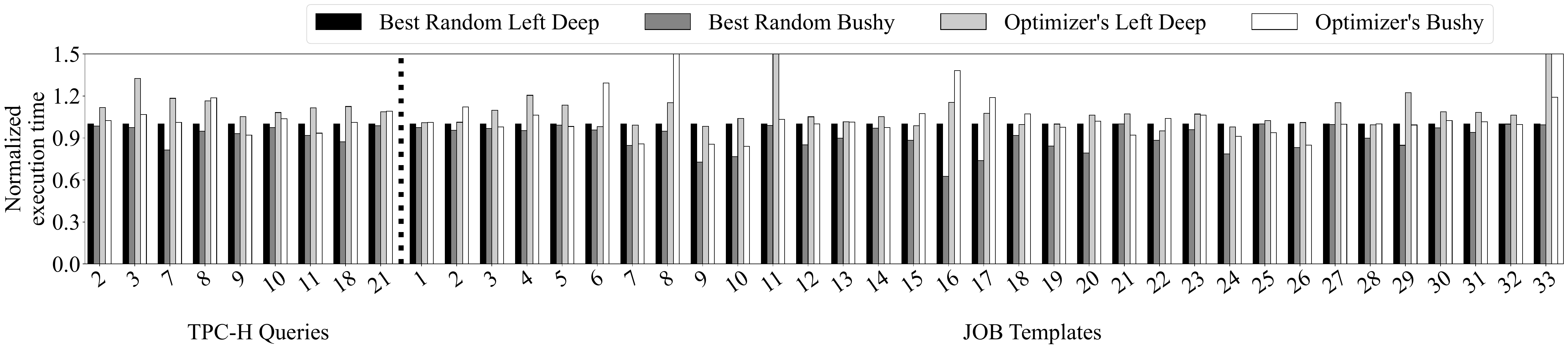}
    \caption{Speed up of bushy over left-deep plans. \textnormal{-- We draw the minimum execution time of \rpt out of random left-deep/bushy plans as well as the execution time of \rpt with the optimizer's left-deep/bushy plan for each query in \tpch and \job.}}
    \label{fig:left-bushy}
\end{figure*}

\begin{figure}[t!]
    \centering
    \includegraphics[width=\linewidth]{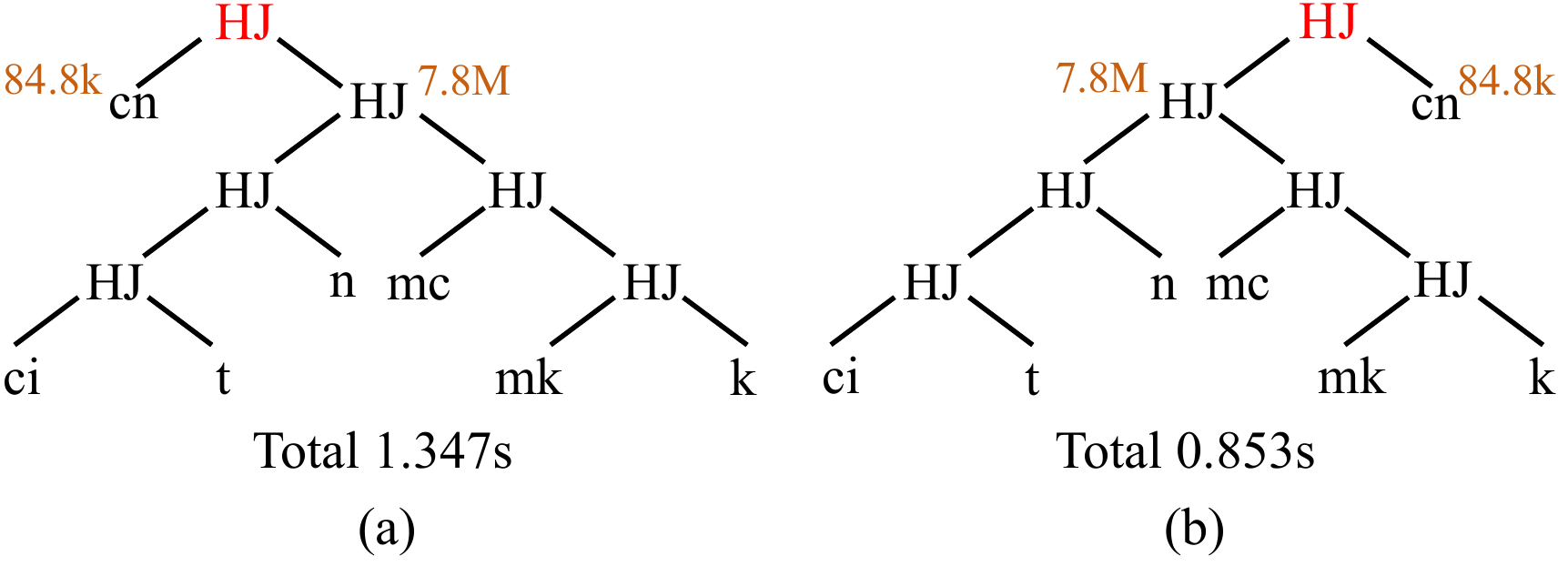}
    \caption{Slowdown caused by picking the wrong build side of a hash join (JOB 17e) \textnormal{-- (a) is a random plan with the incorrect build side for the top HJ; (b) is the fixed plan with the build side and probe side flipped.}}
    \label{fig:wrong-build-side}
\end{figure}

To demonstrate the performance gain of considering bushy plans, we select the best (i.e., with minimum execution time) random left-deep plan and the best random bushy plan for each query and compare their performance in \cref{fig:left-bushy}. We also include the left-deep plan and the bushy plan produced by \duckdb's optimizer for each query (labeled as Optimizer's Left Deep and Optimizer's Bushy, respectively) as references in the figure. We observe that considering bushy plans in the join phase of \rpt only speeds up the end-to-end execution by $6\%$ and $11\%$ for \tpch and \job, respectively compared to left-deep. Most optimizer's plans are slightly slower than the best ones from our randomly generated join orders, but the relative speedups of considering bushy plans remain small ($10\%$ for \tpch and $5\%$ for \job).
The semi-join reduction carried out in the transfer phase of \rpt significantly reduces the benefit of exploring a larger plan enumeration space. Therefore, we conclude that it is unnecessary to explore bushy plans when applying \RPT because bushy plans could sacrifice robustness for modest performance improvement.

\subsubsection{Cyclic Queries}
\rpt does not provide robustness guarantees for cyclic queries, as shown by the red-labeled queries in \cref{fig:eval-left-deep} and \cref{fig:eval-bushy} (i.e., \tpch Q5, \tpcds 19, 24, 46, 64, 68, 72, and 85). Although \rpt improves the execution time in most cases, the performance gap between the best and worst plans for a cyclic query is still huge. We propose that a robust execution engine in the future should adopt a hybrid approach to handle joins: executing the cyclic part of the query using worst-case optimal joins while processing the rest with \RPT.

\subsubsection{Case Study}
We present a case study on \job 2a to better illustrate the robustness guarantees brought by \rpt. \cref{fig:case} shows the best and worst left-deep plans for the baseline and \rpt (join phase only) with the size of each base or intermediate table marked. Without \rpt, the worst join order produces $179\times$ more intermediate tuples than the best. The worst join order suffers from the ``diamond problem'' described in \cite{birler2024robust}: small input $\rightarrow$ large intermediate result $\rightarrow$ small output, thus wasting computation. In comparison, the ratio of total intermediate results between the worst and best plans reduces to $1.2\times$ with \rpt. No matter what the join order is, the size of each intermediate table is bounded by the output size (i.e., 7.8k) and is monotonically increasing as the query executes.

\begin{figure}[t!]
    \centering
    \includegraphics[width=\linewidth]{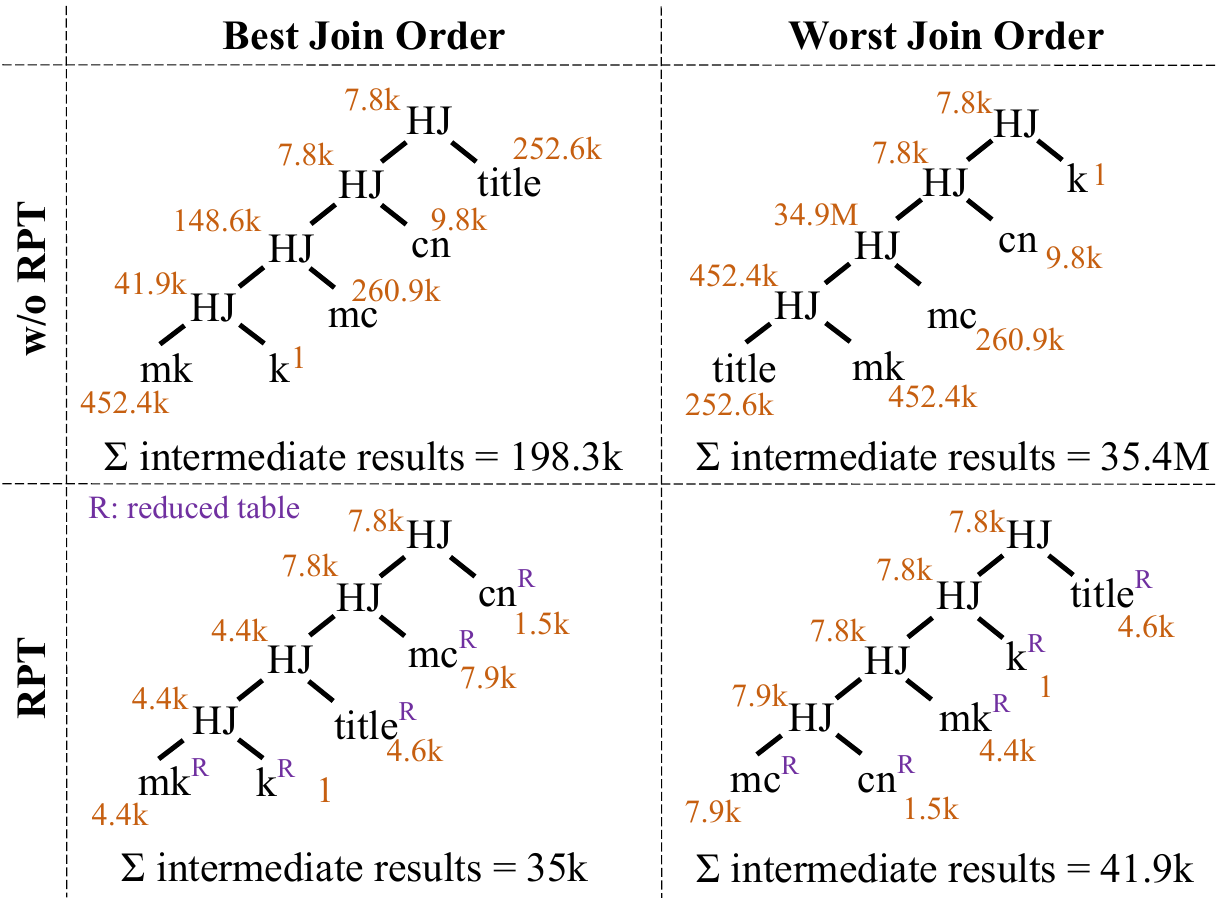}
    \caption{Case study on the robustness of JOB 2a \textnormal{-- We consider the reduced tables in \rpt (i.e., filtered base tables after the transfer phase) as intermediate results here.}}
    \label{fig:case}
\end{figure}

We also notice that even the best plan from the baseline must process a much larger ($\approx$$5\times$) intermediate result than any \rpt plan. This is because \rpt has a strict complexity supremacy over the baseline. \cref{fig:worst-supremacy} shows an example where query $R \Join S \Join T$ outputs nothing but any baseline plan (\NoPT) must process $N^2/2$ tuples, a quadratic explosion compared to \rpt plans. This example can be extended to create an exponential explosion as the number of tables increases. In comparison, \YannAlg guarantees that the size of $\sum$intermediate results for \rpt can be at most $n\times$ the output size, where $n$ is the number of joins.

\begin{figure}[t!]
    \centering
    \includegraphics[width=0.8\linewidth]{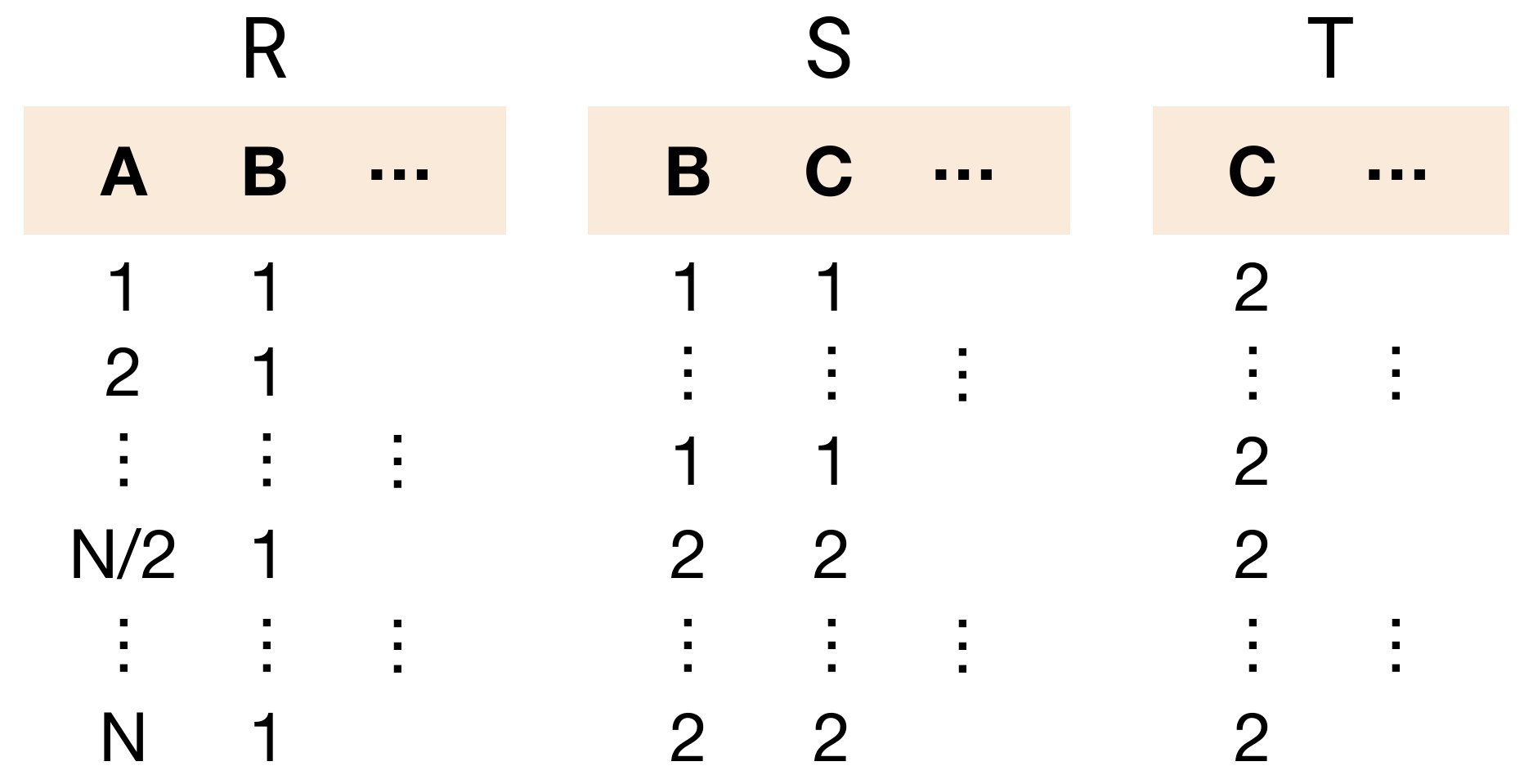}
    \caption{An example query (with an empty output) where any \NoPT plan must process $N^2/2$ tuples.}
    \label{fig:worst-supremacy}
\end{figure}

\subsection{Robustness of \TreeStruct}
\label{sec:eval:largestroot}

\begin{figure*}[t!]
    \centering
    \begin{subfigure}{0.39\linewidth}
        \includegraphics[width=\linewidth]{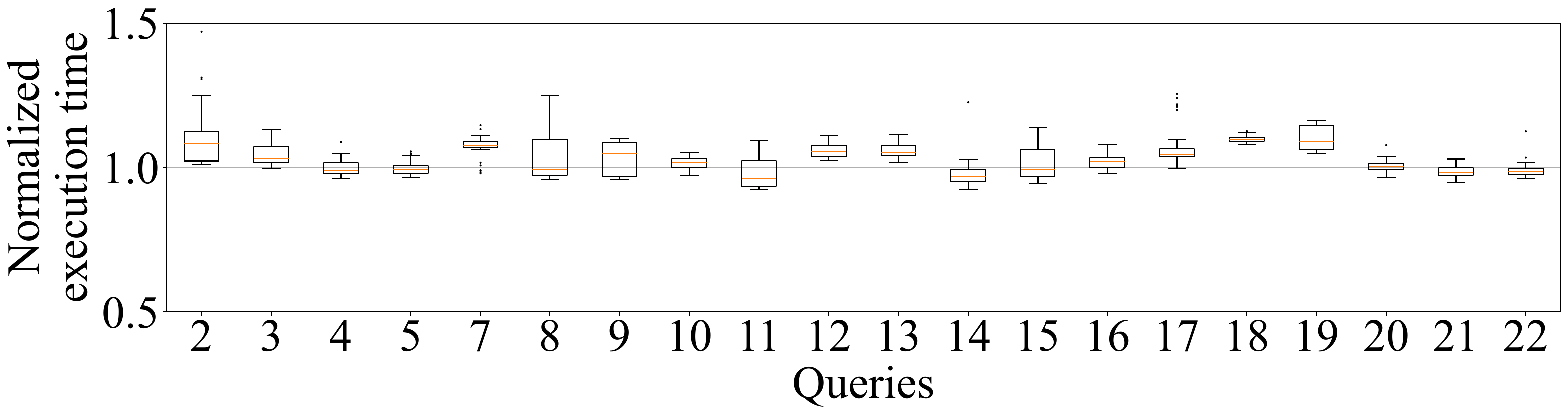}
        \caption{TPC-H}
        \label{fig:transfer-robustness-tpch}
    \end{subfigure}
    \begin{subfigure}{0.57\linewidth}
        \includegraphics[width=\linewidth]{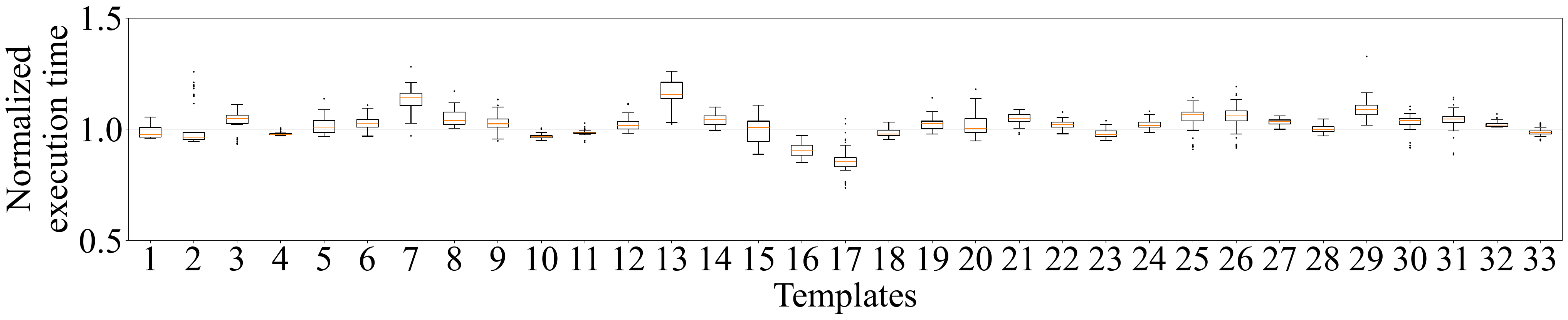}
        \caption{JOB}
        \label{fig:transfer-robustness-job}
    \end{subfigure}
    \caption{Distribution of the execution time of 50 random \TreeStruct transfer graphs for each query in \tpch and \job \textnormal{-- The box denotes 25- to 75-percentile (with the orange line as the median), while the horizontal lines denote min and max (excluding outliers).}}
    \label{fig:transfer-robustness}
\end{figure*}

We next zoom in to evaluate the robustness of the transfer phase in \rpt (i.e., the \TreeStruct algorithm). We modified \TreeStruct to generate 50 random join trees, but each of them still has the largest relation as the root. Specifically, we replaced the original \cref{step:choose} in \TreeStruct with ``Find an edge $e = \{R,S\} \in E(G_q)$ such that $R\in\mathcal R\setminus \mathcal R',S\in\mathcal R'$''. We fix the join order in each run to be the one produced by \duckdb's default optimizer. Other experiment settings follow those in \cref{sec:eval:end-to-end}.

\Cref{fig:transfer-robustness} shows the distribution of the end-to-end execution time with random \TreeStruct transfer graphs for each query in \tpch and \job. The 50 execution times for each query are normalized by the query time achieved using the unmodified \TreeStruct. We observe that the performance of the queries is robust against different transfer graphs (i.e., join trees for acyclic queries) as long as the algorithm keeps the largest relation at the root. Additionally, we notice that most boxes in \cref{fig:transfer-robustness} are above 1.0 (i.e., slower than the original \TreeStruct), indicating that the edge-picking heuristic used in \cref{step:choose} of \TreeStruct is effective in speeding up the transfer phase of \rpt.

\subsection{Robustness with Multi-Threaded Execution}
\label{sec:eval:multi-thread}

We repeat the left-deep experiments (i.e., \cref{fig:eval-left-deep}) in \cref{sec:eval:end-to-end} with 32 threads to investigate how multi-threaded execution affects the robustness of \rpt. As shown in \cref{fig:eval-32-threads}, \rpt still exhibits outstanding query robustness with orders-of-magnitude improvement over the baseline on the Robustness Factor (\robustmetric). Compared to \cref{fig:eval-left-deep}, we notice that the variance of the execution times across different left-deep plans increases for some of the queries when switching from single-threaded to multi-threaded execution. This is because some random left-deep plans placed a relatively small (reduced) table on the probe side of the long (probing) pipeline, which does not have enough data chunks to distribute across 32 parallel threads to fully utilize the computation. The problem is orthogonal to the robustness guarantees offered by \rpt.

\begin{figure*}[t!]
    \centering
    \begin{subfigure}{0.37\linewidth}
        \includegraphics[width=\linewidth]{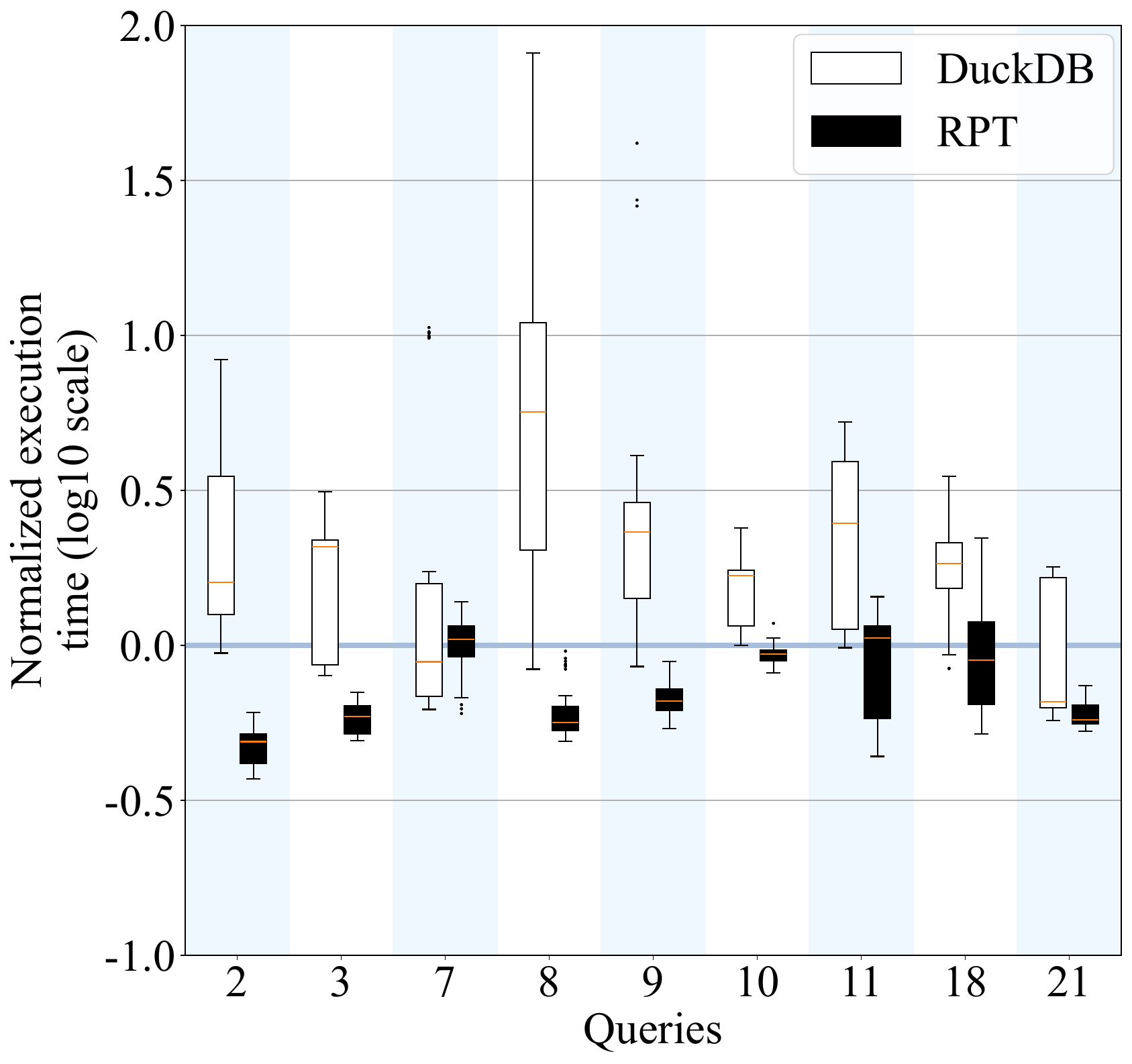}
        \caption{\tpch (32 threads)}
        \label{fig:eval-32-threads-tpch}
    \end{subfigure}
    \begin{subfigure}{0.62\linewidth}
        \includegraphics[width=\linewidth]{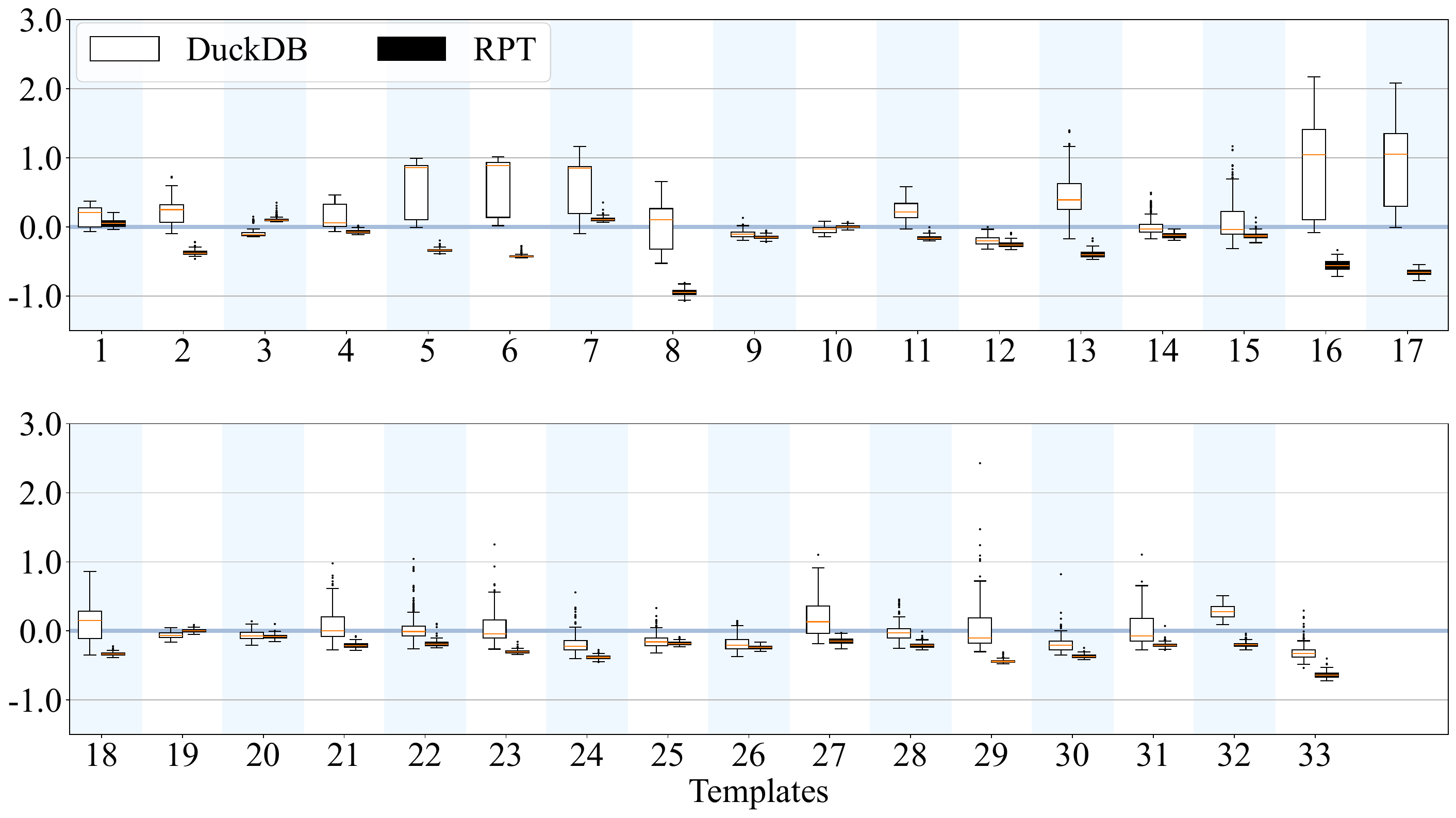}
        \caption{\job (32 threads)}
        \label{fig:eval-32-threads-job}
    \end{subfigure}
    \caption{Distribution of multi-threaded execution time of random left-deep plans for each acyclic query in \tpch and \job \textnormal{-- Normalized by the execution time of default \duckdb. The figure is log-scaled. The box denotes 25- to 75-percentile (with the orange line as the median), while the horizontal lines denote min and max (excluding outliers).}}
    \label{fig:eval-32-threads}
\end{figure*}

\subsection{Performance with Data On Disk}
\label{sec:eval:on-disk}

We extend our evaluation to the case where (1) the base tables reside on disk (labeled as ``on-disk''), and (2) some intermediate results of \rpt do not fit in memory (labeled as ``+spill''). The intermediate results of \rpt refer to the materialized data chunks that contain the remaining tuples after the forward pass in the semi-join phase. We evaluate the optimizer's plan for \duckdb and \rpt for each query in \tpch and \job. For ``+spill'', we configure the available memory to be $\approx50\%$ of \rpt's peak memory usage for each query and make sure that the spilled data reside on disk. As shown in~\cref{fig:on-disk}, \rpt still archives an average (geometric mean) speedup of $1.3\times$ and $1.5\times$ over the default \duckdb for the ``on-disk'' and ``on-disk+spill'' cases, respectively. Although the backward pass in the semi-join phase of \rpt incurs repeated data accesses, the overhead is small. This is because (1) the volume of the materialized data after the forward pass is small due to the selective semi-join filters, and (2) the backward-pass scans on the materialized data are sequential.

\begin{figure*}[t!]
    \centering
    \begin{subfigure}{0.46\linewidth}
        \includegraphics[width=\linewidth]{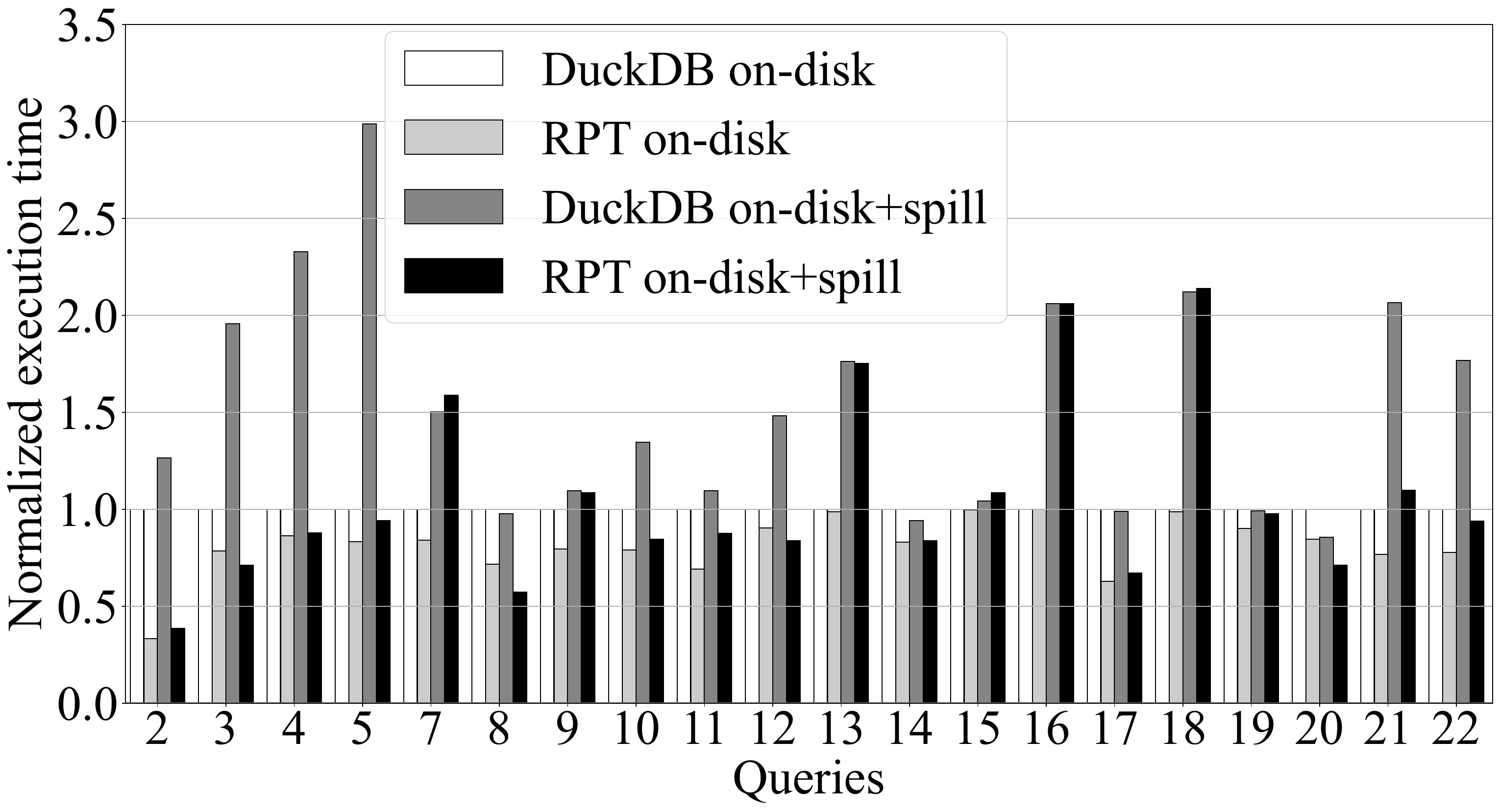}
        \caption{\tpch}
    \end{subfigure}
    \begin{subfigure}{0.51\linewidth}
        \includegraphics[width=\linewidth]{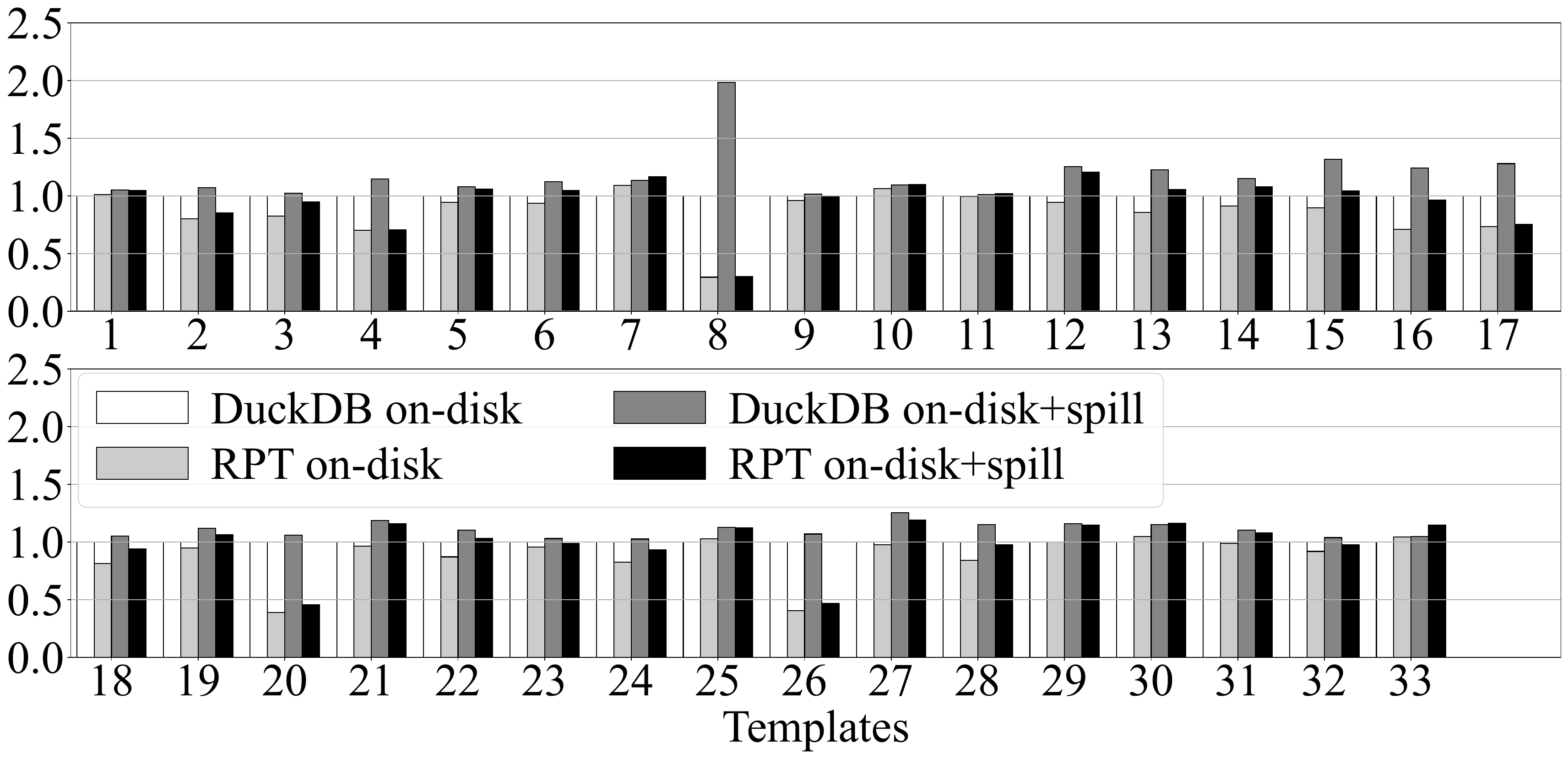}
        \caption{\job}
    \end{subfigure}
    \caption{Comparison of the execution time of \duckdb and \rpt (with optimizer's plan) for each query in \tpch and \job when the base tables reside on disk (on-disk) + the intermediate results do not fit in memory (+spill) \textnormal{-- Normalized by the execution time of default \duckdb with base tables on disk.}}
    \label{fig:on-disk}
\end{figure*}

\subsection{Performance of Bloom Filters}
\label{sec:eval:bloomfilter}

\begin{figure}[t!]
    \centering
    \includegraphics[width=\linewidth]{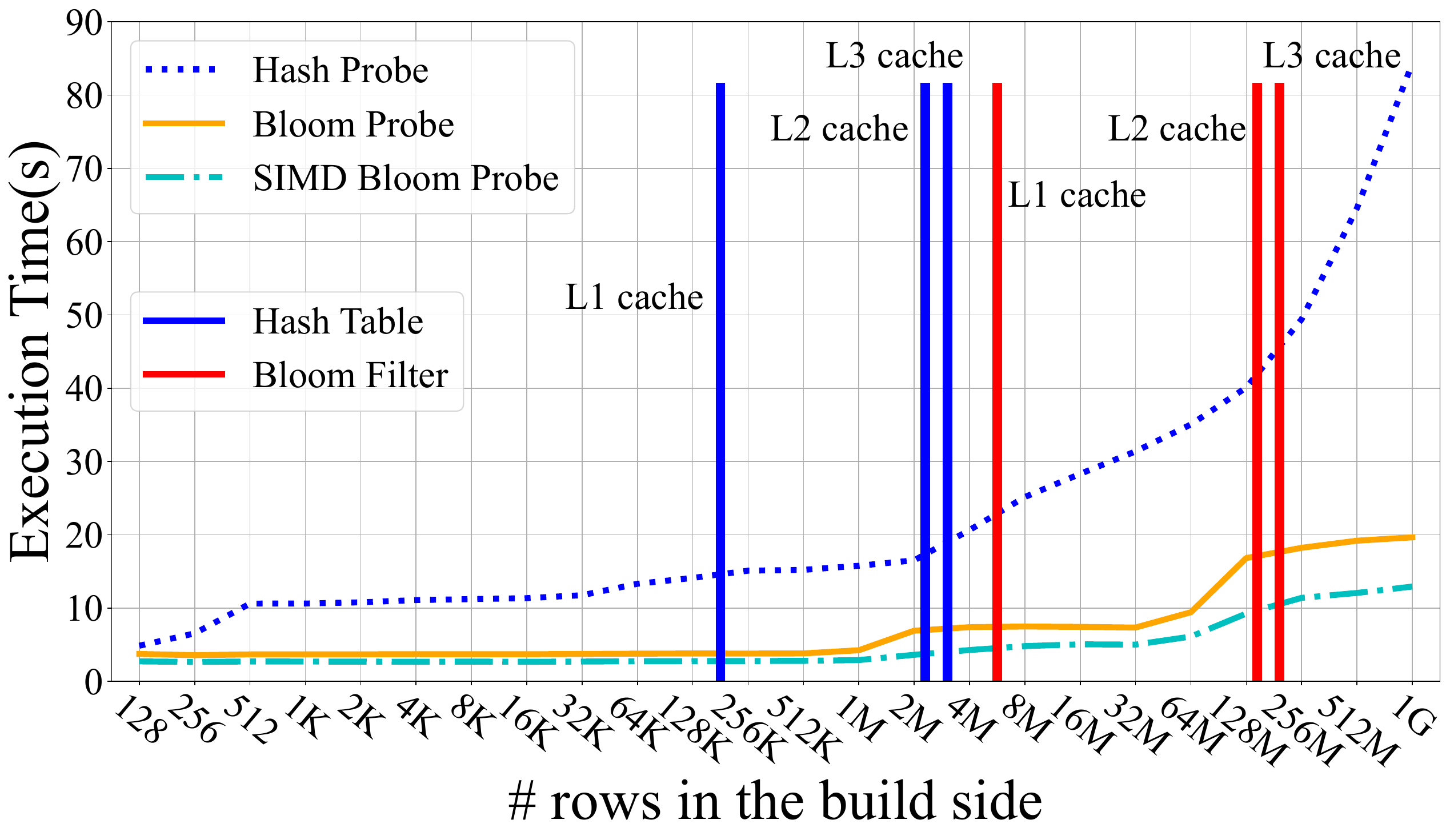}
    \caption{Microbenchmark on Bloom Probe vs. Hash Probe \textnormal{-- We fix the probe side to 1 billion entries while varying the size of the build side on the x-axis.}}
    \label{fig:bf-probe}
\end{figure}

We presented in \cref{fig:eval-left-deep} that \rpt improves the overall query performance by $\approx$$1.5\times$ besides robustness, and our performance breakdown shows that the \BF operations in the transfer phase of \rpt account for on average $28\%$, $12\%$, and $46\%$ of the total execution time in \tpch, \job, and \tpcds, respectively. In this section, we evaluate the performance gap between \BF probes and hash table probes through a microbenchmark. We create a synthetic dataset with two single-column tables. We fix the size of the probe-side table to 1 billion rows while varying the size of the build-side table. The integer values of each column are uniformly distributed between $0$ and $2^{30}$.

\cref{fig:bf-probe} reports the execution time of performing 1 billion probes on hash tables or \BFs with different sizes. We use \duckdb's vectorized hash table implementation for hash probes and our modified version of Arrow's blocked \BF for Bloom probes. The blue (red) vertical lines denote the points where the size of the hash table (\BF) exceeds L1, L2, and L3 caches. We observe that the SIMD version of Bloom probes outperforms vectorized hash probes by $2-7\times$. The performance gap grows as the size of the hash table / \BFs increases, indicating a potentially greater performance advantage of \rpt on larger datasets.

\section{Related Work}
\label{sec:related}

\subsection{Sideways Information Passing (SIP)}

Sideways Information Passing (SIP) refers to techniques that optimize join operations by transmitting predicate information to the target table to facilitate tuple pre-filtering in a database. Existing SIP techniques can be categorized as Bloom join~\cite{bloomjoin, distributedbloomjoin, optimizingdistributedbloomjoin, zhu2017LIP} and semi-join reduction~\cite{usingsemi}. In Bloom join, a \BF is generated on the build side of a hash join and passed to the probe side to filter tuples before accessing the hash table. Semi-join reduction, on the other hand, applies a semi-join operation to pre-filter tuples before conducting the actual hash join.

Lookahead Information Passing (LIP)~\cite{zhu2017LIP} can be considered a special case of \RPT with star schema. LIP constructs \BFs for each dimension table and uses them to pre-filter the large fact table before performing the joins. LIP focuses on techniques to reorder the \BFs dynamically and adaptively to reduce the computational overhead of the SIP process. These techniques are orthogonal to our work and can also be applied to \rpt.

In contrast to the existing SIP approaches, \RPT provides strong theoretical guarantees on query robustness by applying pre-filtering (with \BFs) systematically based on the \YannAlg, rather than focusing on particular joins locally.

\subsection{Robust Query Processing}

Previous studies~\cite{2019tutorial_robust, robustoptimization} offer a comprehensive survey of robust query optimization methods. These methods target mitigating the impact of inaccurate cardinality estimations, and they can classified into two categories: robust plans~\cite{2002LEC, 2005RCE, 2007plan_diagram, 2008strict_plan_diagram, Abhirama2010BDSH, AlyoubiHW15, Wolf2018RobustMetric} and re-optimization~\cite{1998reopt, 1999reopt_shared_nothing, 2000eddies, 2004pop, 2007pop_parallel, 2016planbouquets, Perron19, 2023reopt_zhao, justen2024polar}.

Robust plans, such as Least Expected Cost~\cite{2002LEC, 2005RCE}, estimate the distributions of the filter/join selectivities. In contrast, the Cost-Greedy approach reduces the search space by low-cardinality approximations to favor the choices of performance-stable plans~\cite{2007plan_diagram}. Similarly, SEER applies low-cardinality approximations to accommodate arbitrary estimation errors~\cite{2008strict_plan_diagram}, while~\cite{Abhirama2010BDSH, AlyoubiHW15, Wolf2018RobustMetric} propose metrics to quantify the robustness of execution plans during query optimization.

ReOpt~\cite{1998reopt, 1999reopt_shared_nothing} introduces mid-query re-optimization, where the query engine detects cardinality estimation errors at execution time and re-invokes the optimizer to refine the remaining query plan. Eddies routes data tuples adaptively through a network of query operators during execution~\cite{2000eddies}. The POP algorithm introduces the concept of a "validity range" for selected plans, triggering re-optimization when the actual parameter values fall outside this range~\cite{2004pop, 2007pop_parallel}. Plan Bouquet eliminates the need for estimating operator selectivities by identifying a set of "switchable plans" that can accommodate runtime selectivity variations~\cite{2016planbouquets}. Experiments in~\cite{Perron19} demonstrate that query re-optimization achieves excellent performance on PostgreSQL with the Join Order Benchmark. QuerySplit~\cite{2023reopt_zhao} introduces a novel re-optimization technique to minimize the probability of explosive intermediate results during re-optimization. POLAR~\cite{justen2024polar} avoids intertwining query optimization and execution by inserting a multiplexer operator into the physical plan.

A few recent works~\cite{birler2024robust, treetrackerjoin} developed algorithms fundamentally equivalent to the \YannAlg. They focused on avoiding performance regression when applying semi-join reductions even on worst-case input (i.e., input where pre-filtering is ineffective).

Compared to \rpt, most existing robust query processing approaches lack theoretical guarantees on join-order robustness. Nevertheless, some of the techniques related to physical operator selections and operators beyond join are orthogonal to \rpt and can complement our approach to boost query performance further.

\subsection{Worst-Case Optimal Join}

While the \YannAlg performs acyclic joins in optimal time (linear in the input and output size), answering general cyclic queries in polynomial time in terms of input, output, and query size is impossible unless $\textsf{P}=\textsf{NP}$.

A tractable extension for the cyclic case is near-acyclic queries, whose intricacy can be measured by different notions of width, such as treewidth~\cite{ROBERTSON1986309}, 
query width~\cite{chandra1997}, hypertree width~\cite{gottlob1999}, and submodular width~\cite{Marx10}. Generally speaking, a query with a width of $k$ has an upper bound $O(N^k+OUT)$ on the time complexity.
The hierarchy of bounds is summarized in a survey~\cite{suciu2023} and a recent result~\cite{lpnorm2024}.

Worst-case optimal join~(WCOJ) algorithms are developed to guarantee the above bounds on the running time. Binary joins are ubiquitous in relational DBMS but fail short on certain database instances compared to WCOJ algorithms. NPRR~\cite{nprr12} is the first algorithm that achieves the AGM bound~\cite{agm08}, and then an existing algorithm LFTJ is also proved to be running in the AGM bound~\cite{2014leapfrog}. These algorithms are unified as the Generic Join~\cite{ngo2014SIGMOD,ngo2018}, which determines one variable at a time using tries. The PANDA algorithm~\cite{panda2017,panda2024} eliminates one inequality at a time using horizontal partitioning and achieves the polymatroid bound. Variants of WCOJ algorithms have been adopted in distributed query processing~\cite{chu2015theory, koutris2016worst, ammar2018distributed}, graph  processing~\cite{zhang2014evaluating, aberger2017emptyheaded, ammar2018distributed, hogan2019worst, mhedhbi2019optimizing, zhu2019hymj}, and general-purpose query processing~\cite{aref2015design, aberger2018levelheaded,2020hashtrie}. WCOJ algorithms are becoming practical as their performance surpasses traditional binary joins for certain queries~\cite{freejoin2023}.

Unlike WCOJ algorithms, \RPT only provides theoretical guarantees on the runtime for acyclic queries. However, it is strictly better than WCOJ algorithms because it bounds the runtime to the instance-specific output size rather than a more generalized upper bound.

\section{Conclusion}
\label{Conclusion}

We proposed the \RPT algorithm that is provably robust against arbitrary join orders of an acyclic query. Our evaluation in \duckdb shows that \rpt ensures a small variation in the execution time between random join orders for acyclic queries while improving their end-to-end performance at the same time. We hope that our results advance the state-of-the-art of robust SQL analytics and will simplify the join optimization logic in future query optimizers.

\bibliography{ref}

\newpage

\appendix

\section{RPT Performance with Optimizer's Plan}

In Appendix A, we present the full performance results of \RPT using the optimizer's plan, compared to our baseline methods: DuckDB, Bloom Join, and \PT. 

\Cref{fig:tpch-perf} shows the execution time with the optimizer's plan for each query in \tpch. Note that we exclude Q1 and Q6, as they only involve scanning and filtering a single table. On average (geometric mean), \RPT outperforms vanilla DuckDB by $1.53\times$ and Bloom Join by $1.33\times$. Additionally, \RPT achieves the same performance as \PT. This is because \tpch queries are relatively simple, and the transfer scheduling of \PT and \RPT does not differ significantly.

\Cref{fig:job-perf} displays the execution time with the optimizer’s plan for one result from each of the 33 query templates in the \job. On average (geometric mean), \RPT outperforms vanilla DuckDB by $1.46\times$ and Bloom Join by $1.29\times$, while matching the performance of \PT.

\Cref{fig:tpcds-perf} presents the execution time with the optimizer's plan for each query in \tpcds. On average (geometric mean), \RPT outperforms DuckDB by $1.56\times$, Bloom Join by $1.48\times$, and \PT by $1.23\times$. However, for certain queries (e.g., Q16, Q61, and Q69), \RPT performs poorly compared to vanilla DuckDB and Bloom Join. This is due to the result being empty for these queries, causing \RPT to scan more tables than vanilla DuckDB and Bloom Join, as query execution stops upon encountering empty intermediate results.

In \Cref{fig:dsb-perf}, we show the execution time with the optimizer's plan for each query in \dsb.  On average (geometric mean), \RPT outperforms vanilla DuckDB by $1.54\times$, Bloom Join by $1.45\times$, and \PT by $1.23\times$. Similar to \tpcds, for some specific queries, \RPT exhibits poor performance for certain queries due to empty intermediate results, resulting in additional table scans compared to vanilla DuckDB and Bloom Join.

\begin{figure*}[t!]
    \centering
    \includegraphics[width=\linewidth]{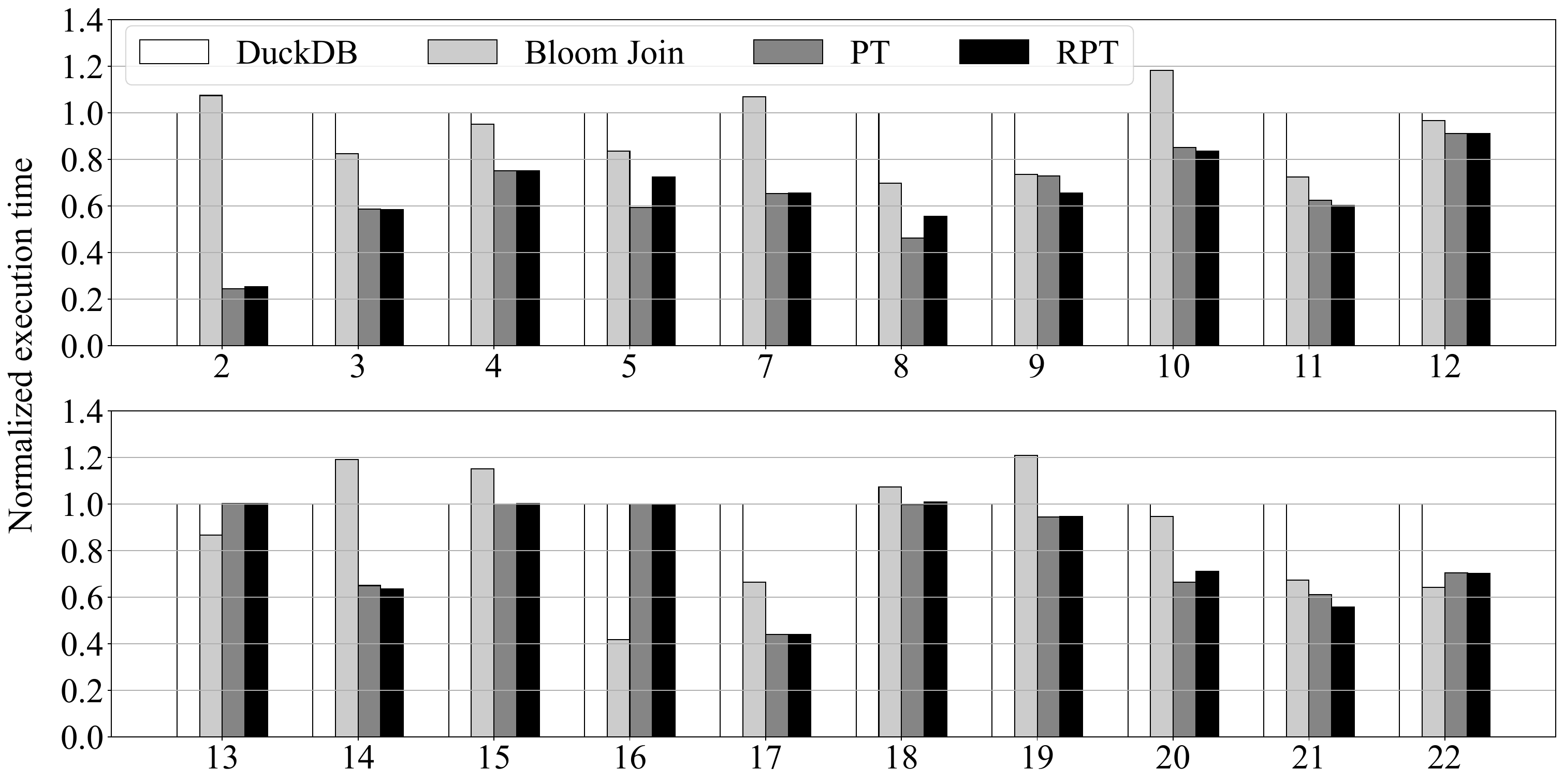}
    \caption{The execution time with optimizer's plans for each query in \tpch \textnormal{-- Normalized by the execution time of default \duckdb. We omit Q1 and Q6 as they are only the table scan and filtering.}}
    \label{fig:tpch-perf}
\end{figure*}

\begin{figure*}[t!]
    \centering
    \includegraphics[width=\linewidth]{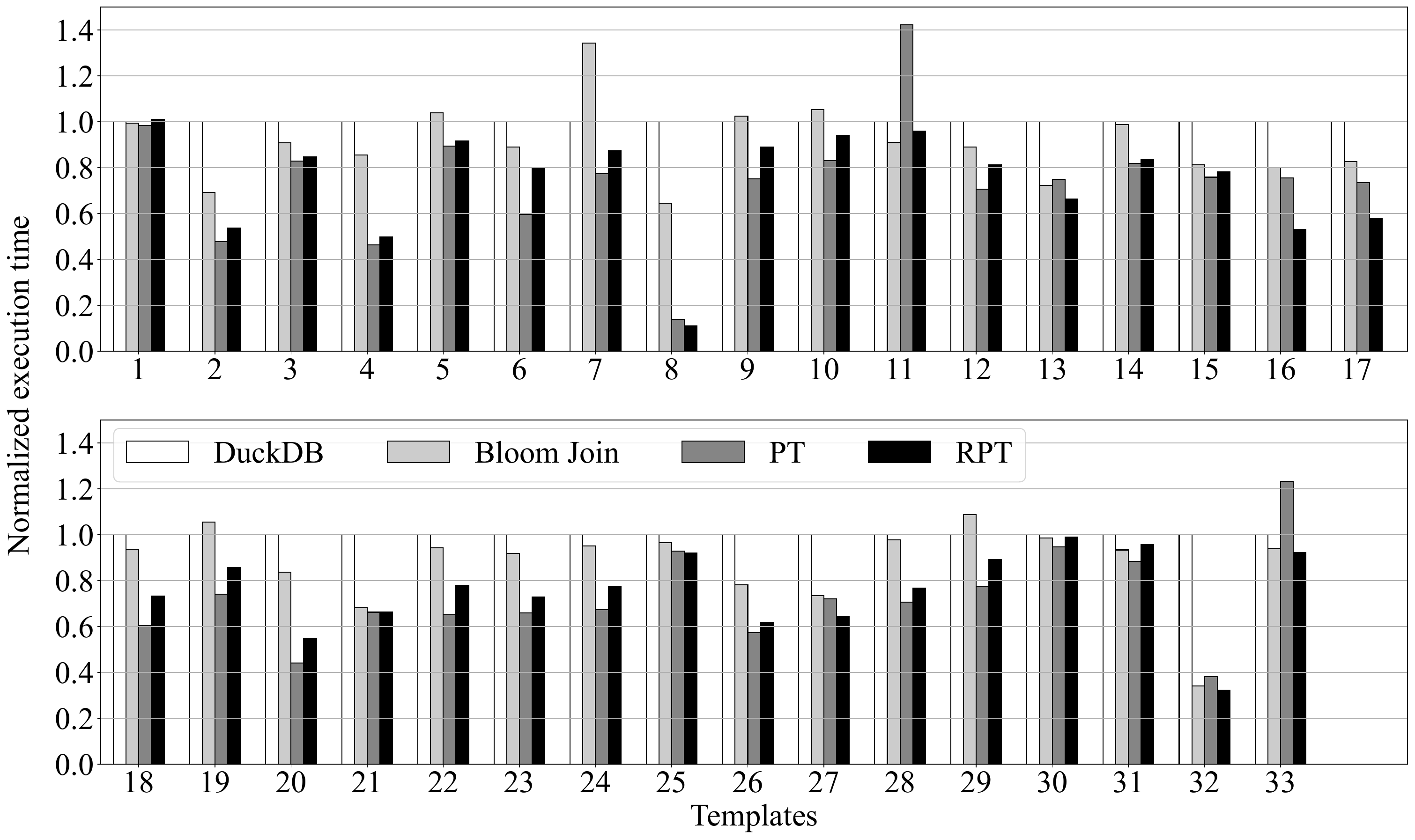}
    \caption{The execution time with optimizer's plans for each query in \job \textnormal{-- Normalized by the execution time of default \duckdb.}}
    \label{fig:job-perf}
\end{figure*}

\begin{figure*}[t!]
    \centering
    \includegraphics[width=\linewidth]{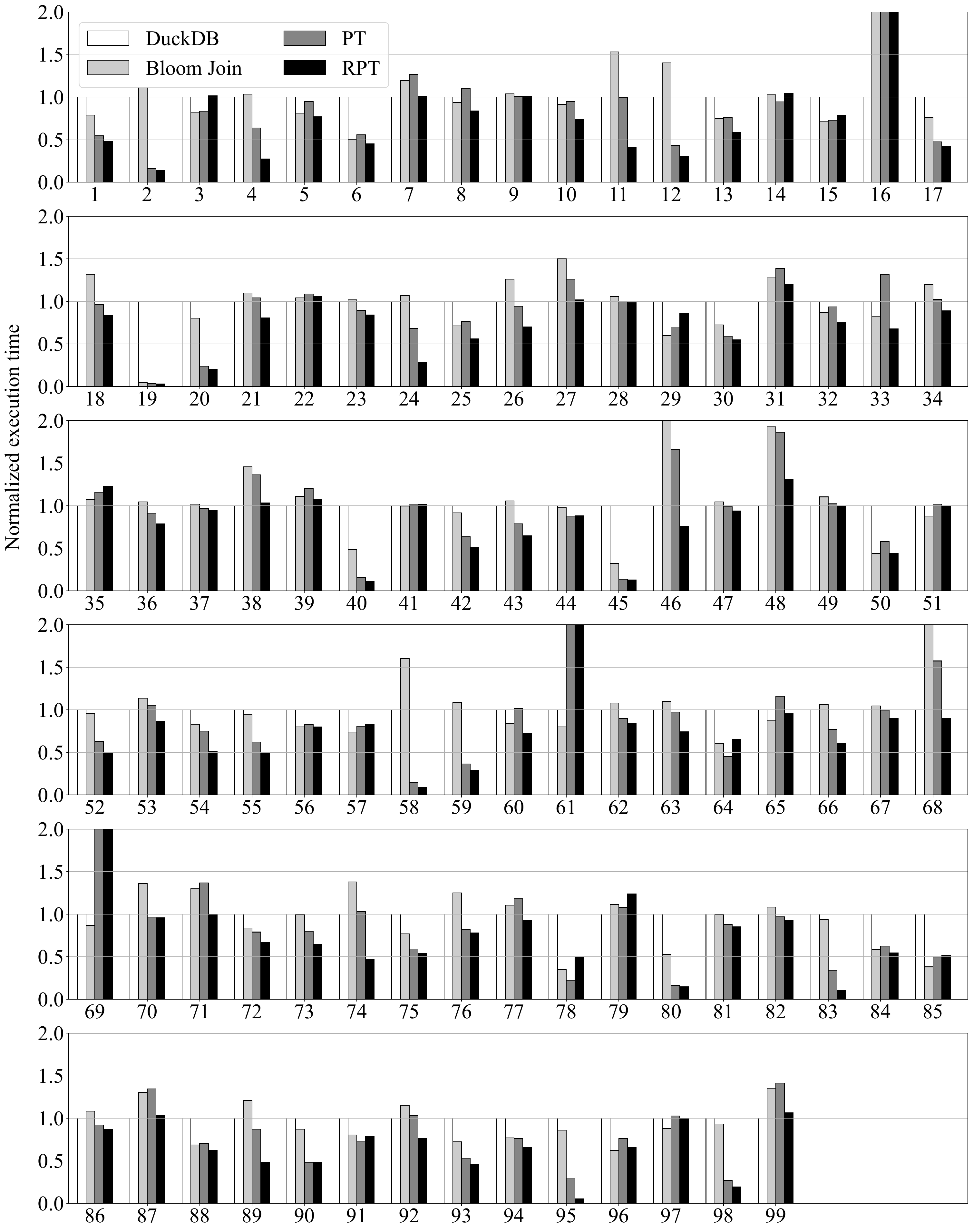}
    \caption{The execution time with optimizer's plans for each query in \tpcds \textnormal{-- Normalized by the execution time of default \duckdb.}}
    \label{fig:tpcds-perf}
\end{figure*}

\begin{figure*}[t!]
    \centering
    \includegraphics[width=\linewidth]{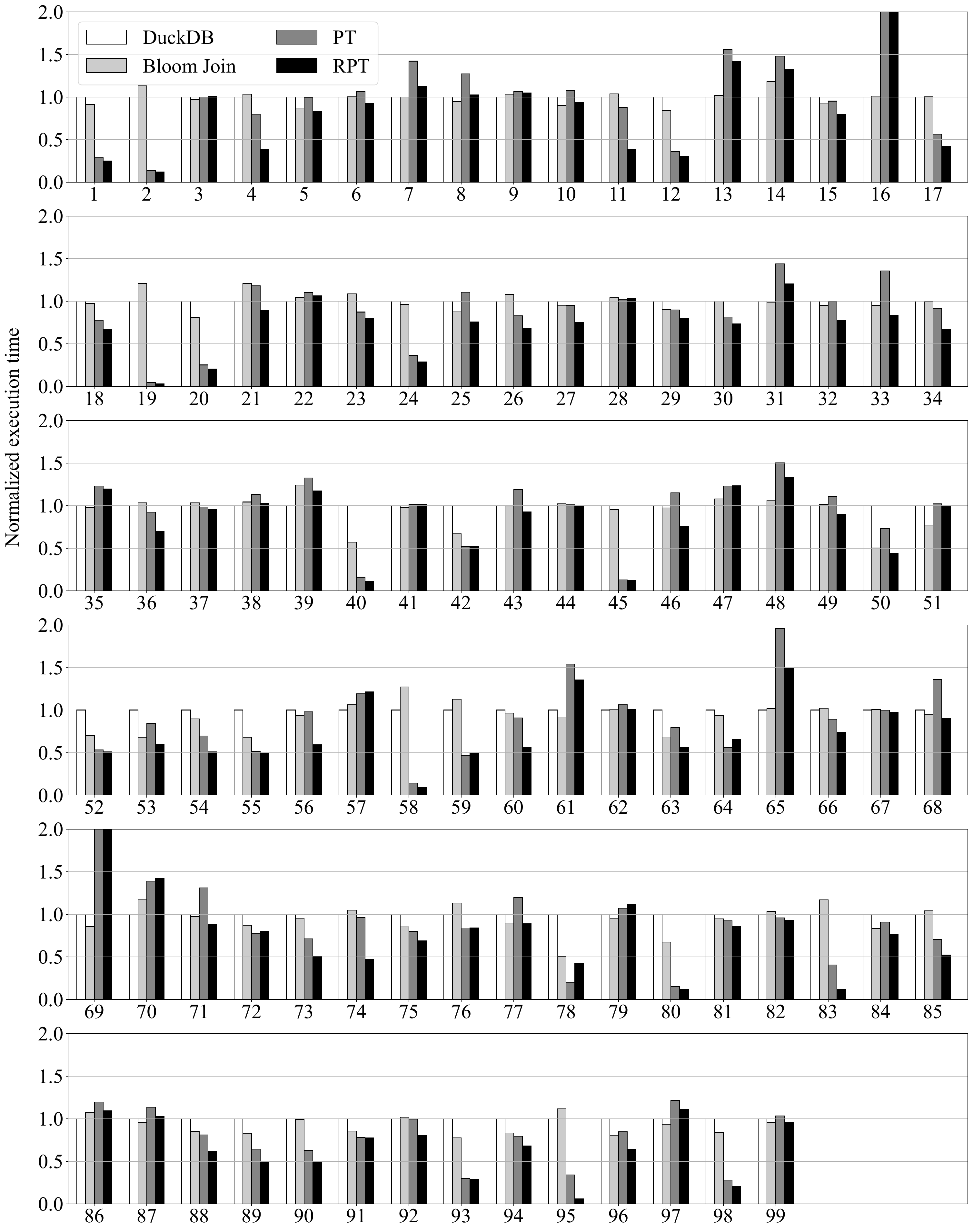}
    \caption{The execution time with optimizer's plans for each query in \dsb \textnormal{-- Normalized by the execution time of default \duckdb.}}
    \label{fig:dsb-perf}
\end{figure*}
\section{Additional Robustness Results (Left Deep)}

In Appendix B, we present the distribution of execution times with random left-deep plans for each query of \RPT, compared to our baseline methods: DuckDB, Bloom Join, and \PT. These results are shown in \Cref{fig:tpch-left} (\tpch), \Cref{fig:job-left} (\job), \Cref{fig:tpcds-left-a} (\tpcds query 1-52), \Cref{fig:tpcds-left-b} (\tpcds query 53-99), \Cref{fig:dsb-left-a} (\dsb query 1-52) and \Cref{fig:dsb-left-b} (\dsb query 53-99).

For most acyclic queries, both \RPT and \PT outperform vanilla DuckDB and Bloom Join in terms of robustness. However, for specific acyclic queries (e.g., JOB 32a and 32b, TPC-DS 54 and 83, DSB 54 and 83), \PT is also not robust. This is because the \StoL transfer algorithm used by \PT lacks a theoretical guarantee.

Even for cyclic queries, \RPT can improve robustness to some extent. However, due to the absence of the theoretical guarantee for cyclic queries, \RPT fails to constrain their maximum execution time.

\begin{figure*}[t!]
    \centering
    \begin{subfigure}{0.44\linewidth}
        \includegraphics[width=\linewidth]{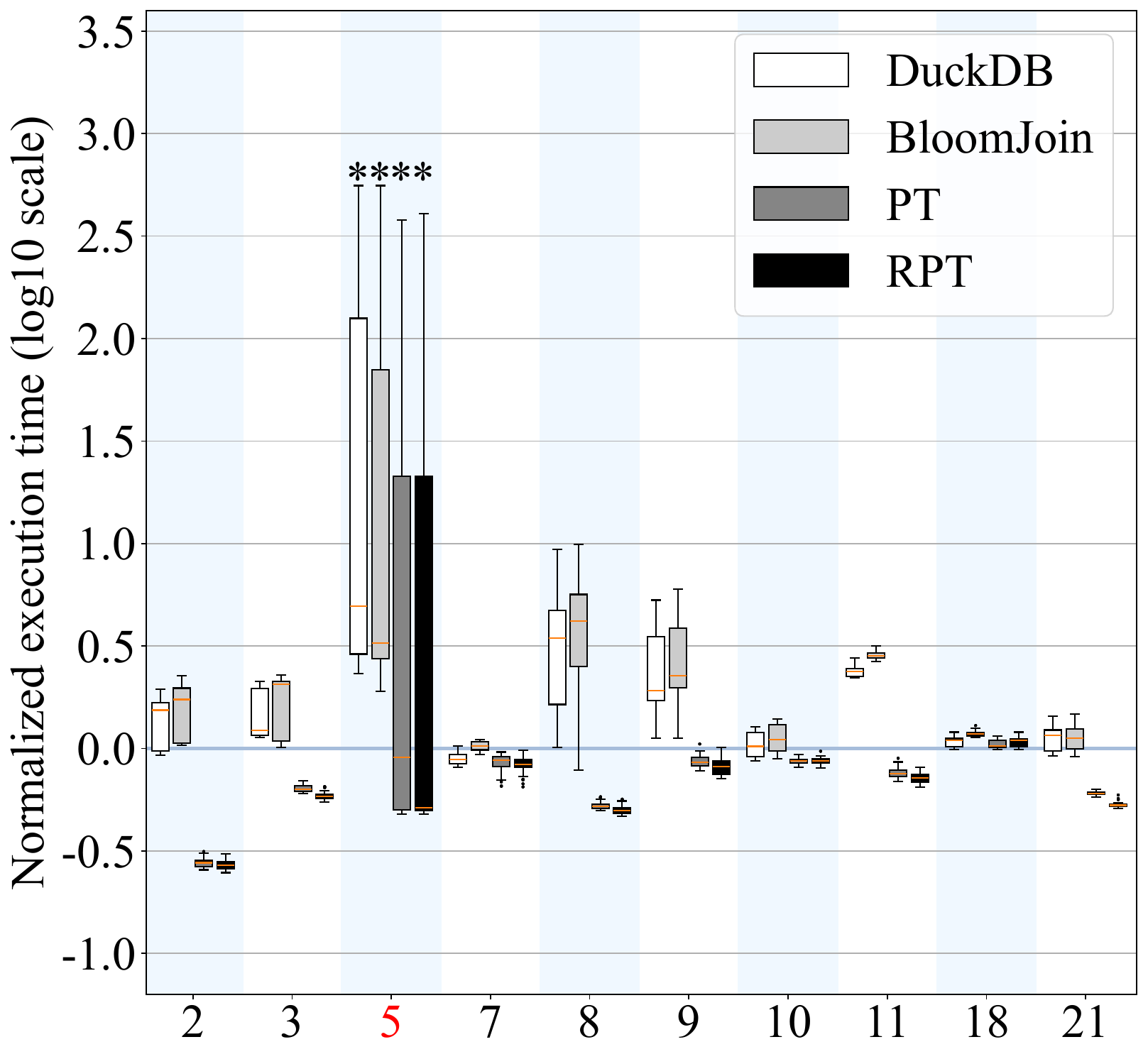}
        \caption{Only left deep}
        \label{fig:tpch-left}
    \end{subfigure}
    \begin{subfigure}{0.44\linewidth}
        \includegraphics[width=\linewidth]{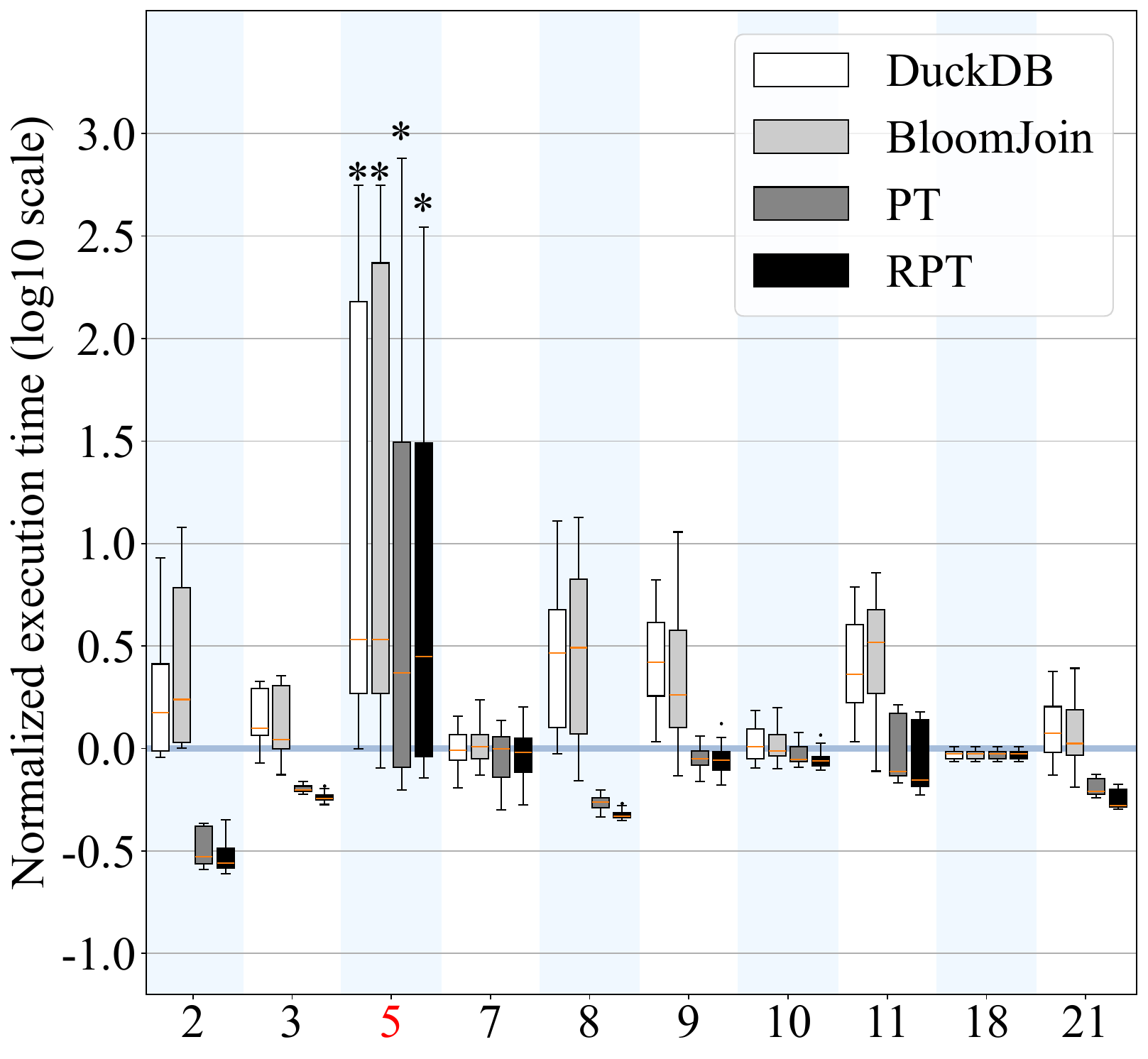}
        \caption{Bushy}
        \label{fig:tpch-bushy}
    \end{subfigure}
    \centering
    \caption{The distribution of execution time with random left deep plans for each query in \tpch \textnormal{-- Normalized by the execution time of default \duckdb. The figure is log-scaled. The box denotes 25- to 75-percentile (with the orange line as the median), while the horizontal lines denote min and max (excluding outliers). `*' indicates timeouts. Cyclic queries are in red.}}
\end{figure*}

\begin{figure*}[t!]
    \centering
    \includegraphics[width=0.9\linewidth]{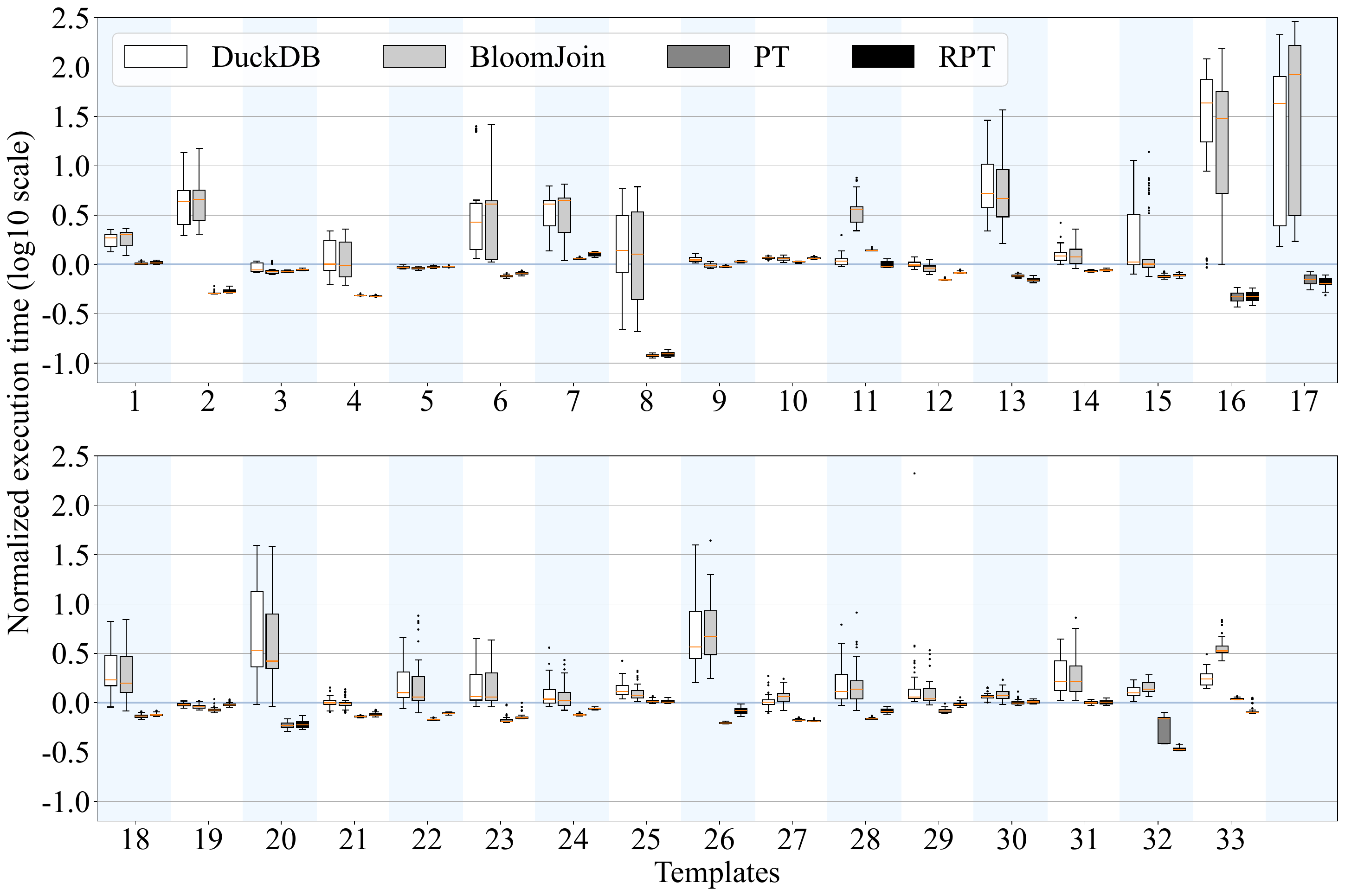}
    \caption{The distribution of execution time with random left deep plans for each template in \job \textnormal{-- Normalized by the execution time of default \duckdb. The figure is log-scaled. The box denotes 25- to 75-percentile (with the orange line as the median), while the horizontal lines denote min and max (excluding outliers).}}
    \label{fig:job-left}
\end{figure*}

\begin{figure*}[t!]
    \centering
    \includegraphics[width=\linewidth]{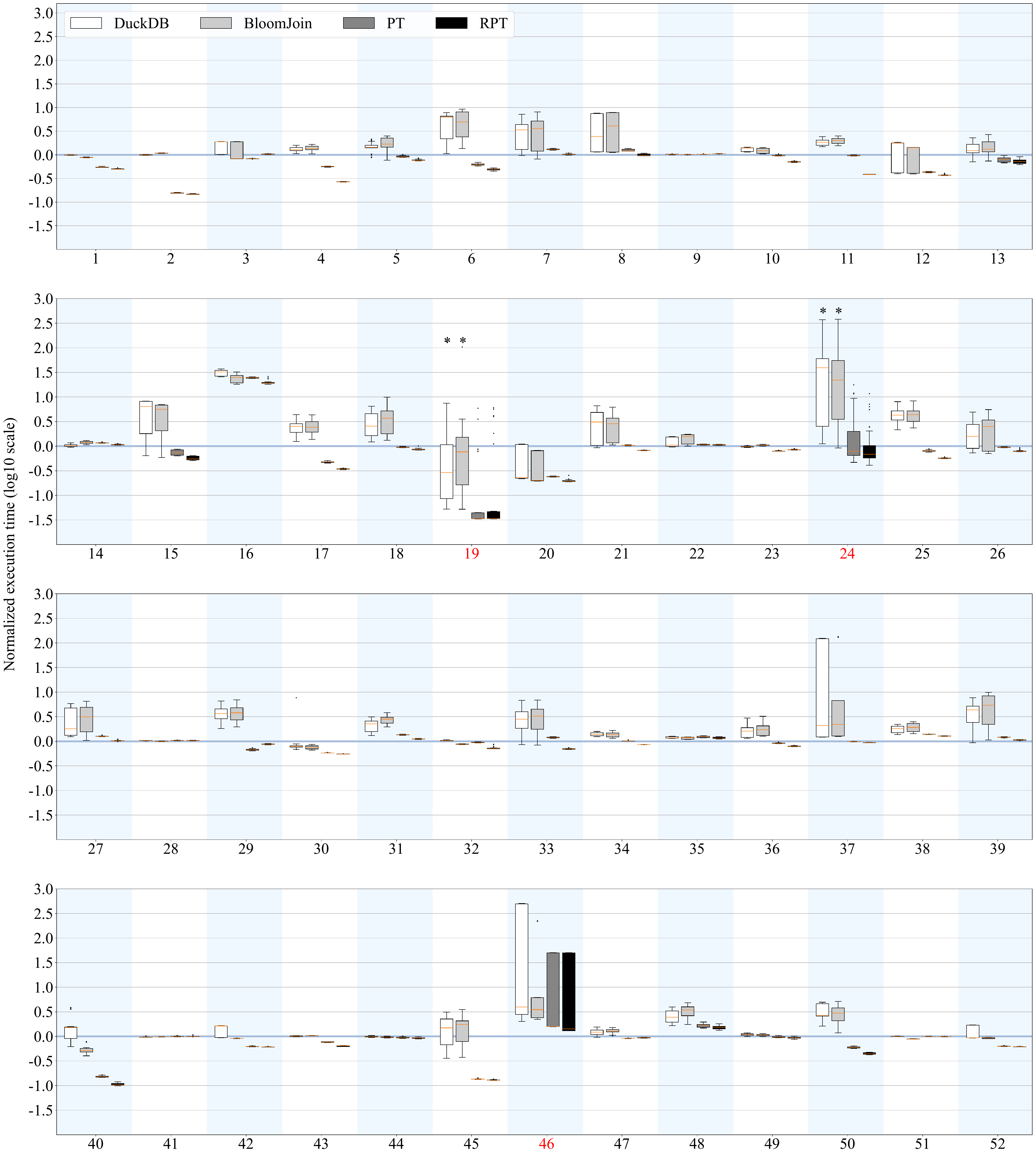}
    \caption{The distribution of execution time with random left deep plans for each query (1 - 52) in \tpcds \textnormal{-- Normalized by the execution time of default \duckdb. The figure is log-scaled. The box denotes 25- to 75-percentile (with the orange line as the median), while the horizontal lines denote min and max (excluding outliers). `*' indicates timeouts. Cyclic queries are in red.}}
    \label{fig:tpcds-left-a}
\end{figure*}

\begin{figure*}[t!]
    \centering
    \includegraphics[width=\linewidth]{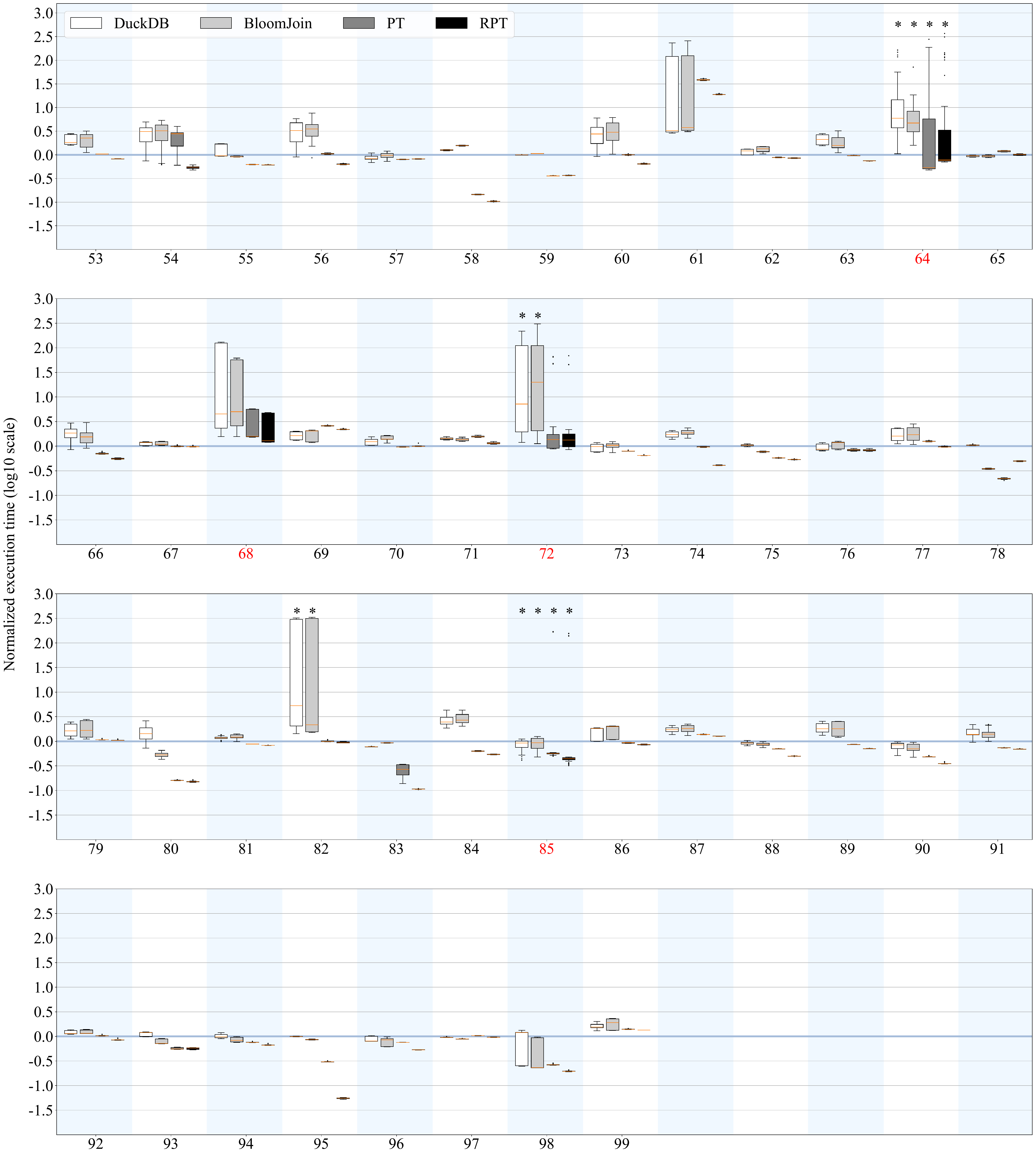}
    \caption{The distribution of execution time with random left deep plans for each query (53 - 99) in \tpcds \textnormal{-- Normalized by the execution time of default \duckdb. The figure is log-scaled. The box denotes 25- to 75-percentile (with the orange line as the median), while the horizontal lines denote min and max (excluding outliers). `*' indicates timeouts. Cyclic queries are in red.}}
    \label{fig:tpcds-left-b}
\end{figure*}

\begin{figure*}[t!]
    \centering
    \includegraphics[width=\linewidth]{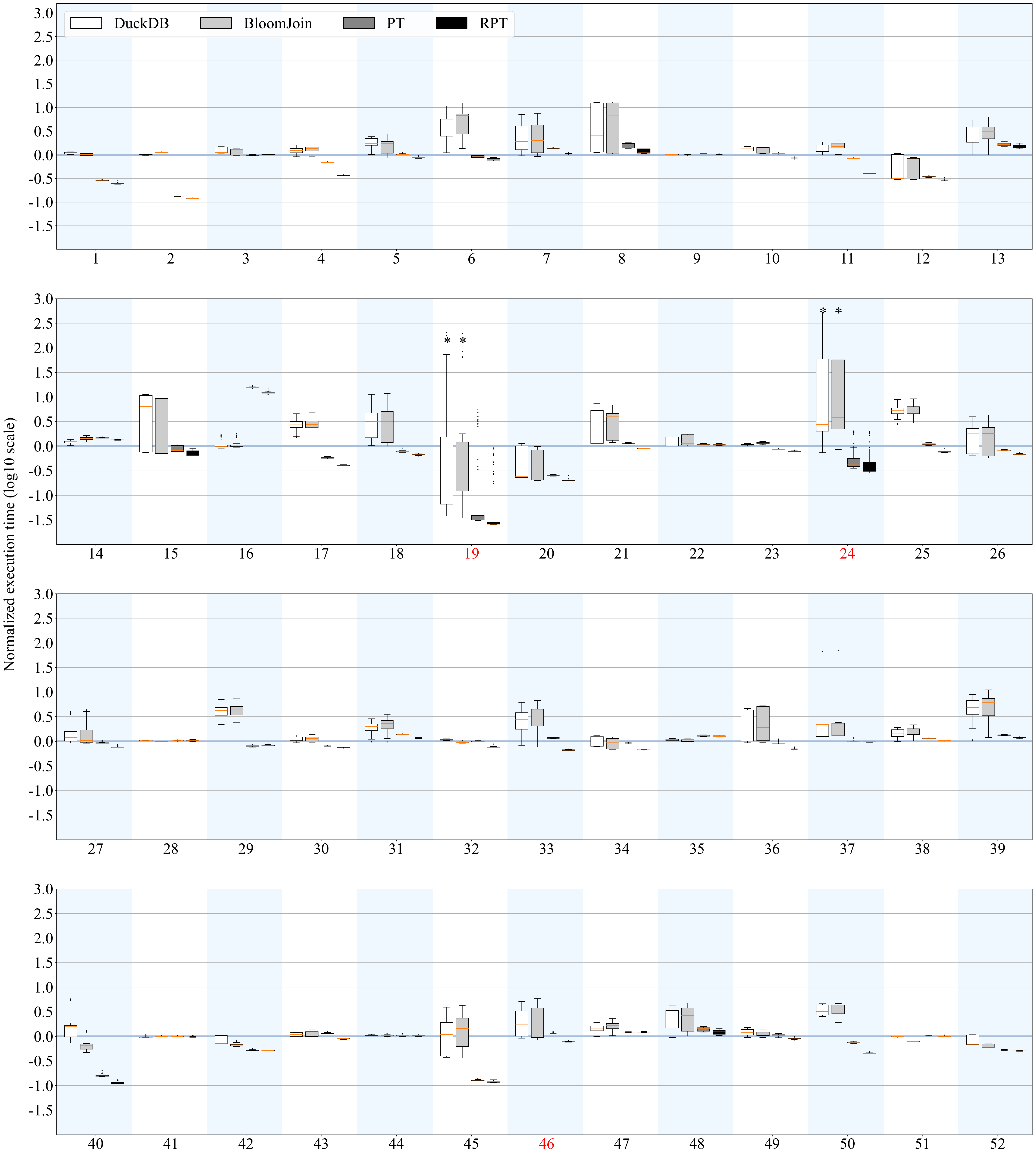}
    \caption{The distribution of execution time with random left deep plans for each query (1 - 52) in \dsb \textnormal{-- Normalized by the execution time of default \duckdb. The figure is log-scaled. The box denotes 25- to 75-percentile (with the orange line as the median), while the horizontal lines denote min and max (excluding outliers). `*' indicates timeouts. Cyclic queries are in red.}}
    \label{fig:dsb-left-a}
\end{figure*}

\begin{figure*}[t!]
    \centering
    \includegraphics[width=\linewidth]{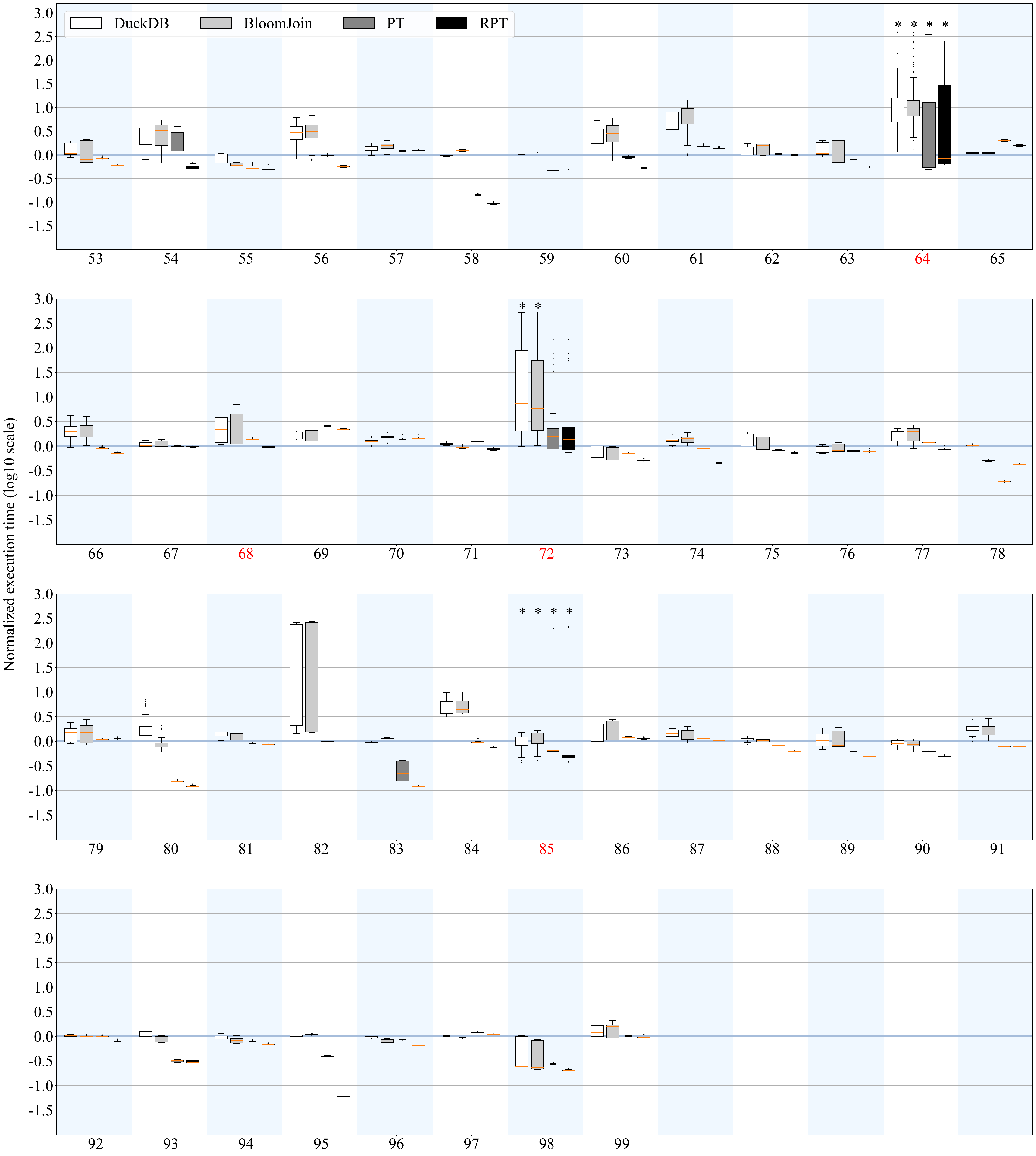}
    \caption{The distribution of execution time with random left deep plans for each query (53 - 99) in \dsb \textnormal{-- Normalized by the execution time of default \duckdb. The figure is log-scaled. The box denotes 25- to 75-percentile (with the orange line as the median), while the horizontal lines denote min and max (excluding outliers). `*' indicates timeouts. Cyclic queries are in red.}}
    \label{fig:dsb-left-b}
\end{figure*}
\section{Additional Robustness Results (Bushy)}

In Appendix C, we present the distribution of execution time with random bushy plans for each query of \RPT, compared to our baseline methods: vanilla DuckDB, Bloom Join, and \PT. These results are shown in \Cref{fig:tpch-bushy} (\tpch) and \Cref{fig:job-bushy} (\job), \Cref{fig:tpcds-bushy-a} (\tpcds query 1-52), \Cref{fig:tpcds-bushy-b} (\tpcds query 53-99), \Cref{fig:dsb-bushy-a} (\dsb query 1-52) and \Cref{fig:dsb-bushy-b} (\dsb query 53-99).

The conclusions are consistent with the left-deep results, but we observe a deterioration in robustness. As discussed in the paper, this can be attributed to incorrect probe-build side selection during the hash join.

\begin{figure*}[t!]
    \centering
    \includegraphics[width=\linewidth]{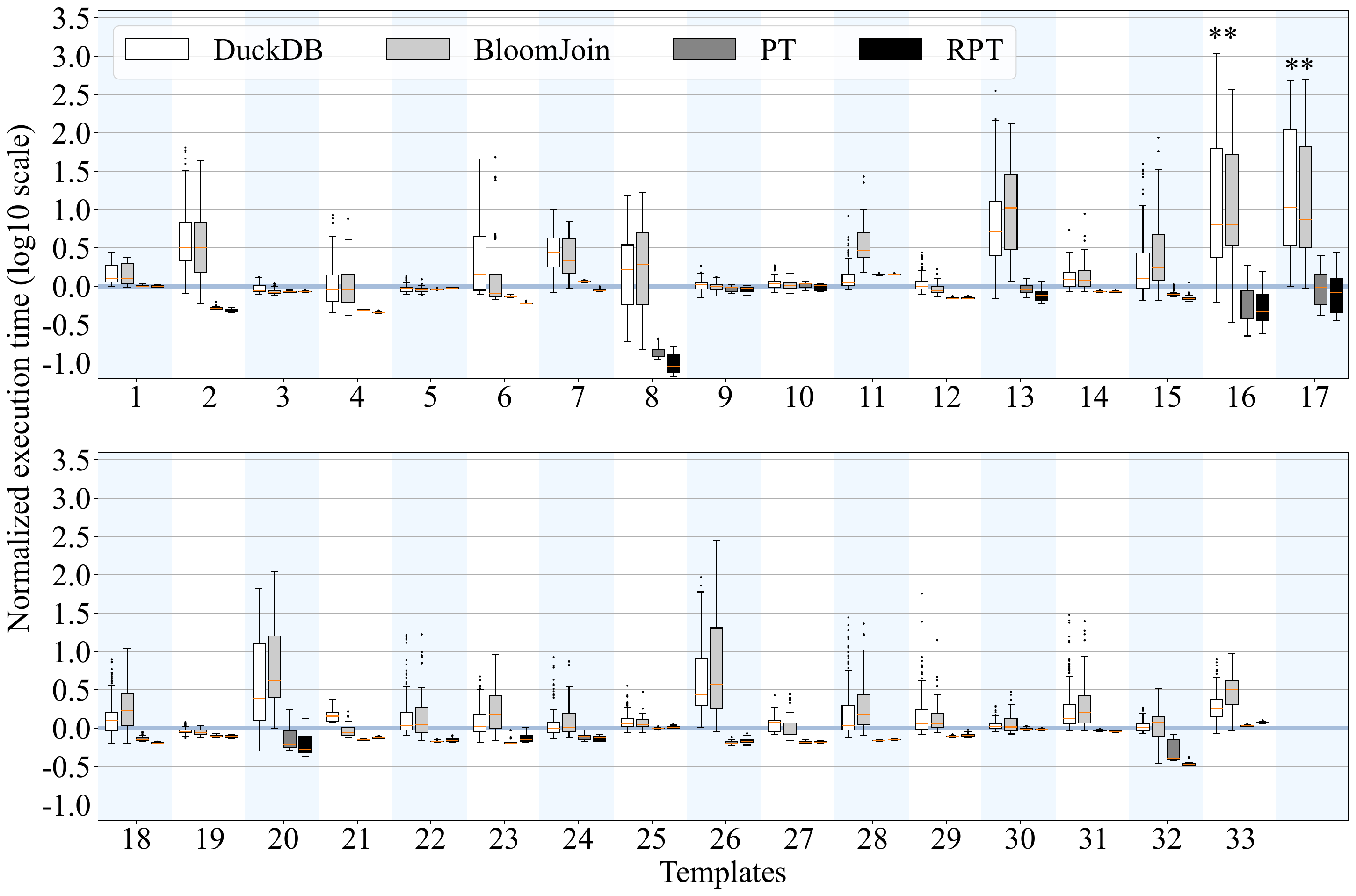}
    \caption{The distribution of execution time with random bushy plans for each query in \job \textnormal{-- Normalized by the execution time of default \duckdb. The figure is log-scaled. The box denotes 25- to 75-percentile (with the orange line as the median), while the horizontal lines denote min and max (excluding outliers).}}
    \label{fig:job-bushy}
\end{figure*}

\begin{figure*}[t!]
    \centering
    \includegraphics[width=\linewidth]{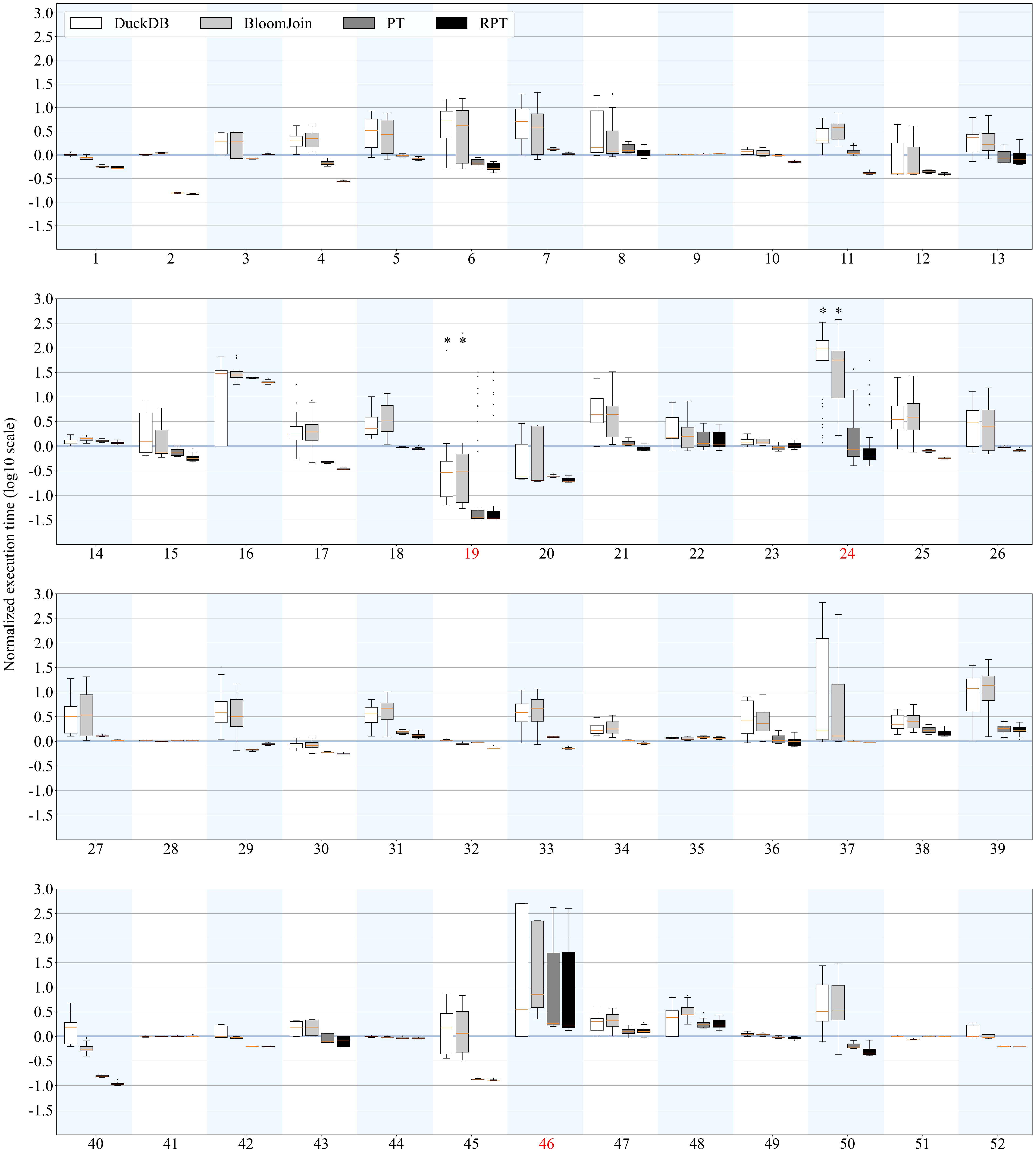}
    \caption{The distribution of execution time with random bushy plans for each query (1-52) in \tpcds \textnormal{-- Normalized by the execution time of default \duckdb. The figure is log-scaled. The box denotes 25- to 75-percentile (with the orange line as the median), while the horizontal lines denote min and max (excluding outliers). `*' indicates timeouts. Cyclic queries are in red.}}
    \label{fig:tpcds-bushy-a}
\end{figure*}

\begin{figure*}[t!]
    \centering
    \includegraphics[width=\linewidth]{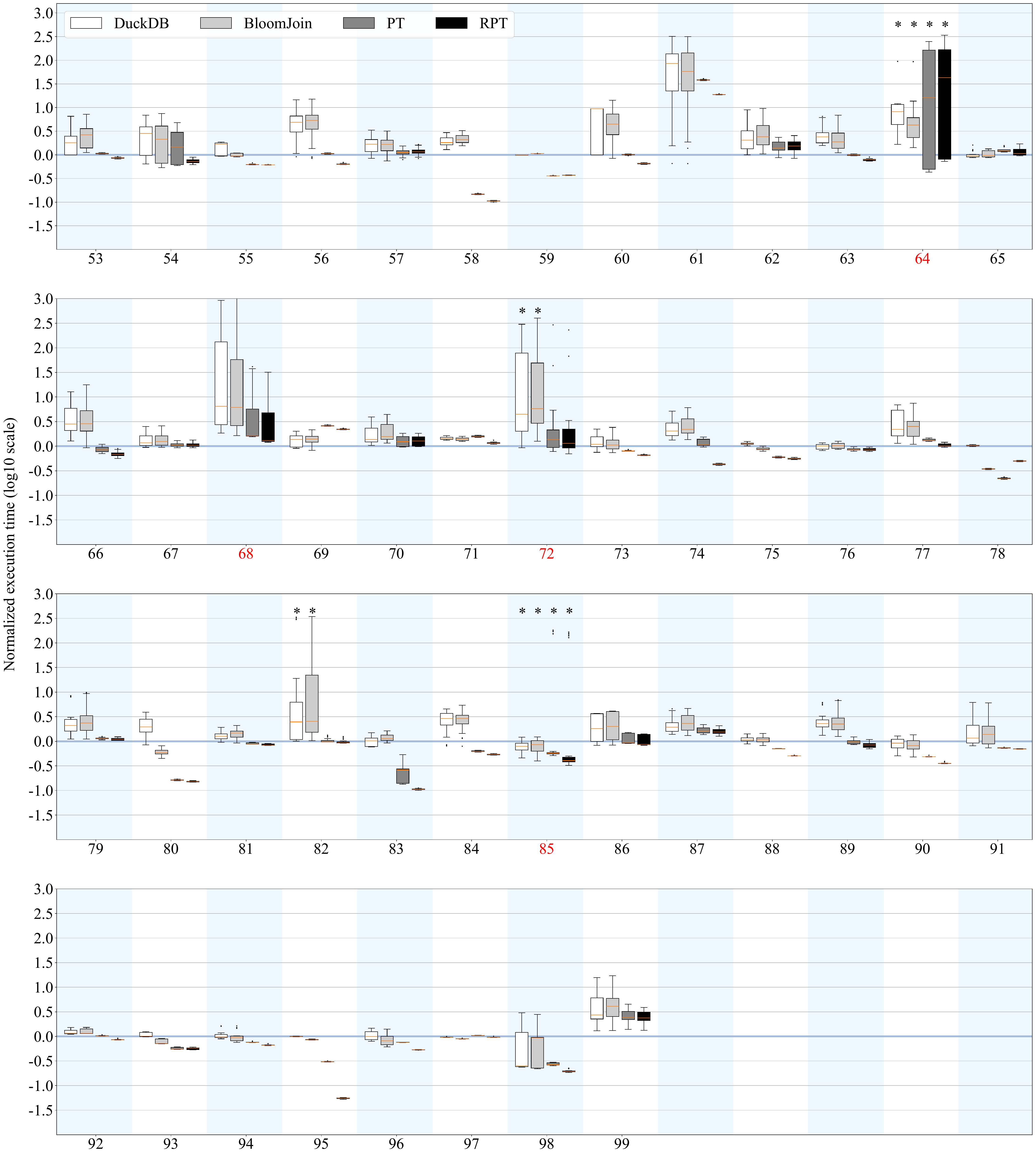}
    \caption{The distribution of execution time with random bushy plans for each query (53-99) in \tpcds \textnormal{-- Normalized by the execution time of default \duckdb. The figure is log-scaled. The box denotes 25- to 75-percentile (with the orange line as the median), while the horizontal lines denote min and max (excluding outliers). `*' indicates timeouts. Cyclic queries are in red.}}
    \label{fig:tpcds-bushy-b}
\end{figure*}

\begin{figure*}[t!]
    \centering
    \includegraphics[width=\linewidth]{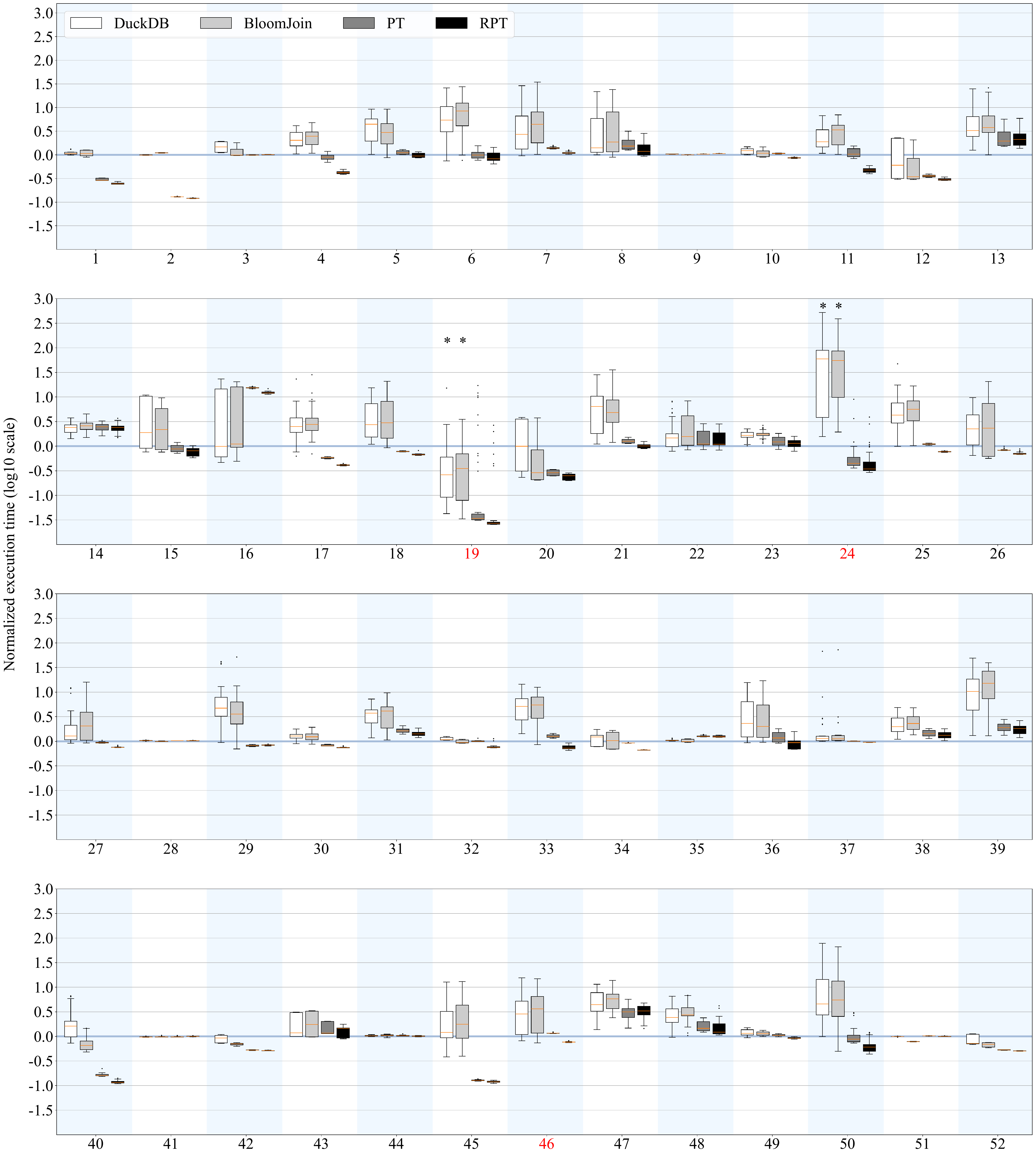}
    \caption{The distribution of execution time with random bushy plans for each query (1 - 52) in \dsb \textnormal{-- Normalized by the execution time of default \duckdb. The figure is log-scaled. The box denotes 25- to 75-percentile (with the orange line as the median), while the horizontal lines denote min and max (excluding outliers). `*' indicates timeouts. Cyclic queries are in red.}}
    \label{fig:dsb-bushy-a}
\end{figure*}

\begin{figure*}[t!]
    \centering
    \includegraphics[width=\linewidth]{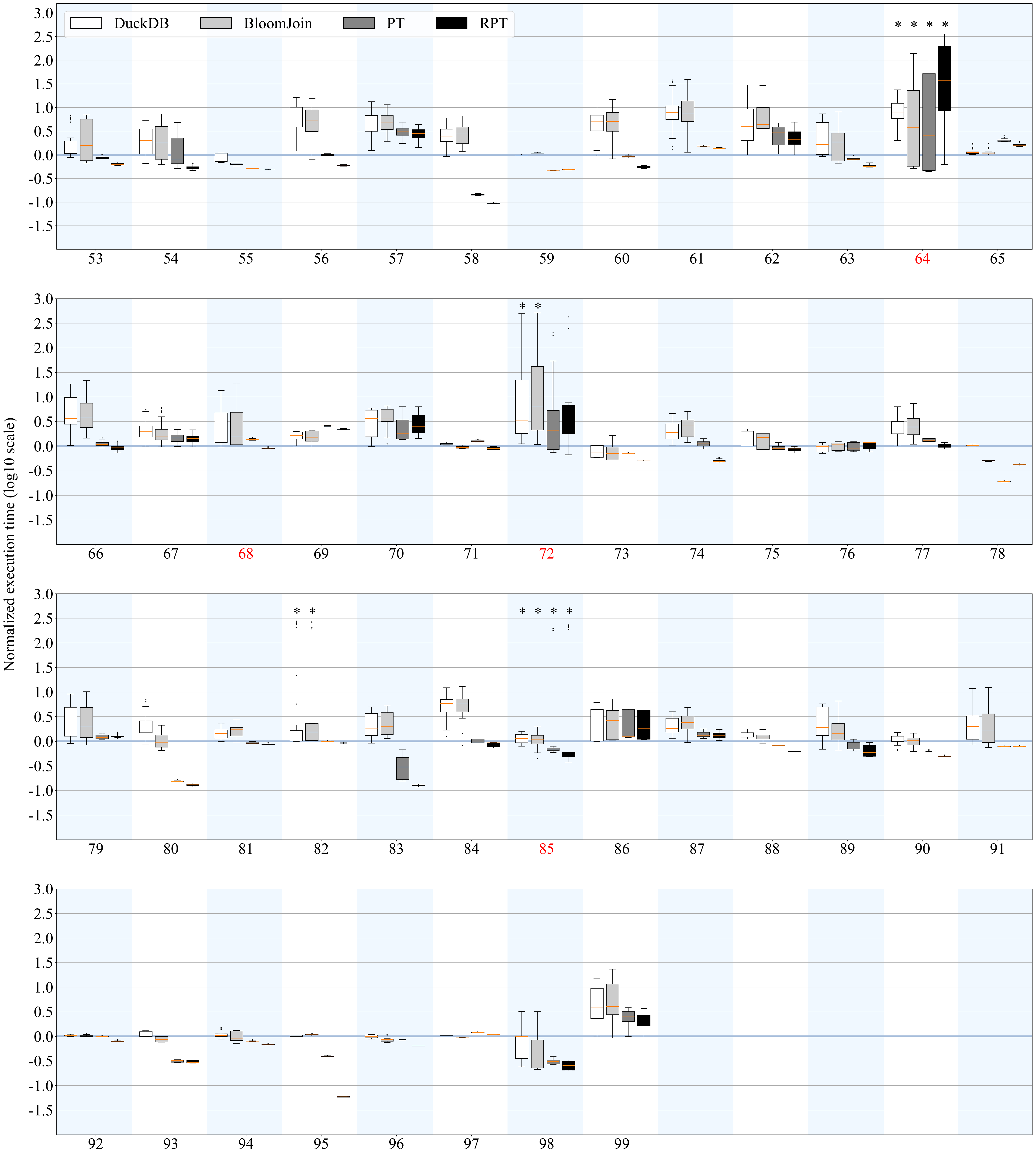}
    \caption{The distribution of execution time with random bushy plans for each query (53 - 99) in \dsb \textnormal{-- Normalized by the execution time of default \duckdb. The figure is log-scaled. The box denotes 25- to 75-percentile (with the orange line as the median), while the horizontal lines denote min and max (excluding outliers). `*' indicates timeouts. Cyclic queries are in red.}}
    \label{fig:dsb-bushy-b}
\end{figure*}

\end{document}